\documentclass{article}
\usepackage[utf8]{inputenc}
\usepackage{subcaption}
\usepackage{longtable}
\usepackage{amsmath}
\usepackage{booktabs}
\usepackage{siunitx}
\usepackage{caption}

\DeclareMathOperator*{\argmin}{arg\,min}
\usepackage{amssymb}
\usepackage{pgf, tikz}
\usetikzlibrary{arrows, automata}
\usepackage[natbibapa]{apacite}
\usepackage{hyperref}
\usetikzlibrary{tikzmark}
\usepackage{colortbl}
\usepackage{floatrow}
\newfloatcommand{capbtabbox}{table}[][\FBwidth]
\usepackage[normalem]{ulem}
\usepackage{floatrow}
\usepackage[colorinlistoftodos]{todonotes}
\usepackage{graphicx}
\usepackage{tabularx,colortbl}
\usepackage{makecell}
\usepackage{multicol}
\definecolor{LightBlue}{RGB}{140,186,252}
\newcolumntype{a}{>{\columncolor{LightRed}}c}
\usepackage{amsthm,amsmath,amssymb,mathtools}
\usepackage{algorithm}
\usepackage{enumerate}
\usepackage{enumitem}
\usepackage{bbm}
\usepackage{titlesec}
\usepackage{bigints}
\usepackage{setspace} 
\usepackage{verbatim} 
\usepackage{booktabs} 
\usepackage[width=0.9\textwidth]{caption}
\usepackage{fullpage}
\usepackage{hyperref}
\usepackage[toc,page]{appendix}
\usepackage[natbibapa]{apacite}

\def\1{\mathbbm{1}}

\usepackage{tikz}
\usepackage{algorithm}
\usepackage{algpseudocode} 
\usetikzlibrary{shapes,decorations,arrows,calc,arrows.meta,fit,positioning}
\tikzset{
    -Latex,auto,node distance =1 cm and 1 cm,semithick,
    state/.style ={ellipse, draw, minimum width = 0.7 cm},
    point/.style = {circle, draw, inner sep=0.04cm,fill,node contents={}},
    bidirected/.style={Latex-Latex,dashed},
    el/.style = {inner sep=2pt, align=left, sloped}
}
\usetikzlibrary{shapes,arrows}
\usetikzlibrary{calc}

\newtheorem{theorem}{Theorem}
\newtheorem{proposition}{Proposition}
\newtheorem{corollary}{Corollary}

\newtheorem{lemma}{Lemma}
\newtheorem{remark}{Remark}

\usepackage{xcolor}

\title{Summary-powered prediction under distribution shift}

\author{Ivy Zhang and Dominik Rothenh\"ausler}
\usepackage{fullpage}

\begin{document}

\maketitle

\begin{abstract}
   Prediction models can perform poorly when the deployment population differs from the training population. Data from the target population would help, but individual-level target data may be inaccessible because of access restrictions or reporting conventions. We consider a multi-resolution setting in which individual-level data are available from a source population, while the target population is observed only through subgroup summaries. We propose SAGE, a one-step estimator that updates a source-trained predictor using a gradient estimated from these summaries. Motivated by diagnostics consistent with the random distribution shift model, we choose SAGE's step-size to account for both sampling and distributional uncertainty. Under this model, SAGE reduces mean asymptotic target excess risk relative to the source-trained predictor. We also show that, under the model, entropy-balancing weighted empirical risk minimization (EB), which reweights source observations to match the target summaries, is asymptotically equivalent to a full-step SAGE update. SAGE with the optimal step-size has asymptotic mean squared error no larger than that of EB. Across real-world datasets, SAGE generally improves on the source-trained predictor and one-step updates that ignore distribution shift, including when the random shift model only partially captures the observed shifts. Compared to EB, SAGE improves prediction more consistently across the sample size and shift settings studied.

\end{abstract}

\section{Introduction}\label{sec: intro}


Prediction models may generalize poorly when deployed in populations that differ from those represented in the training data. Under such distribution shift, the source risk minimizer can differ from the target risk minimizer $\theta_t$, limiting predictive performance in the target population \citep{liu2023need, WILDS, tableshift}. Incorporating information from the target population can help correct this mismatch and improve predictive accuracy. Building prediction models typically requires individual-level observations, but in many applications the target population $P_t$ is observed only in aggregate form because of access restrictions, resource constraints, or reporting conventions. We study this multi-resolution setting, in which individual-level data are available from the source population $P_s$, but only subgroup summaries are observed from $P_t$, where predictions are needed. This setting raises three questions. First, how should summary information be used to improve a source-trained predictor without individual target observations? Second, how much more can be gained by accounting for distribution shift, compared with an update that assumes i.i.d.\ data? Finally, how should we quantify uncertainty about the target risk minimizer $\theta_t$?

To illustrate the problem, suppose an analyst at a retail chain wants to predict daily item demand across its locations. They have detailed records from a company-operated location, `Store 01', including item name, number of units sold, and calendar features such as date and seasonality. However, detailed sales records from independently operated franchises are unavailable for model training. The analyst has access only to aggregate reports, such as the proportion of sales made on weekends or average demand across years for each item. A predictor trained only on Store 01 is easy to obtain, but it may not transfer well to the franchise locations because locations differ in customer demand, seasonal patterns, and item-specific sales relationships. The analyst therefore wants to use these summaries to adapt demand forecasts to the franchise locations.

Before deciding how best to use the target summary information to adapt the source-trained predictor, we examine the structure of the observed source-target shifts. Consider the Store Item Demand data from Kaggle, which records daily sales for 50 items across 10 stores from January 2013 through December 2017 \citep{store_item_demand_forecasting}, with one store as the source and another as the target. For each source-target store pair, with $n_s$ source and $n_t$ target records on $L$ covariates and an outcome, we compute for each retained whitened coordinate $q=1,\ldots,L+1$ the standardized mean difference
\begin{equation}\label{eqn: qq_stat}
\left(\frac{1}{n_s} + \frac{1}{n_t}\right)^{-1/2}
\left(\frac{1}{n_s}\sum_{i=1}^{n_s} U_{i,q}^{(s)}
-
\frac{1}{n_t}\sum_{i=1}^{n_t} U_{i,q}^{(t)}\right),
\end{equation}
where, for $k\in\{s,t\}$, denotes the whitened version of
the data matrix with rows $((X_i^{(k)})^\intercal,y_i^{(k)})$, using a common whitening transformation estimated from source records so that the empirical covariance matrix of the transformed source variables is the identity. Each whitened coordinate is generally a linear combination of the original covariates and outcome. Coordinates redundant with an intercept are removed before whitening.

If the source and target observations were sampled from the same distribution, the vector obtained by stacking the statistics in Equation~\eqref{eqn: qq_stat} over $q$ would be asymptotically standard Gaussian, so its entries would fall near the reference line of a standard Gaussian QQ plot. Instead, Figure~\ref{fig: qqplot_store_intro} shows that the entries deviate from this i.i.d.\ reference line while remaining approximately linear with an inflated slope. This suggests that the source-target differences are still approximately Gaussian, but with larger variance than sampling noise alone would produce.

\begin{figure}[ht]
    \centering
    \includegraphics[width=\linewidth]{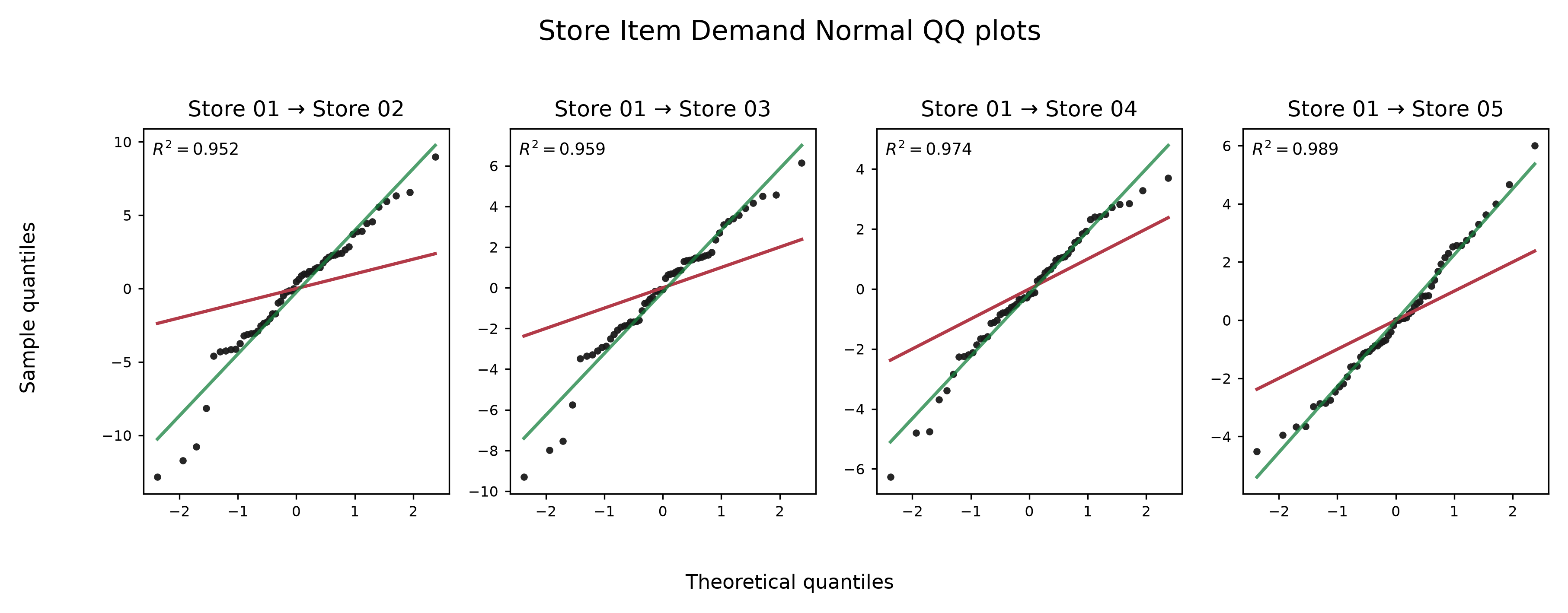}
    \caption{Normal QQ plots for the Store Item Demand dataset for Stores 02-05. Each point corresponds to the standardized mean difference statistic in Eqn.~\eqref{eqn: qq_stat} for one whitened coordinate $q$. The red line is the standard Gaussian reference line, corresponding to i.i.d.\ sampling from the same distribution. The green line represents a Gaussian pattern with inflated variance. The store pair diagnostics do not follow the i.i.d.\ reference line, but are approximately linear with inflated variance, consistent with the random shift structure discussed in detail in Section~\ref{sec: setting}. QQ plots for Stores 06-10 can be found in Appendix \ref{app: qq}.} 
    \label{fig: qqplot_store_intro}
\end{figure}

As we will see in Section~\ref{sec: setting}, this diagnostic pattern is consistent with the random shift model introduced by \citet{random_shift1}. Similar patterns have been observed in GTEx data \citep{random_shift2} and in social science replication studies \citep{random_shift3}. Many distribution shift frameworks assume a specific structure linking the source and target distributions: covariate shift keeps $P(Y\mid X)$ invariant \citep{Shimodaira}, and domain adaptation learns representations that are invariant across domains \citep{Ganin2015}. The random shift model takes a different approach by treating the density ratio between source and target as random, which allows both $P(X)$ and $P(Y\mid X)$ to change. Under the random shift model, there are two sources of variation: sampling variation and variation due to the random distribution shift. When updating a source-trained predictor with target subgroup summaries, the size of this additional variation indicates how strongly to update. We therefore use these summaries both to construct the update and to adapt its step-size to the estimated strength of the shift.

We propose SAGE (Summary Adapted Gradient Estimator), a one-step estimator which updates an initial source-trained predictor with a gradient estimated using target subgroup summaries. SAGE also incorporates a data-adaptive step-size estimate that accounts for sample size and distribution shift strength present in the data. The choice of step-size is non-trivial as it represents a trade-off between correcting for distribution shift and introducing excessive variance from potentially noisy summaries. Furthermore, without individual-level target data, the step-size cannot be tuned by target validation, and under distribution shift, source validation is not informative.


We preview SAGE's predictive performance in Figure~\ref{fig: store1}, highlighting two sources of improvement over source-only training. First, incorporating target summary data through a conservative one-step update (``One-step shift-ignoring''), with a step-size chosen under an i.i.d.\ assumption, already improves prediction. Second, choosing the step-size adaptively to account for distribution shift (SAGE) yields a greater MSE reduction. Both methods require no individual-level target data for training. For comparison, we also include an oracle that chooses the step-size using the target test set. SAGE, which does not use the test set, performs close to this oracle. We return to this example in more detail in Section~\ref{sec: store}, including a comparison with entropy-balancing weighted empirical risk minimization (EB).



\begin{figure}[ht]
    \centering
    \includegraphics[width=0.75\linewidth]{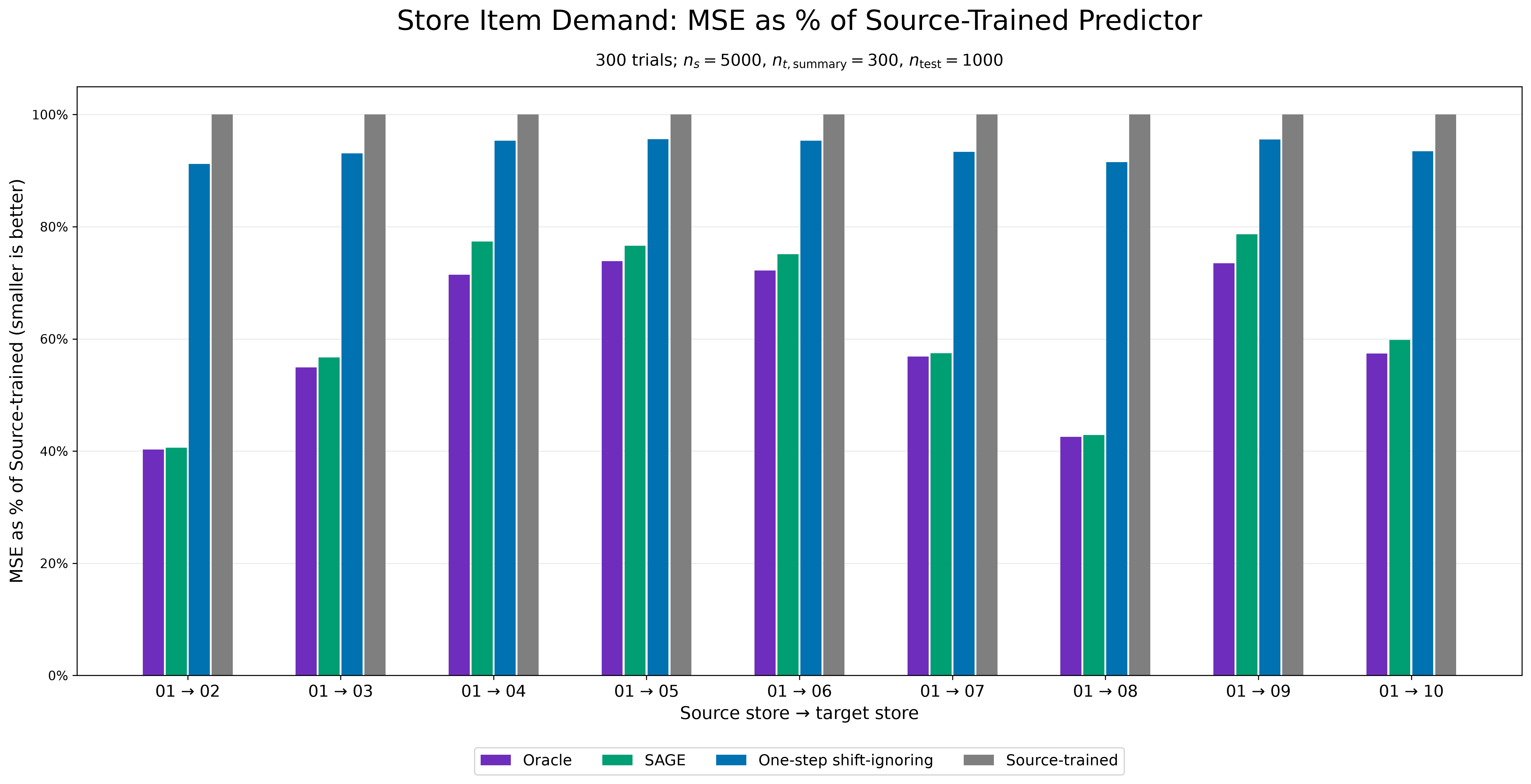}
    \caption{Store item demand prediction across stores. Average test MSE as a percentage of source-only predictor test MSE averaged over 300 trials. SAGE (green) outperforms a baseline that ignores distribution shift (blue) and is close to an oracle method (purple). Full experimental details are given in Section~\ref{sec: store}.}
    \label{fig: store1}
\end{figure}

\paragraph{Contribution.}
We summarize our contributions to prediction and uncertainty quantification using individual-level source data and target subgroup summaries, without requiring individual target observations.

\begin{enumerate}
\item \textbf{Improving prediction using target subgroup summaries.} We propose SAGE, a one-step estimator that incorporates target subgroup summaries to improve a source-trained predictor. Among estimators that adjust the source-trained predictor by a linear function of the source-target summary differences, SAGE with the optimal step-size has the smallest asymptotic mean squared-error.

\item \textbf{Accounting for distribution shift.} Under the random shift model, we derive the asymptotically optimal step-size $\gamma^*$ and show that SAGE reduces mean asymptotic target excess risk relative to the source-trained predictor for any fixed $\gamma \in (0,2\gamma^*)$. We also propose an estimator $\hat\gamma$ that adapts the update to the estimated shift strength and summary sample size. Its accuracy depends on the summary dimension $r$.
\item \textbf{Connection to entropy balancing.} We show that, under the random shift model, EB is asymptotically equivalent to a full-step SAGE update ($\gamma=1$). SAGE with the optimal step-size therefore has asymptotic mean squared-error no larger than that of EB.
\item \textbf{Uncertainty quantification.} We construct plug-in confidence intervals for $\theta_t$ that account for both sampling variability and residual distributional uncertainty under the random shift model. Their asymptotic validity requires a growing summary dimension $r$.

\item \textbf{Empirical evaluation.} Across simulations and real data examples, we demonstrate gains over a source-trained predictor by incorporating target subgroup summaries, with further improvements from accounting for distribution shift. These gains appear both when diagnostics support the random shift model and when the observed shifts are only partially explained by it. We also compare SAGE with EB empirically.

\item \textbf{Extensions.} 
We extend the SAGE framework to settings where target subgroup variances are available and adapt the method to neural networks through SAGE-NN.

\end{enumerate}

\paragraph{Paper outline.}
Section~\ref{sec: setting} formalizes the problem setting and provides background on the random shift model. Section~\ref{sec: related work} reviews related work. Section~\ref{sec: sage} introduces SAGE and establishes its theoretical properties. Section~\ref{sec: sim} examines the theoretical results through simulations. Sections~\ref{sec: store} and \ref{sec: Nhanes} present real-world applications, comparing SAGE with a source-trained predictor, other one-step estimator baselines, and EB. Section \ref{sec: medicare} discusses a neural network adaption of SAGE and presents an empirical application. Section~\ref{sec: disc} concludes with a discussion.

\section{Setting}\label{sec: setting}

We consider a source distribution $P_s$ and a target distribution $P_t$. From the source distribution, we observe individual-level data $D_i^{(s)}=(X_i^{(s)}, y_i^{(s)}, A_i^{(s)})$, $i=1,\ldots,n_s$, where $X_i^{(s)}=(X_{1,i}^{(s)},\ldots,X_{L,i}^{(s)})$ denotes covariates, $y_i^{(s)}$ denotes the outcome, and $A_i^{(s)}=(A_{i1}^{(s)},\ldots,A_{iK}^{(s)})\in\{0,1\}^K$ records subgroup membership. Here, $A$ may encode demographic groups, study arms, or other reporting categories. A unit $i$ may have $A_{ig}=1$ for more than one $g$, allowing overlapping subgroups, and the reported subgroups need not form a complete partition of either population.

From the target distribution, we observe only subgroup counts and means; individual-level target records are unavailable. We do not require target summaries for every covariate, subgroup, or outcome present in the source dataset. We assume that all available target subgroup summaries are computed from a common underlying i.i.d.\ sample $\{D_i^{(t)}\}_{i=1}^{n_t}$ from $P_t$. Consequently, summaries for different potentially overlapping subgroups may be statistically dependent.

Our goal is to improve prediction under $P_t$ using individual-level source data together with these target subgroup summaries. If individual-level target observations were available, we could directly estimate the target risk minimizer $\theta_t = \argmin_{\theta\in\Theta} E_t[\ell(\theta; X, y)]$ for a loss function $\ell$ and parameter space $\Theta\subset\mathbb R^p$. Without them, $\theta_t$ cannot be estimated by direct target empirical risk minimization. We therefore begin with an initial source-trained predictor, $\hat\theta_{\mathrm{init}}$, and ask how target subgroup summaries can be used to adapt it toward the target population and reduce target risk by accounting for distribution shift.

\subsection{Random distribution shift model}
The diagnostic in Figure~\ref{fig: qqplot_store_intro} suggests a model of distribution shift under which source-target differences are approximately Gaussian with inflated variance, and which allows both $P(X)$ and $P(Y\mid X)$ to change. The random shift model of \citet{random_shift1} has this property.

The random shift model treats the source distribution $P_s$ as a random reweighting of the fixed target distribution $P_t$ over the joint distribution of $X$ and $Y$, thereby capturing simultaneous shifts in $P(X)$ and $P(Y\mid X)$. Under this model, different parts of the data space are randomly assigned higher or lower probability mass compared to $P_t$. Formally, let the data space be partitioned into regions $(I_j)_{j=1,\ldots,m}$ of equal probability under $P_t$, each assigned a random weight $W_j^{(s)}$ that determines how much probability mass that region receives under $P_s$, that is, with Radon--Nikodym derivative
\begin{equation*}
    \frac{dP_s}{dP_t}(D) = \frac{W_j^{(s)}}{\frac{1}{m} \sum_{j'=1}^m W_{j'}^{(s)}} \qquad \text{ for } D \in I_j.
\end{equation*}
The random weights $W_j^{(s)}$, $j=1,\ldots,m$, are i.i.d.\ across $j$, positive and satisfy $E[W_j^{(s)}]=1$ and $\operatorname{Var}(W^{(s)})<\infty$. The target sample is independent of the source weights and, conditional on the weights, independent of the source sample. Technical construction details, including regularity assumptions, are provided in Appendix \ref{app: random_shift_construction}. Adapting Lemma 1 of \cite{duc} to a single-source, fixed-target regime for vector-valued test functions $f$, we present a distributional CLT that underlies our subsequent theoretical analysis. 

\begin{lemma}
    [Distributional CLT]\label{lemma: distr_clt}

Let $f_t: \mathcal{D} \rightarrow \mathbbm{R}^{p_t}, f_s: \mathcal{D} \rightarrow \mathbbm{R}^{p_s}$ be any Borel measurable functions with finite second moments under $P_{t}$. Consider the setting where distributional and sampling uncertainty are of the same order, i.e., $m/n_k \rightarrow c_k > 0$ for each $k \in \{s, t\}$. Let $D$ denote a draw from the underlying population distribution $P_t$ on the data space $\mathcal D$. Then 
$$
\sqrt{m}
\begin{pmatrix}

\frac{1}{n_t} \sum_{i=1}^{n_t} f_{t}(D^{(t)}_i) 
-
E_t[f_{t}(D)]\\
\frac{1}{n_s} \sum_{i=1}^{n_s} f_{s}(D^{(s)}_i) 
-
E_t[f_{s}(D)]
\end{pmatrix}
\overset{d}{\rightarrow} N\big(0, \Omega \big),
$$
where $\Omega$ is the asymptotic covariance matrix,
\begin{equation*}
    \Omega \coloneqq
        \begin{pmatrix}
c_t\operatorname{Var}_t(f_t(D))
&
0_{p_t\times p_s}
\\
0_{p_s\times p_t}
&
\left(c_s+\operatorname{Var}(W^{(s)})\right)\operatorname{Var}_t(f_s(D)).
\end{pmatrix}
\end{equation*}
\end{lemma}

Lemma~\ref{lemma: distr_clt} separates two sources of variation. The constants $c_k=\lim m/n_k$ are inversely proportional to the sample sizes $n_k$ and capture sampling variation, while $\operatorname{Var}(W^{(s)})$ captures distributional variation, that is, the strength of the shift. In finite samples, a source sample mean therefore has variance approximately $(1/n_s+\delta^2_{\rm dist})\operatorname{Var}_t(f_s(D))$, where $\delta^2_{\rm dist}\coloneqq\operatorname{Var}(W^{(s)})/m$, while a target sample mean has variance approximately $\operatorname{Var}_t(f_t(D))/n_t$. The number of regions $m$ thus sets the scale of the distributional variation relative to the sampling variation.

Returning to the diagnostic plot of Figure \ref{fig: qqplot_store_intro}, we can view Store 01 as a perturbed distribution from the target store. Then under a random distribution shift (Lemma \ref{lemma: distr_clt}), we would expect the points from Figure~\ref{fig: qqplot_store_intro}, represented by the statistic in Equation~\eqref{eqn: qq_stat}, to follow an approximately Gaussian distribution with variance inflated by the factor $(c_s+c_t+\operatorname{Var}(W^{(s)}))/(c_s+c_t)$. 

\section{Related Work}\label{sec: related work}

\paragraph{Distribution shift and domain adaptation}
A large literature studies prediction under covariate shift. This motivates importance weighting methods for shift correction as in \cite{Shimodaira, sugiyama2007covariate}. Other methods such as distributionally robust optimization (DRO) take a worst-case shift approach \citep{benTal2013robust, wiesemann2014distributionally, namkoongDuchi2016stochastic, dro_matching}. These methods do not exploit the distributional shift patterns observed in Figure~\ref{fig: qqplot_store_intro}. In addition, importance weighting methods that estimate density ratios
from individual target covariates also require access to those covariates.

In domain adaptation, CORAL aligns source and target covariance matrices using target moments \citep{coral}; unlike SAGE, this targets unsupervised domain adaptation. Others that leverage target moments computed from unlabeled target data include \cite{li2017adabn}, \cite{NEURIPS2020_85690f81}, which replace source activation moments with target ones; \cite{NIPS2006_a2186aa7}, which reweights source observations
to match the weighted source kernel mean
embedding to its target counterpart; and  \cite{pmlr-v139-taskesen21a}, which takes a DRO approach using target mean and covariance (but requires individual target data). 

Our work uses the random shift model introduced by \citet{random_shift1}, and empirically investigated by \cite{random_shift3}, to develop a framework for prediction with individual source and summary target data. \citet{random_shift2} choose weights for combining risks from multiple source distributions. SAGE instead updates a single source-trained predictor using target subgroup summaries.

\paragraph{Combining summary data with individual data}
Calibration and double-sampling methods provide classical approaches to improving estimation efficiency with auxiliary information \citep{deville1992calibration,chen2000unified}. Related work combines individual data with external summaries for disease-risk modeling \citep{ecological,hieracheal}; regression estimation using summary statistics, including fitted-model summaries \citep{chatterjee2016constrained,gim,Deng2024}; and high-resolution estimation using unbiased aggregate summaries \citep{datafusion}. For time-to-event outcomes, \citet{Zheng02102022} combine source individual data with target survival and risk-factor summaries under invariant relative risks.

Beyond these applications, \citet{hu2026fusion} derive semiparametric efficiency bounds and efficient estimators for combining individual data with noisy external summaries under transportability, and develop adaptive fusion when transportability fails. \citet{estes2018empirical} use data-adaptive empirical Bayes shrinkage to balance efficiency gains from external model summaries against bias from violations of transportability.

Another approach incorporates target summaries through balancing weights. Entropy balancing \citep{entropy_bal_Hainmueller} minimizes KL divergence from baseline weights subject to moment constraints. \citet{Josey2021,entropy_bal} apply this approach to transport treatment effects using target covariate summaries, while \citet{transportable_inf} study entropy balancing for target outcome means under covariate shift. SAGE incorporates summaries through a one-step correction and develops prediction and uncertainty quantification under a random shift model that allows both covariate and conditional outcome distributions to change.

\paragraph{Prediction-powered inference}
Prediction-powered inference (PPI) combines a small labeled
sample with predictions on a larger unlabeled sample to
improve estimation and inference \citep{angelopoulos2023prediction}.
PPI++ adaptively weights a prediction-based correction to
reduce asymptotic variance \citep{angelopoulos2023ppi},
a principle related to SAGE's adaptive weighting of its
summary-based gradient correction.
Standard PPI and PPI++ use individual labeled and unlabeled
observations, whereas SAGE uses individual source data and
target subgroup summaries.
Versions of PPI accommodate covariate, domain, and label
shift \citep{angelopoulos2023prediction,
zou2026generalizedpredictionpoweredinferenceapplication,
li2026multisourcepredictionpoweredinference}.
SAGE instead models source-target discrepancies through
random distributional perturbations.

\paragraph{Fine-tuning and one-step corrections}
Fine-tuning adapts a previously trained model by optimizing on data from a new but related task or population; see the section \emph{Fine-Tuning} in \citet{zhang2023dive}, Ch.~19 of \cite{murphy2022probabilistic}, and surveys such as \cite{jiang2022transferability,zhuang2021comprehensive}. Our work shares the goal of adapting a source-trained predictor to a target population, but differs in the information available. We consider the setting where there is no target sample on which to directly fine-tune the model. SAGE instead uses an update estimated from target subgroup summaries, and provides risk reduction guarantees and asymptotic confidence intervals under the random shift model.

The correction step in SAGE is also related to classical one-step estimators, which improve an initial estimator by taking a single Newton step \citep{vanderVaart1998asymptotic, lehmannCasella1998point}. One-step and targeted learning estimators are widely used in semiparametric inference to obtain asymptotic linearity, efficiency, and confidence intervals under i.i.d.\ or covariate shift settings \citep{bickel1993efficient, vanderLaanRubin2006targeted, vanderLaanRose2011targeted, qiu2023prediction}. Unlike these classical settings, SAGE does not require individual-level target data to construct its correction step.

\section{Gradient updates with subgroup summaries}\label{sec: sage}
We now describe SAGE, which uses target subgroup summaries to improve a predictor trained on source data. We start from the source-trained empirical risk minimizer $\hat\theta_\mathrm{init} \coloneqq \argmin_\theta \frac{1}{n_s}\sum_{i=1}^{n_s} \ell(\theta; X^{(s)}_i, y_i^{(s)})$. The goal of SAGE is to update $\hat\theta_\mathrm{init}$ in the direction of the target risk minimizer, $\theta_t \coloneqq \argmin_\theta E_t[\ell(\theta;X,y)]$. A standard one-step Newton update has the form
\begin{equation*}
    \hat\theta(\gamma)
    =
    \hat{\theta}_{\mathrm{init}}
    -
    \gamma {H}_t^{-1}(\hat{\theta}_{\mathrm{init}}){\Psi}_t(\hat{\theta}_{\mathrm{init}}),
\end{equation*}
for some step size $\gamma$. Here, $\Psi_t(\theta)\coloneqq E_t[\nabla_\theta\ell(\theta;X,y)]$
and $H_t(\theta)\coloneqq E_t[\nabla_\theta^2\ell(\theta;X,y)]$
denote the gradient and Hessian of the target risk at
$\theta$, respectively. If individual-level target data were available, a natural estimate of the target gradient would be
\begin{equation*}
    \widehat\Psi_t^{\mathrm{ind}}(\hat\theta_{\mathrm{init}}) \coloneqq
    \frac{1}{n_t}\sum_{i=1}^{n_t} \nabla_\theta \ell(\hat\theta_{\mathrm{init}};X^{(t)}_i,y^{(t)}_i).
\end{equation*}
Without individual target observations, $\widehat\Psi_t^{\mathrm{ind}}$ cannot be computed. Instead, SAGE approximates it in two steps. First, since $\hat\theta_{\mathrm{init}}$ minimizes the source empirical risk, the average source gradient at $\hat\theta_{\mathrm{init}}$ is zero, so the target gradient equals the change in the average gradient from source to target. Second, a least-squares regression of the source gradients on subgroup features approximates the gradient as a linear function of $A_g$, $A_gX$, and $A_gy$. The change in the average gradient is then approximately the regression coefficients times the change in the average subgroup features, which can be computed from subgroup summaries. Concretely, we form the summary differences
\begin{align*}
    \Delta_0(g)
    &\coloneqq
    \hat\pi_{t,g}-\hat\pi_{s,g},\\
    \Delta_X(g)
    &\coloneqq
    \hat\pi_{t,g}\hat E_t[X\mid A_g=1]
    -
    \hat\pi_{s,g}\hat E_s[X\mid A_g=1],\\
    \Delta_y(g)
    &\coloneqq
    \hat\pi_{t,g}\hat E_t[y\mid A_g=1]
    -
    \hat\pi_{s,g}\hat E_s[y\mid A_g=1],
\end{align*}
where $A_g\in\{0,1\}$ denotes membership in reported subgroup $g$,
$\hat\pi_{s,g}
\coloneqq
\frac{1}{n_s}\sum_{i=1}^{n_s}A_{ig}^{(s)}$, and $\hat\pi_{t,g}
\coloneqq
\frac{n_t(g)}{n_t}$, where $n_t(g)$ is the reported number of target observations in subgroup $g$. Throughout, when a source or target subgroup is empty,
its proportion-weighted moments
are defined as zero.

We then form the contribution of subgroup $g$ and aggregate these contributions across subgroups:
\begin{align*}
    \widehat{\Psi}_t(g)
    \coloneqq
    \hat{\beta}_{g,0}^{(s)}\Delta_0(g)
    +
    \left(\hat\beta_{g,X}^{(s)}\right)^\intercal\Delta_X(g)
    +
    \hat{\beta}_{g,y}^{(s)}\Delta_y(g), \qquad \widehat{\Psi}_t \coloneqq \sum_{g=1}^K \widehat{\Psi}_t(g).
\end{align*}


The estimated coefficients are obtained from the joint least-squares projection of the source gradients onto all reported subgroup features. 
\begin{align*}
(\hat\alpha^{(s)},
\{\hat\beta_{g,0}^{(s)},\hat\beta_{g,X}^{(s)},\hat\beta_{g,y}^{(s)}\}_{g=1}^K)
\in
\argmin_{\alpha,\{\beta_g\}_{g=1}^K}
\sum_{i=1}^{n_s}
\left\|
\nabla_\theta\ell(\hat{\theta}_{\mathrm{init}};X_i^{(s)},y_i^{(s)})
-
\alpha
-
\sum_{g=1}^K
A_{ig}^{(s)}
\left(
\beta_{g,0}
+
\beta_{g,X}^\intercal X_i^{(s)}
+
\beta_{g,y}y_i^{(s)}
\right)
\right\|_2^2 .
\end{align*}

We estimate the Hessian from individual-level source data
$$\widehat{H}_s \coloneqq \frac{1}{n_s}\sum_{i=1}^{n_s}\nabla_\theta^2\ell(\hat{\theta}_{\mathrm{init}};X_i^{(s)},y_i^{(s)}).$$
Combining the summary-adapted gradient estimate with the source Hessian gives the SAGE one-step update
\begin{equation*}
    \hat\theta_\mathrm{SAGE}
    =
    \hat\theta_\mathrm{init}
    -\gamma \widehat{H}_s^{-1}\widehat{\Psi}_t ,
\end{equation*}
where $\gamma\geq0$ is a step-size controlling how aggressively the source-trained estimator is adapted toward the target distribution. In the following subsections, we derive an asymptotically optimal choice $\gamma^*$ under the random shift model and show how it can be estimated using only subgroup summaries. We summarize the SAGE method in Algorithm \ref{alg: grad_linear}.

\begin{remark}[Incorporating subgroup variances]
When subgroup variances are reported, SAGE can augment the subgroup
features with $X_1^2,\ldots,X_L^2,y^2$
(Algorithm~\ref{alg: sage var}) to capture quadratic terms and
variance shifts beyond subgroup means. Redundant coordinates are
then removed using the same fixed selection for source and target.
For example, $A_gX_j^2$ is omitted when $X_j\in\{0,1\}$, since
$A_gX_j^2=A_gX_j$; likewise for $A_gy^2$ when $y$ is binary.
The subsequent theory extends to the retained augmented
features, with the corresponding adjustment class and
influence functions, provided the respective assumptions
of each result hold for these features.
\end{remark}

\begin{algorithm}[ht]
\caption{SAGE: Summary Adapted Gradient Estimator}
\label{alg: grad_linear}
\begin{algorithmic}[1]
\Require Source data
$\{(X_i^{(s)},y_i^{(s)},A_{i1}^{(s)},\ldots,A_{iK}^{(s)})\}_{i=1}^{n_s}$,
where $A_{ig}^{(s)}\in\{0,1\}$ indicates membership in subgroup $g$
\Require Target summaries
$\{\hat E_t[X\mid A_g=1],\hat E_t[y\mid A_g=1]\}_{g:n_t(g)>0}$
\Require Target subgroup counts $\{n_t(g)\}_{g=1}^K$ and total target summary sample size $n_t$
\Require Loss $\ell(\theta;X,y)$, initial source estimator $\hat{\theta}_{\mathrm{init}}$, fixed step-size $\gamma$
\Ensure Updated estimator $\hat{\theta}_{\mathrm{SAGE}}(\gamma)$

\State \label{alg: step_hess} Compute source Hessian: 
$\widehat H_s
\gets
\frac{1}{n_s}\sum_{i=1}^{n_s}
\nabla_\theta^2\ell(\hat{\theta}_{\mathrm{init}};X_i^{(s)},y_i^{(s)})$

\State Compute source gradients:
$u_i
\gets
\nabla_\theta\ell(\hat{\theta}_{\mathrm{init}};X_i^{(s)},y_i^{(s)})$,
for $i=1,\ldots,n_s$

\State Compute subgroup proportions:
$\hat\pi_{s,g}\gets \frac{1}{n_s}\sum_{i=1}^{n_s}A_{ig}^{(s)}$,
\quad
$\hat\pi_{t,g}\gets \frac{n_t(g)}{n_t}$
\State Compute source subgroup moments for nonempty subgroups:
$\hat E_s[X\mid A_g=1]\gets \frac{\sum_{i=1}^{n_s}A_{ig}^{(s)}X_i^{(s)}}{\sum_{i=1}^{n_s}A_{ig}^{(s)}}$,
\quad
$\hat E_s[y\mid A_g=1]\gets \frac{\sum_{i=1}^{n_s}A_{ig}^{(s)}y_i^{(s)}}{\sum_{i=1}^{n_s}A_{ig}^{(s)}}$
\State Form source-target summary differences:
\Statex \hspace{0.6cm}
$\Delta_0(g)
\gets
\hat\pi_{t,g}-\hat\pi_{s,g}$
\Statex \hspace{0.6cm}
$\Delta_X(g)
\gets
\hat\pi_{t,g}\hat E_t[X\mid A_g=1]
-
\hat\pi_{s,g}\hat E_s[X\mid A_g=1]$
\Statex \hspace{0.6cm}
$\Delta_y(g)
\gets
\hat\pi_{t,g}\hat E_t[y\mid A_g=1]
-
\hat\pi_{s,g}\hat E_s[y\mid A_g=1]$

\Statex \emph{In the regression below, apply the coordinate
reduction of Assumption~\eqref{assump: var(z) pd}, fixing
coefficients of omitted features at zero.}

\State \label{alg: ols_step}Fit on source using all reported subgroups:
\Statex \hspace{0.6cm}
$(\hat\alpha^{(s)},
\{\hat\beta_{g,0}^{(s)},\hat\beta_{g,X}^{(s)},\hat\beta_{g,y}^{(s)}\}_{g=1}^K)
\in
\arg\min_{\alpha,\{\beta_g\}_{g=1}^K}
\sum_{i=1}^{n_s}
\left\|
u_i
-
\alpha
-
\sum_{g=1}^K
A_{ig}^{(s)}
\left(
\beta_{g,0}
+
\beta_{g,X}^\intercal X_i^{(s)}
+
\beta_{g,y}y_i^{(s)}
\right)
\right\|_2^2$

\For{$g=1,\dots,K$}
    \State Estimate subgroup contribution using target summaries:
    \Statex \hspace{0.6cm}
    $\widehat{\Psi}_t(g)
    \gets
    \hat\beta_{g,0}^{(s)}\Delta_0(g)
    +
    \bigg(\hat\beta_{g,X}^{(s)}\bigg)^\intercal \Delta_X(g)
    +
    \hat\beta_{g,y}^{(s)}\Delta_y(g)$
\EndFor

\State Aggregate gradient estimates:
$\widehat{\Psi}_t
\gets
\sum_{g=1}^K\widehat{\Psi}_t(g)$

\State One-step update:
$\hat{\theta}_{\mathrm{SAGE}}(\gamma)
\gets
\hat{\theta}_{\mathrm{init}}
-
\gamma\widehat H_s^{-1}\widehat{\Psi}_t$

\State \Return $\hat{\theta}_{\mathrm{SAGE}}(\gamma)$
\end{algorithmic}
\end{algorithm}

\subsection{SAGE reduces target risk for a range of step-sizes}\label{sec: risk}

\begin{theorem}[SAGE reduces the average asymptotic excess risk compared to $\hat{\theta}_\mathrm{init}$]\label{thm: sage_risk} Consider the setting of the random shift model and the regularity conditions as listed in Appendix \ref{app: assumptions}.  Then for a fixed $\gamma \geq 0$, the scaled excess risk satisfies
\begin{equation*}
     m\left(
    E_t[\ell(\hat{\theta}_\text{SAGE}(\gamma);X,y)]
    -
    E_t[\ell(\theta_t;X,y)]
    \right)
    \overset{d}{\rightarrow}
    J(\gamma),
\end{equation*}
for a random variable $J(\gamma)$ with finite mean. The mean of
the limiting scaled excess risk variable, $E[J(\gamma)]$, is minimized at $$\gamma^* = \frac{c_s+\operatorname{Var}(W^{(s)})}{c_s+\operatorname{Var}(W^{(s)}) +c_t} \in (0,1)$$ and SAGE reduces risk compared to not incorporating target summary data, that is,
   \begin{equation*}
    E[J(\gamma)]< E[J(0)], 
   \end{equation*}
for all $\gamma \in (0, 2\gamma^*)$, where $E[J(0)]$ is the mean limiting excess risk of $\hat{\theta}_\mathrm{init}$.
\end{theorem}

The proof can be found in Appendix \ref{app: thm1}.
The optimal step-size $\gamma^*$ also reflects the tradeoff between relying on the individual-level source data and on the target summaries. Consider the following cases:
\begin{enumerate}
    \item \textbf{High source uncertainty.} Recall that $c_s$ is inversely proportional to the source sample size $n_s$, while $\operatorname{Var}(W^{(s)})$ captures distributional variation in $P_s$. When $c_s+\operatorname{Var}(W^{(s)})$ is large compared to $c_t$, the source-trained predictor is less reliable for the target population, so SAGE moves further toward the summary-based correction.
    \item \textbf{Accurate target summaries.} $c_t$ captures sampling uncertainty in the target subgroup summaries and is inversely proportional to the target summary sample size $n_t$. When $c_t$ is small compared to $c_s+\operatorname{Var}(W^{(s)})$, that is, when the target summaries are precise relative to the source sampling and distributional variation, $\gamma^*\rightarrow 1$. In this case, the summary-based correction is reliable, and SAGE approaches the full one-step update.
    \item \textbf{No distribution shift.} If $\operatorname{Var}(W^{(s)})=0$, then
$\gamma^*=c_s/(c_s+c_t)$, with finite-sample analogue
$n_t/(n_s+n_t)$. As we would expect, in the standard i.i.d.\ case, the optimal shrinkage agrees with the inverse variance weighting underlying pooled estimation: when the target summaries are more precise, SAGE moves more toward the target correction, while when the source sample is much larger, SAGE remains closer to the source-trained estimator. The i.i.d.\ step-size converges to
$\gamma_0=c_s/(c_s+c_t)\in(0,2\gamma^*)$,
which is smaller than $\gamma^*$ under distribution shift. By linearity of the update and Slutsky's theorem,
its limiting risk equals that at $\gamma_0$. An update that ignores distribution shift therefore takes too small a step, but it still reduces risk by Theorem~\ref{thm: sage_risk}. In this sense, it is a conservative one-step estimator.
\end{enumerate}  

\subsection{SAGE at \texorpdfstring{$\gamma^*$}{the optimal step-size} is optimal among linear adjustments}\label{sec: lin_opt}
The previous result shows that SAGE improves over the source-trained predictor for a range of step-sizes. We next ask whether a different use of the same subgroup summaries could yield a better linear adjustment. To formalize this question, we consider a broad class of estimators that start from $\hat\theta_\mathrm{init}$ and add an arbitrary linear function of the observed source-target differences in subgroup summaries. The following theorem shows that SAGE, using the optimal step-size $\gamma^*$, is asymptotically optimal within this class. The proof can be found in Appendix \ref{app: thm2}.

\begin{theorem}[Optimal linear adjustment estimator]\label{thm: var_opt}  
For each $k\in\{s,t\}$ and $g=1,\ldots,K$, define
$$
\hat d_{k,g}\coloneqq
\hat E_k\left[
A_g
\begin{pmatrix}
1, X, y
\end{pmatrix}^\intercal
\right],
$$
and let $\hat d_k\in\mathbb R^r$ denote the vector obtained by
stacking $\hat d_{k,g}$ over $g=1,\ldots,K$ and removing redundant
coordinates as in Assumption~\eqref{assump: var(z) pd}.
Here, the same fixed coordinate selection is used for source
and target, and $r$ denotes the number of retained coordinates. Consider the class of linearly adjusted estimators using subgroup summary information, 
\begin{equation*}
    \mathcal{A} \coloneqq \{\hat\theta_{B_m} = \hat\theta_\mathrm{init} +  B_m(\hat{d}_{t}-\hat{d}_{s}): \ B_m \in \mathbb{R}^{p\times r},  \quad B_m \overset{p}{\rightarrow} B \text{ for some deterministic B}\}
\end{equation*}
Let $\Sigma(\hat\theta)$ denote
the covariance matrix of the limiting distribution of
$\sqrt m(\hat\theta-\theta_t)$. Under the conditions of Theorem~\ref{thm: sage_risk}, the SAGE estimator, $\hat\theta_\text{SAGE}(\gamma^*)$, is asymptotically optimal in $\mathcal{A}$, in the sense that for every $\hat{\theta}_{B_m}\in \mathcal{A}$, $\Sigma(\hat\theta_{B_m})-
\Sigma(\hat\theta_\mathrm{SAGE}(\gamma^*))
\succeq0.$
 
\end{theorem}

Theorem~\ref{thm: var_opt} provides a random-shift analogue
of classical variance reduction results based on calibration
and auxiliary information
\citep{deville1992calibration, chen2000unified,
chatterjee2016constrained}.
It characterizes the optimal adjustment within $\mathcal A$
when uncertainty arises from both sampling variability
and random distributional perturbations.

\subsection{Entropy balancing is asymptotically a full-step update}\label{sec: eb}
We next compare SAGE with EB, which reweights source observations to match target subgroup summaries and minimizes the resulting weighted loss. Like SAGE, it can be applied when individual target observations are unavailable but we have access to target summaries. The weights minimize KL divergence from uniform source weights subject to the matching constraints \citep{entropy_bal_Hainmueller, transportable_inf}.
Let 
$Z_{ig}^{(s)}=
\begin{pmatrix}
{A_{ig}^{(s)}}^\intercal &
{A_{ig}^{(s)}X_i^{(s)}}^\intercal &
{A_{ig}^{(s)}y_i^{(s)}}^\intercal
\end{pmatrix}^\intercal \in \mathbb{R}^{L+2}$ denote the subgroup-$g$ feature vector of source observation $i$, and let  $\hat d_{t,g}$ be defined as in Theorem \ref{thm: var_opt}. Denote the vectors obtained by stacking $Z_{ig}^{(s)}$ and $\hat d_{t,g}$ over $g=1,\ldots,K$ by $Z_i$ and $\hat d_t$, respectively. Apply the same fixed coordinate reduction of
Assumption~\ref{assump: var(z) pd} to $Z_i$ and $\hat d_t$,
retaining the notation. The omitted constraints follow from
the retained constraints and $\sum_i w_i=1$, so the feasible
set and entropy-balancing weights are unchanged.

The entropy-balancing weights are then obtained by solving
\begin{equation}
\label{eqn: eb_weights}
\begin{aligned}
\hat{w}
=\argmin_{w_1,\ldots,w_{n_s}\geq 0}
&\quad \sum_{i=1}^{n_s} w_i\log(n_s w_i)\\
\text{subject to}
&\quad \sum_{i=1}^{n_s}w_i=1,
\quad
\sum_{i=1}^{n_s}w_i Z_i=\hat d_t.
\end{aligned}
\end{equation}
We then fit the following predictor using the resulting weights and source observations:
\begin{equation}
\label{eqn: eb_predictor}
\hat\theta_{\mathrm{EB}}
= \argmin_{\theta}
\sum_{i=1}^{n_s}\hat w_i\,
\ell(\theta;X_i^{(s)},y_i^{(s)}).
\end{equation}
The following lemma shows that EB is asymptotically equivalent to full-step SAGE ($\gamma=1$), establishing a random-shift analogue of the classical equivalence between calibration weighting and regression adjustment \citep{deville1992calibration}. Together with Theorem~\ref{thm: var_opt}, this implies that SAGE at the optimal step-size $\gamma^*<1$ has asymptotic covariance no larger than that of EB. The proof is given in Appendix~\ref{app: entropy_bal}.

\begin{lemma}[EB is asymptotically a full-step update]\label{lemma: entropy_bal}
Let $\Sigma(\cdot)$ be as in Theorem~\ref{thm: var_opt}. Assume that there exists a constant $\kappa\in(0,1)$ such that, with probability tending to one as $m\to\infty$, the
solution to Eqn.~\eqref{eqn: eb_weights} admits finite Lagrange multipliers and its weights satisfy
$\min_{1\le i\le n_s} n_s\hat w_i\geq\kappa$. Furthermore, assume that there exists $\eta>0$ such that
$E_t[\exp\{\eta\|Z^{(t)}(D)\|_2\}]<\infty$, where $Z^{(t)}(D)$ denotes the stacked feature vector, defined as $Z_i$ above, of a draw $D$ from $P_t$. Then under the regularity conditions in Appendix~\ref{app: assumptions},
$$
\hat\theta_{\mathrm{EB}}
-
\hat\theta_{\mathrm{SAGE}}(1) = o_p(m^{-1/2}),
$$ and, by Theorem \ref{thm: var_opt},
$$
\Sigma(\hat\theta_{\mathrm{EB}})
-
\Sigma(\hat\theta_{\mathrm{SAGE}}(\gamma^*))
\succeq0.
$$
\end{lemma}

\subsection{Estimating the optimal step-size from summaries}\label{sec: gamma_est}
Theorem~\ref{thm: var_opt} and Lemma~\ref{lemma: entropy_bal} concern SAGE at the optimal step-size $\gamma^*$, which depends on unknown quantities. We now show how to estimate $\gamma^*$ using only subgroup summaries in Algorithm \ref{alg: sage_gamma}.
\begin{algorithm}[H]
\caption{SAGE: Optimal step-size estimation}
\label{alg: sage_gamma}
\begin{algorithmic}[1]
\State Form the source feature vectors $Z_i$, $i=1,\ldots,n_s$, by stacking
$Z_{ig}^{(s)}=
\begin{pmatrix}
A_{ig}^{(s)} &
(A_{ig}^{(s)}X_i^{(s)})^\intercal &
A_{ig}^{(s)}y_i^{(s)}
\end{pmatrix}^\intercal$
over $g=1,\ldots,K$.

\State Form the source and target summary vectors $\hat d_k$ by stacking, over $g=1,\ldots,K$,
\Statex \hspace{0.6cm}
$\hat d_{k,g}
\gets
\begin{pmatrix}
\hat\pi_{k,g}\\
\hat\pi_{k,g}\hat E_k[X\mid A_g=1]\\
\hat\pi_{k,g}\hat E_k[y\mid A_g=1]
\end{pmatrix}$ for $k\in\{s,t\}$

\State Remove summary coordinates that are redundant with the remaining
coordinates and an intercept, using the same fixed coordinate selection
for every $Z_i$, $\hat d_s$, and $\hat d_t$.
Denote the reduced vectors by the same symbols, and let $r=\dim Z$ be the
number of retained coordinates.

\State Form the target-source summary difference:
$\Delta\gets\hat d_t-\hat d_s$

\State Estimate the source covariance:
$\widehat S_s
\gets
\frac{1}{n_s}\sum_{i=1}^{n_s}
(Z_i-\hat d_s)
(Z_i-\hat d_s)^\intercal$

\State Compute the whitened, standardized summary difference:\label{alg: step_z}
$z
\gets
((1/n_s+1/n_t)\widehat S_s)^{-1/2}\Delta$

\State Estimate the distributional variance:
$\hat\delta_{\mathrm{dist}}^2
\leftarrow
\left(\frac1{n_s}+\frac1{n_t}\right)
\max\left\{\frac{\|z\|_2^2}{r}-1,0\right\}.$

\State Estimate the optimal step-size:
$\hat\gamma
\gets
\frac{n_s^{-1}+\hat{\delta}^2_{\rm dist}}
{n_s^{-1}+\hat{\delta}^2_{\rm dist}+n_t^{-1}}$

\State \Return $\hat\gamma$
\end{algorithmic}
\end{algorithm}

The intuition behind Algorithm~\ref{alg: sage_gamma} is that $z$ standardizes the source-target summary differences by sampling variability alone. Without distribution shift, each coordinate of $z$ is approximately standard Gaussian, so $\|z\|_2^2/r\approx 1$. Under the random shift model, Lemma~\ref{lemma: distr_clt} implies that the variance of each coordinate is instead inflated to approximately $1+\delta^2_{\rm dist}/(1/n_s+1/n_t)$, so the excess of $\|z\|_2^2/r$ over $1$ estimates the strength of the shift. Since this estimate averages over the $r$ summary
coordinates, consistency requires the summary dimension $r$ to
grow. We therefore take limits
first in $m$ (with $m/n_k\to c_k$) and then in $L$.

The following lemma establishes consistency of the estimator from Algorithm~\ref{alg: sage_gamma} in this iterated limit. The proof can be found in Appendix \ref{app: gamma_consist}.

\begin{lemma}[Consistency of optimal step-size estimator]
\label{lem: sage_gamma}
Consider the random shift model and the regularity conditions
of Appendix~\ref{app: assumptions}, with the growing-$L$ regime
specified in Assumption~\eqref{assump: grow_L}.
Let $r=\dim Z(D)$, and let $z$ denote the whitened,
standardized source-target summary difference in
Algorithm~\ref{alg: sage_gamma}. Recall that $\delta^2_{\rm dist} = \operatorname{Var}(W^{(s)})/m$. Let
\begin{equation}\label{eqn: delta_hat}
\hat{\delta}^2_{\rm dist}
=
\left(\frac{1}{n_s}+\frac{1}{n_t}\right)
\left(\frac{\|z\|_2^2}{r}-1\right)_+,
\qquad (x)_+\coloneqq\max\{x,0\},    
\end{equation}
and
$$
\hat\gamma
=
\frac{
1/n_s+\hat{\delta}^2_{\rm dist}
}{
1/n_s+\hat{\delta}^2_{\rm dist}+1/n_t
}.
$$
Then, for every $\varepsilon>0$,
$$
\lim_{L\to\infty}\limsup_{m\to\infty}
P\left(
\left|m\hat{\delta}^2_{\rm dist}-\operatorname{Var}(W^{(s)})\right|
>\varepsilon
\right)=0,
$$
and
$$
\lim_{L\to\infty}\limsup_{m\to\infty}
P\left(
|\hat\gamma-\gamma^*|>\varepsilon
\right)=0,
\qquad
\gamma^*=\frac{c_s+\operatorname{Var}(W^{(s)})}{c_s+\operatorname{Var}(W^{(s)})+c_t}.
$$
Here, the dependence on $m$ and $L$ is suppressed, and the limits are taken first with respect to $m$,
followed by $L$.
\end{lemma}

\subsection{Confidence intervals for \texorpdfstring{$\theta_t$}{the target risk minimizer}}\label{sec: ci}
Because $\theta_t$ cannot be estimated through direct target empirical risk minimization, inference must be based on the same summary information used for adaptation and requires accounting for both sampling and distributional variability. 

We first examine the limiting distribution of the SAGE estimator at a fixed $\gamma$. For a draw $D$ from $P_k$, $k\in\{s,t\}$, let $Z^{(k)}(D)$ denote its subgroup feature vector, obtained by stacking $(A_g, A_gX^\intercal, A_gy)^\intercal$ over $g=1,\ldots,K$ as in $Z_i$ above, and let $d_t\coloneqq E_t[Z^{(t)}(D)]$ denote its target mean, the population analog of $\hat d_t$. Let $\beta^{(t)}(\theta_t)$ denote the population slope
coefficient matrix from projecting
$\nabla_\theta\ell(\theta_t;D)$ onto
$\widetilde Z(D)=(1,Z(D)^\intercal)^\intercal$
under $P_t$, the population analog of Step \ref{alg: ols_step} of Alg. \ref{alg: grad_linear}. Define $\Xi_k(D)\coloneqq \left\{\beta^{(t)}(\theta_t)\right\}^\intercal \{Z^{(k)}(D)-d_t\}$. Define the $p$-dimensional source and target influence functions
\begin{align*}
\Phi_s(D;\gamma) \coloneqq
-H_t^{-1}
\left\{
\nabla_\theta\ell(\theta_t; D)
-
\gamma\Xi_s(D)
\right\}, \qquad
\Phi_t(D;\gamma) \coloneqq
\gamma H_t^{-1}\Xi_t(D),
\end{align*} 
where $H_t$ is the population Hessian, which we estimate with $\widehat{H}_s$ in Step \ref{alg: step_hess} of Alg. \ref{alg: grad_linear}. Intuitively, $\Phi_s$ is the source gradient with the part explained by the subgroup features removed, and $\Phi_t$ is the correction contributed by the target summaries. Since $\Phi_t$ is linear in $Z^{(t)}$, its target sample mean depends on the target data only through the subgroup summaries. The following corollary shows that for a fixed $\gamma$, $\hat{\theta}_{\rm SAGE}(\gamma)$ admits a linear representation composed of sample means of the influence functions $\Phi_s(D;\gamma), \Phi_t(D;\gamma)$. Applying the distributional CLT (Lemma~\ref{lemma: distr_clt}) then gives asymptotic normality at $\gamma^*$.


\begin{corollary}[Asymptotic normality of SAGE at $\gamma^*$]\label{cor: sage_normality}
Under the regularity conditions in Appendix~\ref{app: assumptions}
and $m/n_k\to c_k>0$ for $k\in\{s,t\}$ and a fixed $L$, SAGE evaluated at a fixed  $\gamma$ has the asymptotic linear representation 
\begin{equation*}
    \sqrt{m}(\hat{\theta}_\text{SAGE}(\gamma)-\theta_t) =  \sqrt{m}\left(\hat{E}_s[\Phi_s(D;\gamma)]-\hat{E}_t[\Phi_t(D;\gamma)]\right)+o_p(1),
\end{equation*}
 where $\hat{E}_s[\Phi_s(D;\gamma)]$ and $\hat{E}_t[\Phi_t(D;\gamma)]$ are the source and target sample means, respectively. Both influence functions have mean zero under $P_t$, since $E_t[\nabla_\theta\ell(\theta_t;D)]=0$ and $E_t[Z^{(t)}(D)]=d_t$. In particular,
\[
\sqrt m\{
\hat\theta_{\mathrm{SAGE}}(\gamma^*)-\theta_t
\}
\overset{d}{\rightarrow}
N\left(0,\Sigma_{\mathrm{SAGE}}(\gamma^*)\right),
\]
where $\gamma^*$ is defined in Theorem~\ref{thm: sage_risk} and $\Sigma_{\mathrm{SAGE}}(\gamma^*)
=c_t\operatorname{Var}_t\{\Phi_{t}(D;\gamma^*)\}+(c_s+\operatorname{Var}(W^{(s)}))
\operatorname{Var}_t\{\Phi_{s}(D;\gamma^*)\}$. Dependence on $L$ is suppressed in the notation.
\end{corollary}

Using the results from Corollary \ref{cor: sage_normality}, we now turn to constructing confidence intervals for SAGE with the estimated step-size $\hat\gamma$. Although individual target samples are unavailable, source observations can be used to estimate the target variability of the influence functions: by Lemma~\ref{lemma: distr_clt}, source sample moments converge to the corresponding target population moments, so for each fixed $L$ and fixed $\gamma$, empirical
covariances of the fitted influence functions over source
observations consistently estimate
$\operatorname{Var}_t\{\Phi_s(D;\gamma)\}$ and
$\operatorname{Var}_t\{\Phi_t(D;\gamma)\}$. The plug-in confidence interval additionally relies on consistency of $\hat\gamma$ and $\hat\delta^2_{\rm dist}$ in the iterated limit of Lemma~\ref{lem: sage_gamma}. The fixed-$L$ normality result above applies to a fixed step-size, including $\gamma^*$, rather than the estimated step-size $\hat\gamma$.

Let $\hat\gamma$ and $\hat\delta^2_{\rm dist}$ be defined in Lemma~\ref{lem: sage_gamma}. For a source observation $D$, define the fitted centered projection
\[
\widehat\Xi(D)
\coloneqq
\left\{\hat\beta^{(s)}(\hat\theta_{\mathrm{init}})\right\}^\intercal
\{Z^{(s)}(D)-\hat d_t\},
\]
where $\hat\beta^{(s)}(\hat\theta_{\mathrm{init}})$ stacks the fitted subgroup coefficients corresponding to the
retained coordinates from Algorithm~\ref{alg: grad_linear}. The plug-in influence functions, evaluated on source observations for both covariance terms, are
\begin{align*}
\widehat\Phi_s(D;\gamma)
&\coloneqq
-\widehat H_s^{-1}
\left\{\nabla_\theta\ell(\hat\theta_{\mathrm{init}};D)
-\gamma\widehat\Xi(D)\right\},\\
\widehat\Phi_t(D;\gamma)
&\coloneqq
\gamma\widehat H_s^{-1}\widehat\Xi(D).
\end{align*}
These use only source observations and the available target subgroup summaries. Define
\[
\widehat V_{\mathrm{SAGE}}
=
\frac{1}{n_t}
\widehat{\operatorname{Var}}_s
\left\{\widehat\Phi_t(D;\hat\gamma)\right\}
+
\left(\frac{1}{n_s}+\hat\delta^2_{\rm dist}\right)
\widehat{\operatorname{Var}}_s
\left\{\widehat\Phi_s(D;\hat\gamma)\right\},
\]
where both empirical covariances are computed over the source observations. Under the conditions of the following lemma, $m\widehat V_{\mathrm{SAGE}}$ estimates $\Sigma_{\mathrm{SAGE}}(\gamma^*)$ from Corollary~\ref{cor: sage_normality}. The lemma provides a confidence interval that accounts for both sampling and distributional uncertainty.

\begin{lemma}[Confidence intervals]\label{lem: CI}
Under the conditions of Corollary~\ref{cor: sage_normality} and Lemma~\ref{lem: sage_gamma}, suppose additionally that, for each fixed $L$ and $k\in\{s,t\}$, $E_t[\|\Phi_k(D;\gamma^*)\|_2^4]<\infty.$ Then, for any fixed coordinate $j$ satisfying $e_j^\intercal\Sigma_{\mathrm{SAGE}}(\gamma^*)e_j>0$
for all sufficiently large $L$, 
\[
\hat\theta_{\mathrm{SAGE},j}(\hat\gamma)
\pm
z_{1-\alpha/2}
\sqrt{
e_j^\intercal\widehat V_{\mathrm{SAGE}}e_j
}
\]
is an asymptotic $1-\alpha$ confidence interval for $\theta_{t,j}$, with limits taken first
in $m$, followed by $L$. Here $e_j$ is the $j$th standard basis vector in $\mathbb R^p$.
\end{lemma}

\section{Empirical Examples}\label{sec: empirical}

\subsection{Simulations}\label{sec: sim}

We validate our main theoretical result of Theorem \ref{thm: sage_risk} with simulations generated under the random shift model using the \texttt{calinf} R package\footnote{Code to reproduce figures can be found in \url{https://github.com/yzhangi96/summary_powered_prediction}}. The target distribution $P_t$ contains $16$ covariates: 8 binary covariates drawn from $\text{Bern}(0.5)$ and 8 covariates drawn from  $N(0,1)$. We also generate four subgroups using two Bernoulli group indicators that are not included as predictors. The response is defined as $y^{(k)}=\sum_{j=1}^{8} X_j^{(k)} +(X_9^{(k)})^2+(X_{10}^{(k)})^2+ \sum_{j=11}^{16} X_j^{(k)}+1.1X_{1}^{(k)}X_{2}^{(k)}+ \epsilon_{\text{noise}}$, for $k\in\{s,t\}$. Under the target distribution $P_t$, the covariates and noise are drawn from the fixed distributions as described, with  $\epsilon_{\text{noise}}\sim \text{Unif}(-1,1)$. Under the source distribution $P_s$, covariates and noise are drawn from randomly shifted versions of these distributions, with shift strength $\operatorname{Var}(W^{(s)})$. In each trial, we draw source individual data and target subgroup summaries, construct the least-squares estimator $\hat\theta_\mathrm{init}$, and  compute $\hat\theta_\mathrm{SAGE}(\gamma)$ over a grid of step-sizes. No individual-level target data are used to construct the SAGE estimator.

Figure \ref{fig: sim} reports the empirical target excess risk of SAGE on a held-out test set (10,000 samples), averaged over 1,000 trials in three shift regimes. The panels correspond to settings where the source sampling and distributional uncertainty term, $1/n_s+\delta^2_{\rm dist}$, is smaller than, comparable to, or larger than the target summary noise term $1/n_t$. As predicted by the theory, the average target excess risk is minimized near the theoretical optimum $\gamma^*$ of Theorem \ref{thm: sage_risk} in each regime. The pattern also matches the intuition behind $\gamma^*$:  as source noise, $1/n_s+\delta^2_{\rm dist}$, grows relative to the target summary noise $1/n_t$, the target excess risk is minimized at increasing step-sizes, so the update moves further toward the summary-based correction.

\begin{figure}[H]
    \centering
    \includegraphics[width=0.93\linewidth]{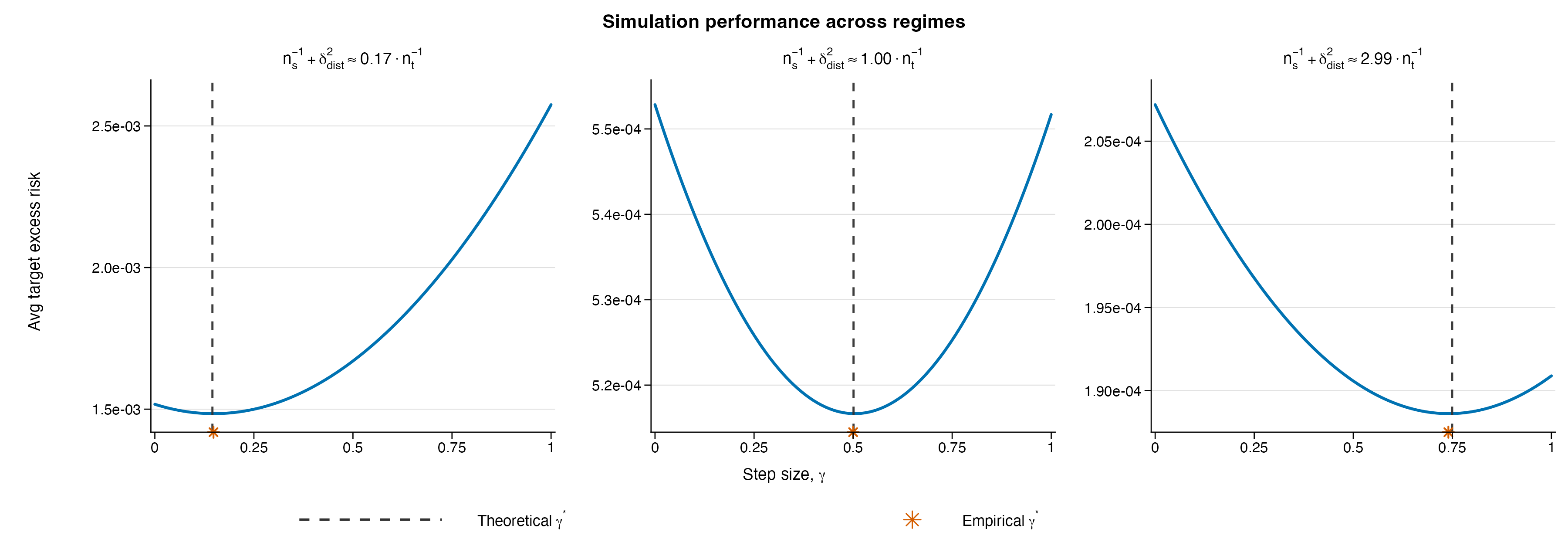}
    \captionsetup{width=\linewidth}
    \caption{Simulation across three random shift regimes of increasing source sampling and distributional variation relative to target subgroup summaries. Each panel shows the average target excess risk of SAGE over a grid of step-sizes, $\gamma$, with sample sizes (from left to right) of $(n_s=3000,n_t=500),(n_s=10000,n_t=8300),(n_s=64000,n_t=64000)$. The dashed vertical line marks the theoretical optimum $\gamma^*$ derived in Theorem \ref{thm: sage_risk}, and the orange star marks the empirical minimizer.}
    \label{fig: sim}
\end{figure}

\subsection{Store item demand prediction}\label{sec: store}

We now return to the Store Item Demand prediction example previewed in Figure~\ref{fig: store1}, using the Kaggle data introduced in Section~\ref{sec: intro}. The QQ plots in Figures~\ref{fig: qqplot_store_intro} and~\ref{fig: qqplot_store} show the variance-inflated pattern motivating the random shift model.
Our goal is to predict (log) sales for a given item and day. Each experiment treats Store 01 as the source distribution and one of Stores 02-10 as the target distribution, with \texttt{Item} used as the subgroup. We train the initial OLS predictor on $n_s=5000$ randomly sampled individual-level source records using 57 features. We then combine the source data with item level target summaries computed from $n_t=300$ randomly sampled target records to estimate the SAGE update.

Figure \ref{fig: store2} benchmarks SAGE, two other one-step estimation approaches, EB as described in Section~\ref{sec: eb}, and an oracle, against the source-trained predictor. All one-step methods use the target summary data to estimate the one-step correction of Algorithm \ref{alg: grad_linear} but use different step-size estimates. `One-step $\gamma=1$' takes an undamped Newton step. The optimal step-size approaches one when target summary noise is negligible relative to source sampling and distributional uncertainty. `One-step shift-ignoring' omits the estimated shift magnitude when setting the step-size, that is, it assumes an i.i.d.\ setting with $\operatorname{Var}(W^{(s)})=0$. Finally, `Oracle' selects the value of $\gamma$ that minimizes MSE on the target test set and represents the optimal step-size if we had access to a large number of target observations. 

The plot in Figure \ref{fig: store2} reports the percentage of MSE relative to the source-trained predictor, averaged over $300$ trials, with lower values indicating better performance. Each method improves over the source-trained predictor in all nine store shifts, indicating the value of incorporating summary data for prediction. SAGE, with the estimated step-size $\hat\gamma$, performs close to the Oracle. It also outperforms the one-step $\gamma=1$ and shift-ignoring update, indicating that the data-adaptive step-size that accounts for both sample size $n_t$ and distribution shift further improves predictive performance. EB matches or outperforms SAGE in three store shifts and performs worse in the remaining six. Under the assumptions of Lemma~\ref{lemma: entropy_bal}, exact EB is asymptotically equivalent to the full-step update ($\gamma=1$), while SAGE at $\gamma^*$ has asymptotic covariance no larger than that of EB. These results do not guarantee the finite-sample MSE ranking of EB and SAGE with an estimated step-size.

\begin{figure}[ht]
    \centering
    \includegraphics[width=0.75\linewidth]{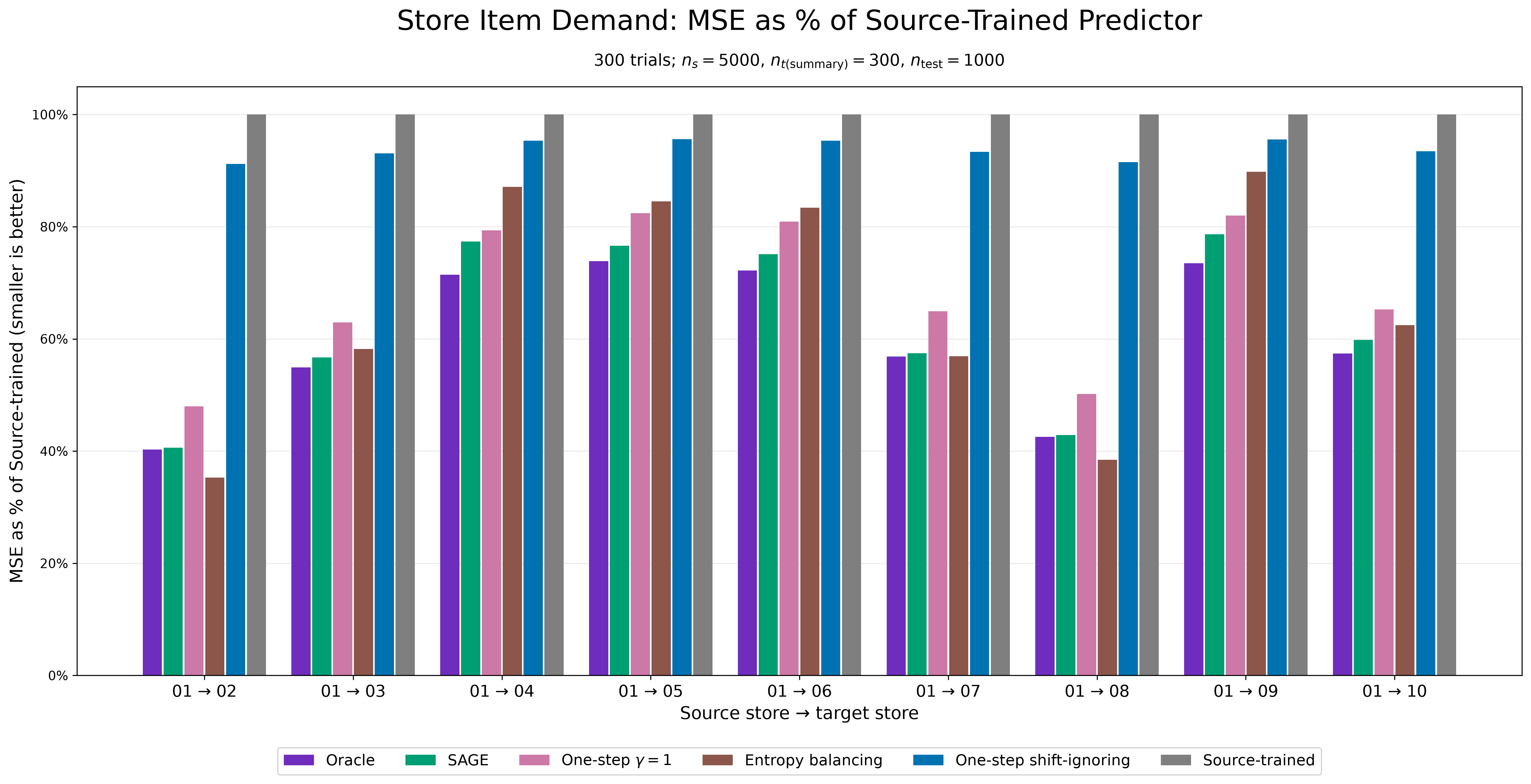}
    \caption{
Store Item Demand prediction results for transfers from Store 01 to Stores 02-10. The plot shows the percentage of target test MSE relative to the source-only OLS predictor, averaged over 300 trials. SAGE nearly matches the Oracle and consistently improves over the full-step and shift-ignoring one-step estimators.
}
    \label{fig: store2}
\end{figure}


\subsection{NHANES lead level prediction}\label{sec: Nhanes}
We next evaluate SAGE on the preprocessed NHANES lead dataset from the \texttt{tableshift} package \citep{tableshift} across temporal shifts. The goal is to predict individuals' blood lead levels using 16 covariates. For each source-target year pair, we train an initial predictor using $n_s=1000$ individual source observations. Target subgroup summaries are obtained from $n_t=800$ observations, with race as the reported subgroup. 

Before comparing prediction performance, we examine the diagnostic plot (see Figure~\ref{fig: qqplot_nhanes_lead} in Appendix \ref{app: qq}), which visualizes the standardized mean difference statistic in Eqn.~\eqref{eqn: qq_stat}. Unlike the Store Item Demand example, where the QQ plots were largely consistent with the random shift assumption, the NHANES lead shifts exhibit a mix of shift patterns. Some year pairs, such as $2001\rightarrow2005$, lie close to the i.i.d.\ reference line, suggesting little detectable shift in the observed covariates. Other pairs, especially larger gaps such as $1999\rightarrow2017$ and $1999\rightarrow2015$, show a pronounced S-curve. This pattern shows departures from the variance-inflated Gaussian pattern predicted by the random shift model. A hybrid of random perturbations and another shift mechanism is one possible interpretation. 

We show a subset of year pairs (gap $\geq 10$) in Figure~\ref{fig: nhanes_lead} and present the full dataset's 45 shift pairs in Appendix \ref{app: nhanes}. Figure~\ref{fig: nhanes_lead} compares two target summary sample sizes: $n_t=800$ and $n_t=10$, corresponding to $80\%$ and $1\%$ of the source sample size $n_s=1000$. SAGE improves on the source-trained predictor across all displayed year pairs with the larger target sample and most pairs with the smaller sample. When $n_t=800$, SAGE performs similarly to the full-step update ($\gamma=1$) and approaches the oracle benchmark. When $n_t=10$, SAGE generally outperforms the full-step update, which can adjust too aggressively to noisy target summaries. In some cases, the one-step estimator at $\gamma=1$ performs worse than the source-trained predictor. This comparison highlights the importance of choosing an appropriate step-size that accounts for both target summary sample size and estimated distribution shift, a key feature of SAGE. EB approximately matches the full-step estimator ($\gamma=1$) in several cases, but performs worse than the source-trained predictor in others. The asymptotic equivalence in Lemma~\ref{lemma: entropy_bal} concerns exact EB and the full-step update under the random shift model and the lemma's regularity conditions; it does not guarantee matching finite-sample performance for these temporal shifts.

\begin{figure}[H]
    \centering
    \includegraphics[width=\linewidth]{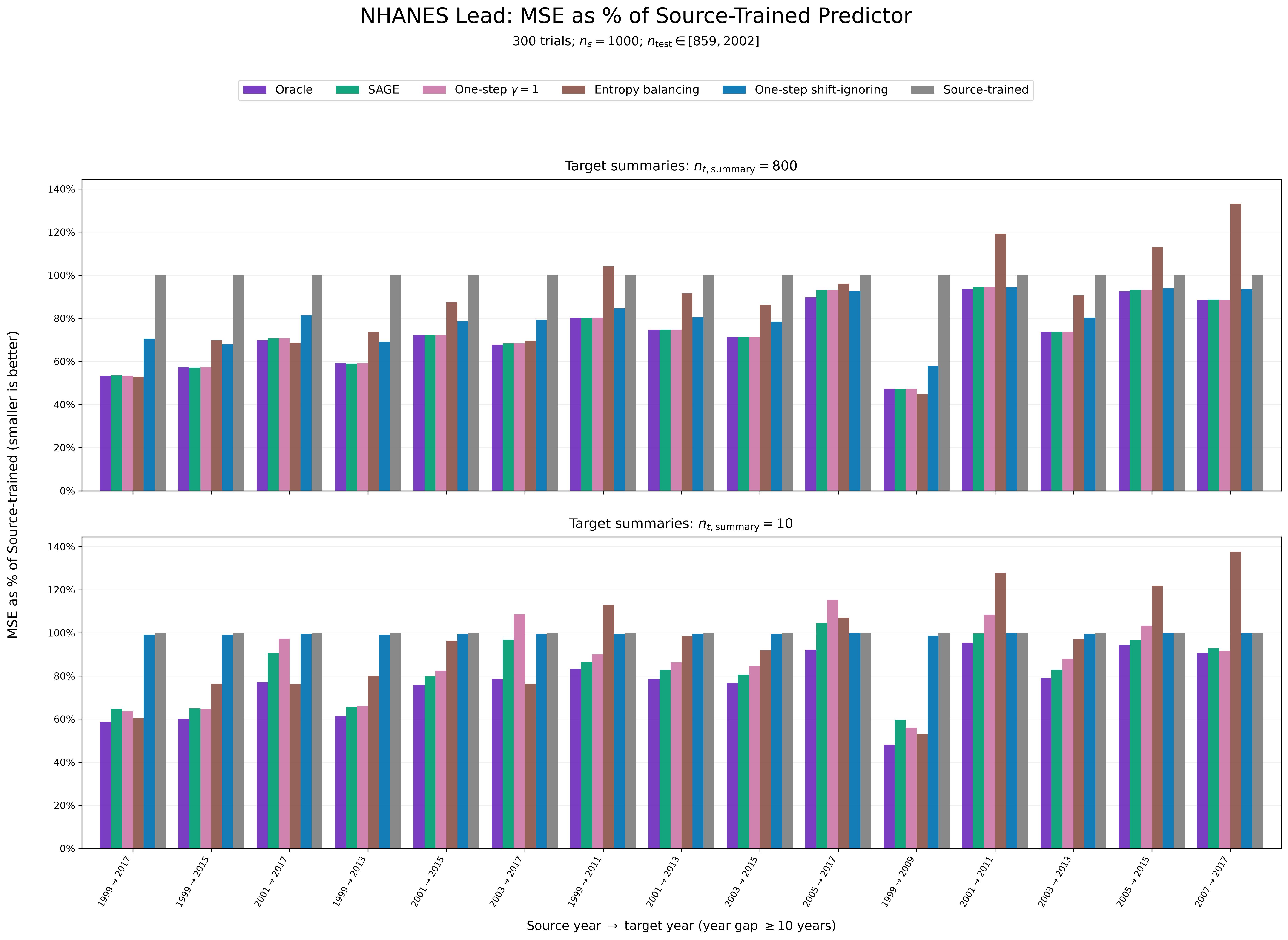}
        \captionsetup{width=\linewidth}
    \caption{
NHANES lead prediction results with target summary sample sizes $n_t=800$ (top) and $n_t=10$ (bottom). Bars show target test MSE as a percentage of the source-trained predictor's MSE, averaged over 300 trials; lower values indicate better performance. SAGE improves prediction for most year pairs. With $n_t=800$, it closely matches the full-step update ($\gamma=1$) and approaches the oracle benchmark; with $n_t=10$, it generally outperforms the full-step update. Meanwhile, for smaller $n_t$, the one-step $\gamma=1$ estimator performs worse than the source-trained predictor in several cases. 
}
    \label{fig: nhanes_lead}
\end{figure}

\begin{table}[ht]
\centering
\begin{tabular}{l r r}
\hline
& \multicolumn{2}{c}{Average test MSE (300 trials)} \\
Method
& $n_{t}=800$
& $n_{t}=10$ \\
\hline
source-trained predictor
& 0.107 & 0.107 \\
Oracle
& \textbf{0.074} & \textbf{0.078} \\
\textbf{SAGE}
& \textbf{0.074} & \textbf{0.085} \\
One-step $\gamma=1$
& \textbf{0.074} & 0.088 \\
One-step shift-ignoring
& 0.083 & 0.106 \\
EB
& 0.087 & 0.095 \\
\hline
\end{tabular}
\caption{Average test MSE across the 15 NHANES lead prediction year pairs (gap $\geq 10$ years), over 300 trials. The columns compare target subgroup summaries computed from $n_{t}=800$ and $n_{t}=10$ observations. Lower values indicate better performance. On average, SAGE performs better than other one-step estimators that do not adapt their step-size to account for sample size and distribution shift.}
\label{table: nhanes}
\end{table}

\subsection{Medicare inpatient payment prediction with SAGE-NN}\label{sec: medicare}

One practical advantage of SAGE is that the target summaries enter through a gradient correction, rather than through a restriction on the model class. This makes the summary adaptation idea compatible with gradient-based learners beyond smooth parametric models. We illustrate this by extending SAGE to neural networks.  

For neural network predictors, SAGE-NN applies the summary-adapted
correction to the final linear layer of a source-trained multilayer
perceptron (MLP). We first train the complete network on source
observations and then freeze the encoder. Then, we standardize the final hidden layer representation to improve numerical stability and refit the final linear layer by least squares on source data, ensuring that the empirical source gradient vanishes at $\hat\theta_{\mathrm{init}}$. We then apply the SAGE procedure to the final linear layer, that is, we construct the gradients via Algorithm~\ref{alg: grad_linear} using target subgroup summaries. For a predetermined encoder, the guarantees of Section~\ref{sec: sage} apply to the final layer parameters when the corresponding regularity conditions hold. SAGE-NN learns its encoder from the same source data used for adaptation, so these guarantees do not directly cover the full learned procedure: freezing the encoder does not account for its estimation uncertainty or dependence on the source sample. Details of SAGE-NN can be found in Algorithm \ref{alg: sage_nn}.

We evaluate SAGE-NN on the preprocessed 2011 CMS Medicare Inpatient Utilization and Payment dataset \citep{cmsinpatient2011}. Our goal is to predict the (log) payment received for a diagnosis group in each hospital (e.g., UCLA Medical Center knee replacement). Suppose our source population consists of hospitals in California. How can we leverage summary information from other states to make predictions for hospitals in those target states? We first perform diagnostics of the data to better understand the distribution shift. Similar to the patterns we saw in the QQ plots of the Store Item Demand data, the diagnostic plots for the Medicare data (see Figure~\ref{fig: qqplot_medicare} in Appendix~\ref{app: qq}) suggest that random perturbations explain part of these geographic shifts, although the S-curves indicate that some tail departures remain. For the NN adaptation of EB, the balanced weights are used to fit the final linear layer of the frozen source embeddings, mirroring the SAGE-NN adjustment on the final linear layer.

For this experiment, we consider Midwest target states: Michigan, Minnesota, Wisconsin, Ohio, Illinois, Indiana, and Missouri. For each source-target pair, we train an initial neural network on $n_s=5000$ California hospital diagnosis group records using 33 covariates after encoding categorical features. Target subgroup summaries are computed from $n_{t}=80$ records using Medicare's Diagnosis-Related Group (DRG) as the subgroup. The model is finally evaluated on an independent target test set of $n_{\mathrm{test}}=1000$ records, and results are averaged over 300 trials.

\begin{figure}[ht]
    \centering
    \includegraphics[width=0.75\linewidth]{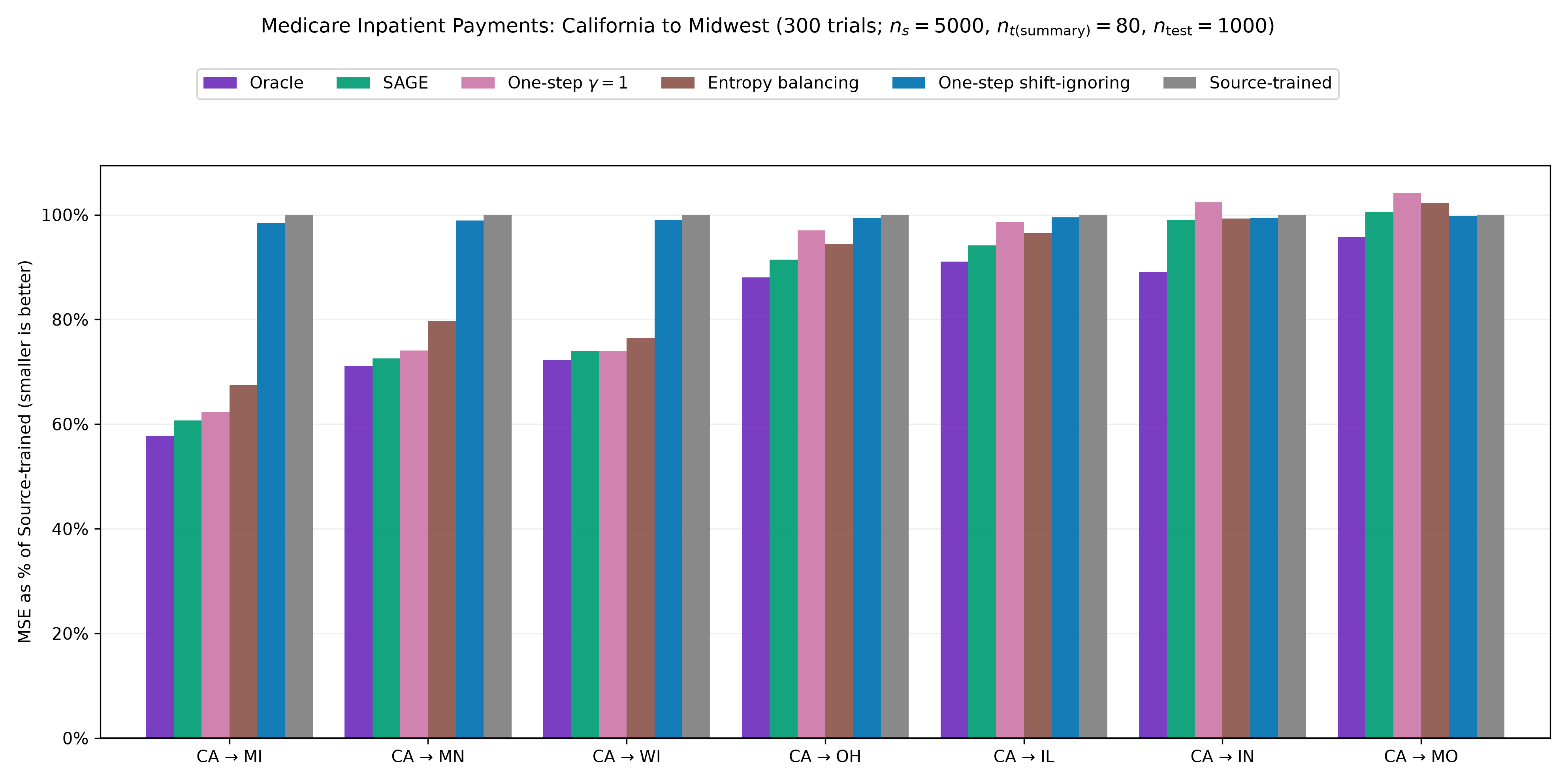}
    \caption{Medicare inpatient payment prediction results for transfers from California to seven Midwest states. Each panel shows target test MSE as a percentage of the source-only NN's MSE, averaged over 300 trials; lower values indicate better performance. SAGE-NN improves prediction in six states and generally outperforms the full-step and shift-ignoring one-step estimators, with its largest gains in Michigan, Minnesota, and Wisconsin, where performance is close to the oracle benchmark.
}
\label{fig: medicare1}
\end{figure}

Figure~\ref{fig: medicare1} reports target MSE as a percentage of that of the source-trained neural network. SAGE-NN improves prediction in six of the seven target states, with reductions of approximately $26$--$39\%$ in Michigan, Minnesota, and Wisconsin, where performance is close to the oracle benchmark. The shift-ignoring one-step estimator yields only small gains, highlighting the benefit of accounting for distribution shift. SAGE-NN also improves on the full-step update (One-step $\gamma=1$) and EB. Missouri and Indiana remain the most challenging: SAGE-NN performs about the same as the source-trained NN. However, in those cases, the one-step estimator with $\gamma=1$ performs worse than the source-only NN. These results show that refitting the final layer of a source-trained neural network with SAGE, using only target subgroup summaries, can improve prediction in the absence of individual-level target data.

\section{Discussion} \label{sec: disc}
Motivated by diagnostics consistent with the random shift model, we proposed SAGE for the setting where individual-level data are available from a shifted source population and the target population is observed only through subgroup summaries.
Our results highlight two sources of improvement over source-only training: incorporating target summaries to update a source-trained predictor and obtaining further gains by adapting the strength of the update to account for distribution shift. Under the random shift model, SAGE provides asymptotic risk reduction guarantees and plug-in confidence intervals that account for both sampling and distributional variability, all without requiring individual target observations. The step-size consistency and plug-in inference guarantees use an iterated limit with a growing summary dimension $r$. Empirically, SAGE and its extensions produce improved predictive performance across several real-world datasets. Furthermore, SAGE produces more consistent improvements compared to EB. These improvements are also present in settings where diagnostics suggest that random perturbations do not fully explain all observed shift.

One limitation is that SAGE is a one-step estimator and therefore relies on a local Taylor approximation to the target risk; when the source estimator is far from $\theta_t$ or the shift induces strongly nonlinear changes, this approximation may be less accurate. While the random shift model captures the variance inflated Gaussian patterns seen in several diagnostics, examples such as the S-shaped NHANES diagnostic in Figure~\ref{fig: qqplot_nhanes_lead} suggest that hybrid models combining random perturbations with additional structured shift components may better describe some datasets. Finite-sample risk and coverage guarantees, and guarantees that account for a source-trained encoder in SAGE-NN, remain directions for future work.

\section*{Appendix}
\appendix
    
\section{Details of main results}
\subsection{Notation}\label{app:notation}

Throughout the appendix, let $k\in\{s,t\}$ index the source and target
distributions, and let $g\in\{1,\ldots,K\}$ index the reported
subgroups. Let
$
D^{(k)}
=
\left(X^{(k)},y^{(k)},A^{(k)}\right)
\in\mathcal D
$
denote a draw from $P_k$, where
$X^{(k)}\in\mathbb R^L$, $y^{(k)}\in\mathbb R$, and
$
A^{(k)}
=
\left(A_1^{(k)},\ldots,A_K^{(k)}\right)
\in\{0,1\}^K.$
When the distribution is already specified by an operator such as
$E_k$, $\hat E_k$, $\operatorname{Var}_k$, or
$\operatorname{Cov}_k$, we suppress the distribution superscript on
the argument (e.g. we write $D=(X,y,A)$ for simplification).
\begin{longtable}{p{0.24\textwidth}p{0.71\textwidth}}
\hline
Notation & Meaning \\
\hline
\endfirsthead
\hline
Notation & Meaning \\
\hline
\endhead

$P_s, P_t$ &
Source and target distributions.
\\

$E_k$, $\operatorname{Var}_k$, $\operatorname{Cov}_k$ &
Expectation, variance, and covariance under $P_k$,
for $k\in\{s,t\}$.
\\

$\hat E_k[f]$ &
Empirical average of $f(D_i^{(k)})$ over the sample
from $P_k$.
\\

$n_s,n_t$ &
Source individual-level sample size and target summary
sample size.
\\

$m$ &
Number of regions in the random perturbation construction.
\\

$c_k$ &
Asymptotic sampling ratio,
$m/n_k\to c_k>0$, for $k\in\{s,t\}$.
\\

$p$ &
Dimension of the parameter vector $\theta$.
\\

$H_t$ &
Target Hessian,
$
H_t
\coloneqq
E_t[\nabla_\theta^2\ell(\theta_t;D)].
$
\\

$\widehat H_s$ &
Source empirical Hessian,
$
\widehat H_s
\coloneqq
\hat E_s[
\nabla_\theta^2\ell(\hat\theta_{\mathrm{init}};D)
].
$
\\

$\Psi_t(\theta)$ &
Population target gradient,
$
\Psi_t(\theta)
\coloneqq
E_t[\nabla_\theta\ell(\theta;D)].
$
\\

$A_g^{(k)}$ &
Indicator that a draw from $P_k$ belongs to reported
subgroup $g$.
\\

$\pi_{k,g}$, $\hat\pi_{k,g}$ &
Population and empirical subgroup proportions,
$
\pi_{k,g}=E_k[A_g],
\qquad
\hat\pi_{k,g}=\hat E_k[A_g].
$
\\

$Z_g^{(k)}(D)$ &
Subgroup summary feature vector,
$
Z_g^{(k)}(D)
\coloneqq
A_g^{(k)}
\left(
1,{X^{(k)}}^\intercal,y^{(k)}
\right)^\intercal
\in\mathbb R^{L+2}.
$
\\

$Z^{(k)}(D)$ &
Summary feature vector formed from the stacked subgroup features with $\operatorname{Var}_t\{Z(D)\}\succ0$
\\
$r$ &
$\dim Z(D)$
\\

$d_{k,g}$, $\hat d_{k,g}$ &
Population and empirical subgroup summary vectors,
$
d_{k,g}\coloneqq E_k[Z_g(D)],
\qquad
\hat d_{k,g}\coloneqq\hat E_k[Z_g(D)].
$
\\

$d_k$, $\hat d_k$ &
Population and empirical means of the summary features,
$
d_k\coloneqq E_k[Z(D)],
\qquad
\hat d_k\coloneqq\hat E_k[Z(D)].
$
\\

$\widetilde Z^{(k)}(D)$ &
Regression feature vector including an intercept,
$
\widetilde Z^{(k)}(D)
\coloneqq
\left(
1,Z^{(k)}(D)^\intercal
\right)^\intercal.
$
\\

$C^{(k)}(\theta)$, $\hat C^{(k)}(\theta)$ &
Population and empirical least-squares coefficient
matrices for projecting
$\nabla_\theta\ell(\theta;D)$
onto $\widetilde Z^{(k)}(D)$.
Their intercept vectors are
$\alpha^{(k)}(\theta)$ and $\hat\alpha^{(k)}(\theta)$,
and their slope matrices are
$\beta^{(k)}(\theta)$ and $\hat\beta^{(k)}(\theta)$,
respectively.
\\

$\varepsilon^{(k)}(\theta;D)$ &
Population projection residual,
$
\varepsilon^{(k)}(\theta;D)
\coloneqq
\nabla_\theta\ell(\theta;D)
-
\alpha^{(k)}(\theta)
-
\{\beta^{(k)}(\theta)\}^\intercal Z^{(k)}(D).
$
\\

$\Xi_k(D)$ &
Centered projected score,
$
\Xi_k(D)
\coloneqq
\{\beta^{(t)}(\theta_t)\}^\intercal
\{Z^{(k)}(D)-d_t\},
$
for $k\in\{s,t\}$.
\\

\hline
\end{longtable}
We write $Z_{ig}^{(s)}=Z_g(D_i^{(s)})$ and
$Z_i=Z(D_i^{(s)})$. Superscripts $k\in\{s,t\}$ on these
features indicate the sampling population; the feature maps
and retained coordinates are the same for source and target.
\subsection{Random distribution shift}\label{app: random_shift_construction}
We adapt the construction details from \cite{duc}, with $P_t$ being the fixed distribution, and include them here for completeness.

Let $D=(X,y,A)$ denote a draw from the underlying population distribution $P_t$ on the data space $\mathcal D$. Assume $P_t$ is nonatomic and admits a sequence of finite equal-probability partitions with the approximation property below. For each $m$, let
$I_m=\{I_{1,m},\ldots,I_{m,m}\}$ be a measurable partition of
$\mathcal D$ into equal-probability regions, so that
$P_t(I_{j,m})=1/m$ for $j=1,\ldots,m$. Assume that for every
$h\in L^2(P_t)$,
\begin{equation}
    \left\|
    h(D)
    -
    \sum_{j=1}^m
    \mathbbm{1}_{\{D\in I_{j,m}\}}
    E_t[h(D)\mid D\in I_{j,m}]
    \right\|_2^2
    \to 0.
\end{equation}
This condition ensures that averages over partition cells approximate any square-integrable function as the partition refines.

Let $W_j^{(s)}$, $j=1,\ldots,m$, be i.i.d.\ positive random variables with mean one and variance $\operatorname{Var}(W^{(s)})$. We assume the weights are bounded away from zero: there exists $w_s>0$ such that $W_j^{(s)}\geq w_s$ for all $j$. For each region $I_{j,m}$, the weight $W_j^{(s)}$ determines how much probability mass source $P_s$ assigns to that region relative to $P_t$. Thus, we define the randomly perturbed distribution $P_s$ by setting its Radon-Nikodym derivative with respect to $P_t$,
\begin{equation*}
    \frac{dP_s}{dP_t}(D)
    =
    \frac{W_j^{(s)}}{\frac{1}{m}\sum_{j'=1}^m W_{j'}^{(s)}}
    \qquad
    \text{for }D\in I_{j,m}.
\end{equation*}
As a result, some regions in the data space of $P_s$ have randomly higher or lower probability mass compared to $P_t$.

Conditionally on $W=(W_j^{(s)})_{j=1,\ldots,m}$, $\{D_i^{(s)}\}_{i=1}^{n_s}$ are i.i.d.\ draws from $P_s$. The target sample is i.i.d.\ from $P_t$ and is independent of both $W$ and the source sample conditional on $W$. Let $m/n_k$ converge to some constant $c_k>0$ for $k\in\{s,t\}$. Formally, we consider sequences $n_k(m)$ as $m\rightarrow\infty$ such that $m/n_k(m)\rightarrow c_k>0$, a constant that measures the ratio between sampling and distributional uncertainty. Intuitively, this assumption means that sampling uncertainty (order $1/n_k$) and distributional uncertainty (order $1/m$) are of the same order. In the following, for simplicity, we notationally drop the dependence of $n_k$ on $m$. We chose this asymptotic regime because, in both other regimes (where either sampling uncertainty or distributional uncertainty dominates), the trade-off between individual source data quality and quantity versus that of target summary's is not apparent. See \cite{random_shift1} and \cite{random_shift3} for additional diagnostics to test the random shift assumption in real data.

\subsection{Technical conditions}
\label{app: assumptions}

We use the notation from Appendix~\ref{app:notation}. 

\begin{enumerate}
    \item \textbf{Sampling regime and random shift model.}
    \begin{enumerate}
        \item \label{assump: grow_L} For Lemmas~\ref{lem: sage_gamma} and~\ref{lem: CI}, fix a sequence
of covariate features $X_1,X_2,\ldots$ on $\mathcal D$.
Let $Z_L(D)\in\mathbb R^{r_L}$ be the reduced summary vector
formed from $X_1,\ldots,X_L$, $y$, and the $K$ fixed subgroup
indicators, as in Assumption~\eqref{assump: var(z) pd} below.
As $L\to\infty$, assume $r_L\to\infty$ while keeping
$P_t$, $\ell$, $\Theta\subset\mathbb R^p$, the distribution
of $W^{(s)}$, and $c_s,c_t$ fixed.
The conditions below hold for each fixed $L$.
Limits are taken first in $m$, with $m/n_k\to c_k$,
and then in $L$; dependence on $L$ is otherwise suppressed. Therefore, only the collection of summary features grows;
the prediction model and $\theta_t$ remain fixed.
        \item The target distribution $P_t$ is fixed. The source distribution $P_s$ is generated as a dense random perturbation of $P_t$, as described in Appendix~\ref{app: random_shift_construction}. 
        \item We consider the asymptotic regime where
        $\frac{m}{n_k}\to c_k\in(0,\infty)$ for $k\in \{s,t\}$.
    \end{enumerate}
\item \textbf{Regularity of the target and source risk.}
\begin{enumerate}
    \item The parameter space $\Theta\subset\mathbbm{R}^p$ is compact, and $\theta_t\in\operatorname{int}(\Theta)$ is the unique minimizer of $R_t(\theta)\coloneqq E_t[\ell(\theta;D)]$.
    \item There exist $L_0,L_1\in L^2(P_t)$ such that, for all
    $\theta,\theta'\in\Theta$, $|\ell(\theta;D)|\leq L_0(D)$ and  $|\ell(\theta;D)-\ell(\theta';D)| \leq L_1(D)\|\theta-\theta'\|_2$ \label{assump: lip_loss}
    \item The loss $\ell(\theta;D)$ is three times continuously differentiable
    in $\theta$ on an open neighborhood $\mathcal N\subset\Theta$ of
    $\theta_t$, for $P_t$-almost every $D$.

    \item The target Hessian
    $H_t\coloneqq E_t[\nabla_\theta^2\ell(\theta_t;D)]$
    exists and is positive definite.

    \item \label{assump: sup_op_bd} There exist functions $F_j\in L^2(P_t)$, $j=1,2,3$, such that $\sup_{\theta\in\mathcal N}
        \left\|
        \nabla_\theta^j\ell(\theta;D)
        \right\|_{\mathrm{op}}
        \le F_j(D)$
    \item \label{assump: lip} For $j=1,2,3$, there exist $M_j\in L^2(P_t)$
    such that for all $\theta,\theta'\in\mathcal N$, $\left\|
        \nabla_\theta^j\ell(\theta;D)
        -
        \nabla_\theta^j\ell(\theta';D)
        \right\|_{\mathrm{op}}
        \le
        M_j(D)\|\theta-\theta'\|_2.$
\end{enumerate}
    \item \textbf{Moment and identifiability conditions.}
\begin{enumerate}
    \item
    $
    E_t\left[
    \left\|
    \widetilde Z(D)\widetilde Z(D)^\intercal
    \right\|_F^2
    \right]
    <\infty
   $ and $
    E_t\left[
    \sup_{\theta\in\mathcal N}
    \left\|
    \widetilde Z(D)
    \nabla_\theta\ell(\theta;D)^\intercal
    \right\|_F^2
    \right]
    <\infty.$
    In addition, \(E_t\|Z(D)\|_2^2<\infty\), so the subgroup
    summary means \(d_k\) and \(\hat d_k\) are well defined. 
    \item For each fixed $L$, remove summary coordinates that are
redundant with the remaining coordinates and an intercept,
using the same fixed choice for source and target.
Throughout, $Z(D)\in\mathbb R^{r}$ denotes the resulting
summary feature vector, and all summary means and
projection coefficients use these coordinates.
Assume that
$\operatorname{Var}_t\{Z(D)\}\succ0$.
    \label{assump: var(z) pd}

    \item
    $
    \operatorname{tr}\left\{H_t^{-1}\operatorname{Var}_t\left(\Xi_t(D)\right)
    \right\}>0
    $
    \label{assump: trace}
\end{enumerate}
\end{enumerate}

\subsection{Proof of Theorem~\ref{thm: sage_risk}}\label{app: thm1}

\begin{proof}
\textbf{(Optimal step-size).}
We first derive the step-size $\gamma^*$ that minimizes the mean of the
limiting scaled excess risk variable. A Taylor expansion of the target risk around $\theta_t$ gives
\begin{align}
\label{eqn: risk_taylor}
\begin{split}
E_t[
    \ell(\hat\theta_\mathrm{SAGE}(\gamma);D)
]
&=
E_t[
    \ell(\theta_t;D)
]
+
E_t[
    \nabla_\theta\ell(\theta_t;D)
]^\intercal
\left\{
    \hat\theta_\mathrm{SAGE}(\gamma)-\theta_t
\right\}
\\
&\quad+
\frac{1}{2}
\left\{
    \hat\theta_\mathrm{SAGE}(\gamma)-\theta_t
\right\}^\intercal
H_t
\left\{
    \hat\theta_\mathrm{SAGE}(\gamma)-\theta_t
\right\}
+
r_m.
\end{split}
\end{align}
Since $\theta_t$ is the target risk minimizer, $E_t[ \nabla_\theta\ell(\theta_t;D)]=0.$
Moreover, Corollary~\ref{corr: grad_linear} gives $\hat\theta_\mathrm{SAGE}(\gamma)-\theta_t
=O_p(m^{-1/2}).$ The remainder in Eqn \eqref{eqn: risk_taylor} can be written as
\begin{align*}
r_m
&=
\frac{1}{2}
\left\{
    \hat\theta_\mathrm{SAGE}(\gamma)-\theta_t
\right\}^\intercal
E_t\left[
    \nabla_\theta^2\ell(\widetilde\theta;D)
    -
    \nabla_\theta^2\ell(\theta_t;D)
\right]
\left\{
    \hat\theta_\mathrm{SAGE}(\gamma)-\theta_t
\right\},
\end{align*}
where $\widetilde\theta$ lies between
$\hat\theta_\mathrm{SAGE}(\gamma)$ and $\theta_t$. By
Assumption~\eqref{assump: lip},
\begin{align*}
|r_m|
&\leq
\frac{1}{2}
E_t[M_2(D)]
\|
    \widetilde\theta-\theta_t
\|_2
\|
    \hat\theta_\mathrm{SAGE}(\gamma)-\theta_t
\|_2^2
\\
&\leq
\frac{1}{2}
E_t[M_2(D)]
\|
    \hat\theta_\mathrm{SAGE}(\gamma)-\theta_t
\|_2^3
\\
&=
O_p(m^{-3/2})
\\
&=
o_p(m^{-1}).
\end{align*}
Therefore,
\begin{align}
\label{eqn: scaled_risk_quadratic}
m\left\{
E_t[
    \ell(\hat\theta_\mathrm{SAGE}(\gamma);D)
]
-
E_t[
    \ell(\theta_t;D)
]
\right\}
=
\frac{1}{2}
\left[
    \sqrt m
    \left\{
        \hat\theta_\mathrm{SAGE}(\gamma)-\theta_t
    \right\}
\right]^\intercal
H_t
\left[
    \sqrt m
    \left\{
        \hat\theta_\mathrm{SAGE}(\gamma)-\theta_t
    \right\}
\right]
+
o_p(1).
\end{align}
By Corollary~\ref{corr: grad_linear}, we have the linearization
\begin{align}
\label{eqn: sage_risk_linearization}
\sqrt m
\left\{
    \hat\theta_\mathrm{SAGE}(\gamma)-\theta_t
\right\}=
-H_t^{-1}
\left[
\sqrt m\,
\hat E_s\left[
    \nabla_\theta\ell(\theta_t;D)
    -
    \gamma\Xi_s(D)
\right]
+
\gamma\sqrt m\,
\hat E_t[\Xi_t(D)]
\right]
+
o_p(1).
\end{align}
Substituting Eqn \eqref{eqn: sage_risk_linearization} into
Eqn \eqref{eqn: scaled_risk_quadratic}, we obtain
\begin{align*}
m\left\{
E_t[
    \ell(\hat\theta_\mathrm{SAGE}(\gamma);D)
]
-
E_t[
    \ell(\theta_t;D)
]
\right\}
=
J_m(\gamma)+o_p(1),
\end{align*}
where
\begin{align*}
J_m(\gamma)
&\coloneqq
\frac{m}{2}
\left(
\hat E_s\left[
    \nabla_\theta\ell(\theta_t;D)
    -
    \gamma\Xi_s(D)
\right]
+
\gamma\hat E_t[\Xi_t(D)]
\right)^\intercal
H_t^{-1}
\left(
\hat E_s\left[
    \nabla_\theta\ell(\theta_t;D)
    -
    \gamma\Xi_s(D)
\right]
+
\gamma\hat E_t[\Xi_t(D)]
\right).
\end{align*}

Since $E_t[\nabla_\theta\ell(\theta_t;D)]=0$ and $E_t[\Xi_t(D)]=0,$ Lemma~\ref{lemma: distr_clt} gives
\begin{equation}
\label{eqn: risk_joint_clt}
\sqrt m
\begin{pmatrix}
\gamma\hat E_t[\Xi_t(D)]
\\[1mm]
\hat E_s\left[
    \nabla_\theta\ell(\theta_t;D)
    -
    \gamma\Xi_s(D)
\right]
\end{pmatrix}
\overset{d}{\rightarrow}
N
\left(
0,
\begin{pmatrix}
c_t\gamma^2
\operatorname{Var}_t\{\Xi_t(D)\}
&
0
\\
0
&
(c_s+\operatorname{Var}(W^{(s)}))
\operatorname{Var}_t
\left\{
    \nabla_\theta\ell(\theta_t;D)
    -
    \gamma\Xi_t(D)
\right\}
\end{pmatrix}
\right).
\end{equation}
It follows from the continuous mapping theorem that $J_m(\gamma)
\overset{d}{\rightarrow}
J(\gamma),$ with mean
\begin{align}
\label{eqn: EJ_before_projection}
E[J(\gamma)]
&=
\frac{1}{2}
\operatorname{tr}
\bigg[
H_t^{-1}
\bigg\{
(c_s+\operatorname{Var}(W^{(s)}))
\operatorname{Var}_t
\left[
    \nabla_\theta\ell(\theta_t;D)
    -
    \gamma\Xi_t(D)
\right]
+
c_t\gamma^2
\operatorname{Var}_t\{\Xi_t(D)\}
\bigg\}
\bigg].
\end{align}
Next, denote the population target projection residual,
\[
\varepsilon^{(t)}(\theta_t; D)
\coloneqq
\nabla_\theta\ell(\theta_t;D)
-
\alpha^{(t)}(\theta_t)
-
\{\beta^{(t)}(\theta_t)\}^\intercal Z^{(t)}(D).
\]
Since the projection
contains an intercept, $E_t[\varepsilon^{(t)}(\theta_t; D)]=0$ and $E_t\left[\varepsilon^{(t)}(\theta_t; D)\{Z^{(t)}(D)-d_t\}^\intercal\right]=0.$
Also, since $E_t[\nabla_\theta\ell(\theta_t;D)]=0$, $$\alpha^{(t)}(\theta_t)
+
\{\beta^{(t)}(\theta_t)\}^\intercal d_t
=
0.$$ Therefore,
$$
\nabla_\theta\ell(\theta_t;D)
=
\Xi_t(D)
+
\varepsilon^{(t)}(\theta_t; D).
$$
It follows that
\begin{align}
\label{eqn: projection_covariance_identity}
\begin{split}
\operatorname{Cov}_t
\left\{
    \nabla_\theta\ell(\theta_t;D),
    \Xi_t(D)
\right\}
&=
E_t\left[
    \{\Xi_t(D)+\varepsilon^{(t)}(\theta_t; D)\}
    \Xi_t(D)^\intercal
\right]
\\
&=
\operatorname{Var}_t\{\Xi_t(D)\}
+
E_t\left[
    \varepsilon^{(t)}(\theta_t; D)
    \{Z^{(t)}(D)-d_t\}^\intercal
\right]
\beta^{(t)}(\theta_t)
\\
&=
\operatorname{Var}_t\{\Xi_t(D)\}.
\end{split}
\end{align}
Consequently,
\begin{align*}
\operatorname{Var}_t
\left\{
    \nabla_\theta\ell(\theta_t;D)
    -
    \gamma\Xi_t(D)
\right\}=
\operatorname{Var}_t
\left\{
    \nabla_\theta\ell(\theta_t;D)
\right\}
+
(\gamma^2-2\gamma)
\operatorname{Var}_t\{\Xi_t(D)\}.
\end{align*}
Substituting this identity into Eqn
\eqref{eqn: EJ_before_projection} gives
\begin{align}
\label{eqn: EJ_gamma}
E[J(\gamma)]
&=
\frac{1}{2}
(c_s+\operatorname{Var}(W^{(s)}))
\operatorname{tr}
\left[
H_t^{-1}
\operatorname{Var}_t
\left\{
    \nabla_\theta\ell(\theta_t;D)
\right\}
\right]\\
&\quad+
\left[
\frac{\gamma^2}{2}
(c_s+\operatorname{Var}(W^{(s)})+c_t)
-
\gamma(c_s+\operatorname{Var}(W^{(s)}))
\right]
\operatorname{tr}
\left[
H_t^{-1}
\operatorname{Var}_t\{\Xi_t(D)\}
\right].
\end{align}
Differentiating Eqn \eqref{eqn: EJ_gamma} with respect to $\gamma$ gives
\begin{align*}
\frac{\partial}{\partial\gamma}E[J(\gamma)]
&=
\left[
\gamma(c_s+\operatorname{Var}(W^{(s)})+c_t)
-
(c_s+\operatorname{Var}(W^{(s)}))
\right]
\operatorname{tr}
\left[
H_t^{-1}
\operatorname{Var}_t\{\Xi_t(D)\}
\right].
\end{align*}
By Assumption \eqref{assump: trace}, $\operatorname{tr}
\left[
H_t^{-1}
\operatorname{Var}_t\{\Xi_t(D)\}
\right]
>
0.$ Moreover,
\begin{align*}
\frac{\partial^2}{\partial\gamma^2}E[J(\gamma)]
&=
(c_s+\operatorname{Var}(W^{(s)})+c_t)
\operatorname{tr}
\left[
H_t^{-1}
\operatorname{Var}_t\{\Xi_t(D)\}
\right]
>
0.
\end{align*}
Thus, the unique minimizer is
\begin{equation*}
\gamma^*
=
\frac{c_s+\operatorname{Var}(W^{(s)})}
{c_s+\operatorname{Var}(W^{(s)})+c_t}
\in(0,1).
\end{equation*}

\textbf{(SAGE reduces asymptotic risk mean).}
We next compare the mean of the limiting scaled excess-risk variable
under SAGE with that under the unadjusted source estimator. Setting
$\gamma=0$ in Eqn \eqref{eqn: EJ_gamma} gives
\begin{equation*}
E[J(0)]
=
\frac{1}{2}
(c_s+\operatorname{Var}(W^{(s)}))
\operatorname{tr}
\left[
H_t^{-1}
\operatorname{Var}_t
\left\{
    \nabla_\theta\ell(\theta_t;D)
\right\}
\right].
\end{equation*}
Therefore,
\begin{align}
\label{eqn: risk_mean_difference}
E[J(\gamma)]-E[J(0)]
=
\frac{\gamma}{2}
\left[
\gamma(c_s+\operatorname{Var}(W^{(s)})+c_t)
-
2(c_s+\operatorname{Var}(W^{(s)}))
\right]
\operatorname{tr}
\left[
H_t^{-1}
\operatorname{Var}_t\{\Xi_t(D)\}
\right].
\end{align}
For
\[
0
<
\gamma
<
2\frac{c_s+\operatorname{Var}(W^{(s)})}
{c_s+\operatorname{Var}(W^{(s)})+c_t}
=
2\gamma^*,
\]
we have
\[
\gamma(c_s+\operatorname{Var}(W^{(s)})+c_t)
-
2(c_s+\operatorname{Var}(W^{(s)}))
<
0.
\]
Since the trace term in Eqn \eqref{eqn: risk_mean_difference} is strictly positive, for every $0<\gamma<2\gamma^*,$
\begin{equation*}
E[J(\gamma)]
<
E[J(0)]
\end{equation*}
\end{proof}

\subsection{Proof of Theorem~\ref{thm: var_opt}}\label{app: thm2}

\begin{proof}
We first verify that
$\hat\theta_\mathrm{SAGE}(\gamma^*)$ belongs to
$\mathcal A$. Let \begin{equation}
\label{eqn: Bm_star}
B_m^*
\coloneqq
-\gamma^*
\widehat H_s^{-1}
\left\{
\hat\beta^{(s)}(\hat\theta_\mathrm{init})
\right\}^\intercal.
\end{equation}
By the definition of the summary-adapted gradient, $\widehat\Psi_t
=
\left\{
\hat\beta^{(s)}(\hat\theta_\mathrm{init})
\right\}^\intercal
(\hat d_t-\hat d_s)$, we have
\begin{align*}
\hat\theta_\mathrm{SAGE}(\gamma^*)=
\hat\theta_\mathrm{init}
-
\gamma^*
\widehat H_s^{-1}
\left\{
\hat\beta^{(s)}(\hat\theta_\mathrm{init})
\right\}^\intercal
(\hat d_t-\hat d_s)=
\hat\theta_\mathrm{init}
+
B_m^*(\hat d_t-\hat d_s).
\end{align*}

By Proposition~\ref{prop: hess_consist} and
Lemma~\ref{lemma: grad_consist}, $\widehat H_s^{-1}
=
H_t^{-1}+O_p(m^{-1/2})$ and 
$\hat\beta^{(s)}(\hat\theta_\mathrm{init})
=
\beta^{(t)}(\theta_t)
+
O_p(m^{-1/2})$, respectively. Therefore,  it follows that
\begin{equation}
\label{eqn: Bm_star_limit}
B_m^*
\overset{p}{\rightarrow}
B^*
\coloneqq
-\gamma^*
H_t^{-1}
\left\{
\beta^{(t)}(\theta_t)
\right\}^\intercal.
\end{equation}
Thus,
$\hat\theta_\mathrm{SAGE}(\gamma^*)\in\mathcal A$.

Now consider any
$\hat\theta_{B_m}\in\mathcal A$, where
$B_m\overset{p}{\rightarrow}B$ for a deterministic matrix $B$.
We have
\begin{align}
\sqrt m
(\hat\theta_{B_m}-\theta_t)
&=
\sqrt m
(\hat\theta_\mathrm{init}-\theta_t)
+
B_m\sqrt m(\hat d_t-\hat d_s)
\nonumber\\
&=
\sqrt m
(\hat\theta_\mathrm{init}-\theta_t)
+
B\sqrt m(\hat d_t-\hat d_s)
+
o_p(1),
\label{eqn: linear_adjustment_B}
\end{align}
because $\sqrt m(\hat d_t-\hat d_s)=O_p(1)$ by Lemma~\ref{lemma: distr_clt}. As a result,
\[
(B_m-B)\sqrt m(\hat d_t-\hat d_s)
=
o_p(1).
\]

Applying the linearization in
Proposition~\ref{prop: theta_old_consist} to Eqn \eqref{eqn: linear_adjustment_B} gives
\begin{align}
\sqrt m
(\hat\theta_{B_m}-\theta_t)
&=
-H_t^{-1}
\sqrt m\,
\hat E_s[
\nabla_\theta\ell(\theta_t;D)
]
+
B\sqrt m(\hat d_t-\hat d_s)
+
o_p(1).
\label{eqn: linear_adjustment_expansion}
\end{align}

Since $E_t[
\nabla_\theta\ell(\theta_t;D)
]
=
0,$ Lemma~\ref{lemma: distr_clt}, applied jointly to
$\nabla_\theta\ell(\theta_t;D)$ and $Z(D)$, gives
\begin{align}
\label{eqn: joint_score_summary_clt}
&\sqrt m
\begin{pmatrix}
\hat E_s[
\nabla_\theta\ell(\theta_t;D)
]
\\[1mm]
\hat d_t-\hat d_s
\end{pmatrix}\\
&\overset{d}{\rightarrow}
N
\left(
0,
\begin{pmatrix}
(c_s+\operatorname{Var}(W^{(s)}))
\operatorname{Var}_t
\{
\nabla_\theta\ell(\theta_t;D)
\}
&
-(c_s+\operatorname{Var}(W^{(s)}))
\operatorname{Cov}_t
\{
\nabla_\theta\ell(\theta_t;D),Z(D)
\}
\\[1mm]
-(c_s+\operatorname{Var}(W^{(s)}))
\operatorname{Cov}_t
\{
Z(D),\nabla_\theta\ell(\theta_t;D)
\}
&
(c_s+\operatorname{Var}(W^{(s)})+c_t)
\operatorname{Var}_t\{Z(D)\}
\end{pmatrix}
\right).
\end{align}
 It follows from
\eqref{eqn: linear_adjustment_expansion} and
\eqref{eqn: joint_score_summary_clt} that the limiting covariance of $\hat\theta_{B_m}$, $\Sigma(\hat\theta_{B_m})$ is given by
\begin{align}
\Sigma(\hat\theta_{B_m})
&=
\Sigma(\hat\theta_\mathrm{init})
+
(c_s+\operatorname{Var}(W^{(s)})+c_t)
B
\operatorname{Var}_t\{Z(D)\}
B^\intercal
\nonumber\\
&\quad+
(c_s+\operatorname{Var}(W^{(s)}))
B
\operatorname{Cov}_t
\{
Z(D),\nabla_\theta\ell(\theta_t;D)
\}
(H_t^{-1})^\intercal
\nonumber\\
&\quad+
(c_s+\operatorname{Var}(W^{(s)}))
H_t^{-1}
\operatorname{Cov}_t
\{
\nabla_\theta\ell(\theta_t;D),Z(D)
\}
B^\intercal,
\label{eqn: Sigma_B}
\end{align}
where
$
\Sigma(\hat\theta_\mathrm{init})
=
(c_s+\operatorname{Var}(W^{(s)}))
H_t^{-1}
\operatorname{Var}_t
\{
\nabla_\theta\ell(\theta_t;D)
\}
(H_t^{-1})^\intercal.
$

Because $\beta^{(t)}(\theta_t)$ is the slope coefficient matrix of the population least-squares projection of $\nabla_\theta\ell(\theta_t;D)$ onto $\widetilde{Z}(D)$,
\begin{equation}
\label{eqn: beta_projection_covariance}
\left\{
\beta^{(t)}(\theta_t)
\right\}^\intercal
=
\operatorname{Cov}_t
\{
\nabla_\theta\ell(\theta_t;D),Z(D)
\}
\left[\operatorname{Var}_t\{Z(D)\}\right]^{-1}.
\end{equation}
Using the optimal step-size $
\gamma^*
=
\frac{c_s+\operatorname{Var}(W^{(s)})}
{c_s+\operatorname{Var}(W^{(s)})+c_t}
$ from Theorem \ref{thm: sage_risk}, the limiting matrix in \eqref{eqn: Bm_star_limit} is therefore
\begin{align}
B^*
&=
-
\frac{c_s+\operatorname{Var}(W^{(s)})}
{c_s+\operatorname{Var}(W^{(s)})+c_t}
H_t^{-1}
\left\{
\beta^{(t)}(\theta_t)
\right\}^\intercal
\nonumber\\
&=
-
\frac{c_s+\operatorname{Var}(W^{(s)})}
{c_s+\operatorname{Var}(W^{(s)})+c_t}
H_t^{-1}
\operatorname{Cov}_t
\{
\nabla_\theta\ell(\theta_t;D),Z(D)
\}
\left[\operatorname{Var}_t\{Z(D)\}\right]^{-1}.
\label{eqn: B_star}
\end{align}
In particular,
\begin{equation}
\label{eqn: B_star_identity}
(c_s+\operatorname{Var}(W^{(s)})+c_t)
B^*
\operatorname{Var}_t\{Z(D)\}
=
-
(c_s+\operatorname{Var}(W^{(s)}))
H_t^{-1}
\operatorname{Cov}_t
\{
\nabla_\theta\ell(\theta_t;D),Z(D)
\}.
\end{equation}

Using \eqref{eqn: B_star_identity}, completing the square in
\eqref{eqn: Sigma_B} gives
\begin{align}
\Sigma(\hat\theta_{B_m})
&=
\Sigma(\hat\theta_\mathrm{init})
-
(c_s+\operatorname{Var}(W^{(s)})+c_t)
B^*
\operatorname{Var}_t\{Z(D)\}
(B^*)^\intercal
\nonumber\\
&\quad+
(c_s+\operatorname{Var}(W^{(s)})+c_t)
(B-B^*)
\operatorname{Var}_t\{Z(D)\}
(B-B^*)^\intercal.
\label{eqn: Sigma_B_square}
\end{align}

Since
$B_m^*\overset{p}{\rightarrow}B^*$,
the asymptotic covariance of the SAGE estimator is
\begin{equation}
\label{eqn: Sigma_SAGE_opt}
\Sigma(
\hat\theta_\mathrm{SAGE}(\gamma^*)
)
=
\Sigma(\hat\theta_\mathrm{init})
-
(c_s+\operatorname{Var}(W^{(s)})+c_t)
B^*
\operatorname{Var}_t\{Z(D)\}
(B^*)^\intercal.
\end{equation}
Subtracting \eqref{eqn: Sigma_SAGE_opt} from
\eqref{eqn: Sigma_B_square} yields
\begin{align}
\Sigma(\hat\theta_{B_m})
-
\Sigma(
\hat\theta_\mathrm{SAGE}(\gamma^*)
)=
(c_s+\operatorname{Var}(W^{(s)})+c_t)
(B-B^*)
\operatorname{Var}_t\{Z(D)\}
(B-B^*)^\intercal\succeq
0.
\end{align}
The final inequality follows because
$c_s+\operatorname{Var}(W^{(s)})+c_t>0$ and
$\operatorname{Var}_t\{Z(D)\}$ is positive definite under
Assumption~\eqref{assump: var(z) pd}. 
\end{proof}

\subsection{Proof of Lemma~\ref{lemma: entropy_bal}}\label{app: entropy_bal}

\begin{proof}
We first rewrite the entropy-balancing weights optimization problem as
$$
\mathcal L(w,\lambda_1,\lambda_2)
=
\sum_{i=1}^{n_s}w_i\log(n_sw_i)
+
\lambda_1\left(\sum_{i=1}^{n_s}w_i-1\right)
-
\lambda_2^\intercal
\left(
\sum_{i=1}^{n_s}w_iZ(D_i^{(s)})-\hat d_t
\right).
$$ Differentiating with respect to $w_i$ gives
$$
0=\frac{\partial\mathcal L}{\partial w_i}
=
\log(n_sw_i)+1+\lambda_1
-
\lambda_2^\intercal Z(D_i^{(s)}).
$$
Solving and rearranging, we get 
$$
\hat w_i
=
\frac{\exp(-1-\hat\lambda_1)}{n_s}
\exp\left(
\hat\lambda_2^\intercal Z(D_i^{(s)})
\right),
$$
where $\hat\lambda_1, \hat\lambda_2$ are chosen to satisfy the constraints.
Since the weights must sum to $1$, we have $n_s^{-1}\exp(-1-\hat\lambda_1)=\left\{\sum_{j=1}^{n_s}\exp\left(\hat\lambda_2^\intercal Z(D_j^{(s)})\right)\right\}^{-1}.$
Therefore,
\begin{equation}\label{eqn: weights_eb}
    \hat w_i
=
\frac{
\exp\left(
\hat\lambda_2^\intercal Z(D_i^{(s)})
\right)
}{
\sum_{j=1}^{n_s}
\exp\left(
\hat\lambda_2^\intercal Z(D_j^{(s)})
\right)
}=
\frac{
\exp\left(
\hat\lambda_2^\intercal\{Z(D_i^{(s)})-\hat d_s\}
\right)
}{
\sum_{j=1}^{n_s}
\exp\left(
\hat\lambda_2^\intercal\{Z(D_j^{(s)})-\hat d_s\}
\right)
},
\end{equation}
where the second equality comes from multiplying the numerator and denominator by $\exp(-\hat\lambda_2^\intercal\hat d_s)$. 

We now apply results from Prop \ref{prop: lamdba_op1} to verify an assumption required for the consistency result of $\hat{\theta}_{\rm EB}$ from Prop \ref{prop: eb_consist}: multiplying both sides of Eqn \eqref{eqn: weights_eb} by $n_s$, subtracting 1, and taking the max over $i$ in absolute values gives
\begin{align}\label{eqn: max_nw}
\begin{split}
    \max_i |n_s\hat w_i-1|
&\leq \frac{
2\max_i
\left|
\exp\left(
\hat\lambda_2^\intercal\{Z(D_i^{(s)})-\hat d_s\}
\right)-1
\right|
}{
1+O_p(m^{-1})} \qquad \text{(Eqn \eqref{eqn: denom_lambda})}\\
&=  \frac{o_p(1)}{1+O_p(m^{-1})} \qquad \text{(Eqn \eqref{eqn: max_lambda_Z})}\\
&=o_p(1). 
\end{split}
\end{align}
Therefore, by Prop \ref{prop: eb_consist}, $\hat{\theta}_{\rm EB} \overset{p}{\rightarrow} \theta_t$. Combining with the assumption that $\theta_t\in\operatorname{int}(\Theta)$, $\hat\theta_{\mathrm{EB}}$ is an interior minimizer. 

We can now proceed with a Taylor expansion around $\theta_t$.
\begin{align}\label{eqn: eb_taylor_main}
0
&=
\sum_{i=1}^{n_s}
\hat w_i
\nabla_\theta\ell(\hat\theta_{\mathrm{EB}};D_i^{(s)})
=
\sum_{i=1}^{n_s}
\hat w_i\nabla_\theta\ell(\theta_t;D_i^{(s)})
+
\left[
\sum_{i=1}^{n_s}
\hat w_i\nabla_\theta^2\ell(\theta_t;D_i^{(s)})
\right]
(\hat\theta_{\mathrm{EB}}-\theta_t)
+
\mathcal R_m,
\end{align}
where the remainder term satisfies
\begin{align*}
\|\mathcal R_m\|_2&=\left\|\sum_{i=1}^{n_s}\hat w_i\left\{\nabla_\theta\ell(\hat\theta_{\mathrm{EB}};D_i^{(s)})-\nabla_\theta\ell(\theta_t;D_i^{(s)})-\nabla_\theta^2\ell(\theta_t;D_i^{(s)})(\hat\theta_{\mathrm{EB}}-\theta_t)\right\}\right\|_2\\
&\leq \sum_{i=1}^{n_s}\hat w_i\left\|\nabla_\theta\ell(\hat\theta_{\mathrm{EB}};D_i^{(s)})-\nabla_\theta\ell(\theta_t;D_i^{(s)})-\nabla_\theta^2\ell(\theta_t;D_i^{(s)})(\hat\theta_{\mathrm{EB}}-\theta_t)\right\|_2\\
&\le \frac12 \|\hat\theta_{\mathrm{EB}}-\theta_t\|_2^2
\frac{1}{n_s}\sum_{i=1}^{n_s}n_s \hat w_i M_2(D_i^{(s)})  \qquad \text{(Assumption \eqref{assump: lip})}\\
&\le \frac12 \|\hat\theta_{\mathrm{EB}}-\theta_t\|_2^2 \left\{1+\max_i|n_s\hat w_i-1|\right\} \underbrace{\frac1{n_s}\sum_{i=1}^{n_s}M_2(D_i^{(s)})}_{\overset{p}{\rightarrow} E_t[M_2(D)]}
\\
&=
O_p(\|\hat\theta_{\mathrm{EB}}-\theta_t\|_2^2)(1+o_p(1))O_p(1) \qquad \text{(Prop \ref{prop: eb_consist}, Eqn \eqref{eqn: max_nw})}
\\
&=
o_p(\|\hat\theta_{\mathrm{EB}}-\theta_t\|_2),
\end{align*}
where the last inequality uses that $n_s\hat w_i \leq 1+ |n_s \hat w_i-1|\leq 1+ \max_j|n_s \hat w_j-1|$.

We next evaluate the first term of the Taylor expansion.  Substituting Eqn.~\eqref{eqn: weights_eb} gives
$$
\sum_{i=1}^{n_s}
\hat w_i\nabla_\theta\ell(\theta_t;D_i^{(s)})
=
\frac{
\frac1{n_s}\sum_{i=1}^{n_s}
\exp\left(
\hat\lambda_2^\intercal\{Z(D_i^{(s)})-\hat d_s\}
\right)
\nabla_\theta\ell(\theta_t;D_i^{(s)})
}{
\frac1{n_s}\sum_{i=1}^{n_s}
\exp\left(
\hat\lambda_2^\intercal\{Z(D_i^{(s)})-\hat d_s\}
\right)
}.
$$
Since Eqn \eqref{eqn: denom_lambda} already tells us that the denominator is $1+O_p(m^{-1})$, it remains to evaluate the numerator. A Taylor expansion of the exponential gives
\begin{align*}
\frac1{n_s}\sum_{i=1}^{n_s}
\exp\left(
\hat\lambda_2^\intercal\{Z(D_i^{(s)})-\hat d_s\}
\right)
\nabla_\theta\ell(\theta_t;D_i^{(s)})
&=
\frac1{n_s}\sum_{i=1}^{n_s}
\left[
1+
\hat\lambda_2^\intercal\{Z(D_i^{(s)})-\hat d_s\}
+r_{i,m}
\right]
\nabla_\theta\ell(\theta_t;D_i^{(s)})
\\
&=
\frac1{n_s}\sum_{i=1}^{n_s}
\nabla_\theta\ell(\theta_t;D_i^{(s)})
+
\left[
\frac1{n_s}\sum_{i=1}^{n_s}
\nabla_\theta\ell(\theta_t;D_i^{(s)})
\{Z(D_i^{(s)})-\hat d_s\}^\intercal
\right]\hat\lambda_2
\\
&\qquad+
\frac1{n_s}\sum_{i=1}^{n_s}
r_{i,m}\nabla_\theta\ell(\theta_t;D_i^{(s)}).
\end{align*}
Applying the inequality $|e^x-1-x|\le \frac12e^{|x|}x^2$ to the remainder term, 
\begin{align*}
\max_i|r_{i,m}|
&\le
\frac12
\exp\left\{
\max_i\left|
\hat\lambda_2^\intercal\{Z(D_i^{(s)})-\hat d_s\}
\right|
\right\}
\left[
\max_i\left|
\hat\lambda_2^\intercal\{Z(D_i^{(s)})-\hat d_s\}
\right|
\right]^2
\\
&=
O_p(1)
O_p\left(\frac{(\log m)^2}{m}\right)
\\
&=
O_p\left(\frac{(\log m)^2}{m}\right).
\end{align*}
Therefore, 
\begin{align*}
\left\|
\frac1{n_s}\sum_{i=1}^{n_s}
r_{i,m}\nabla_\theta\ell(\theta_t;D_i^{(s)})
\right\|_2
&\le
\max_i|r_{i,m}|
\frac1{n_s}\sum_{i=1}^{n_s}
\|\nabla_\theta\ell(\theta_t;D_i^{(s)})\|_2
\\
&\le
\max_i|r_{i,m}|\frac1{n_s}\sum_{i=1}^{n_s}F_1(D_i^{(s)}) \qquad \text{(Assump.\eqref{assump: sup_op_bd}, Lemma \ref{lemma: distr_clt})}\\
&=O_p\left(\frac{(\log m)^2}{m}\right) O_p(1)\\
&=o_p(m^{-1/2}).
\end{align*}

Combining the above with results of Lemma \ref{lemma: distr_clt}, $\frac1{n_s}\sum_{i=1}^{n_s}
\nabla_\theta\ell(\theta_t;D_i^{(s)}) = E_t[\nabla_\theta\ell(\theta_t;D)] + O_p(m^{-1/2})=O_p(m^{-1/2})$ and $\frac1{n_s}\sum_{i=1}^{n_s} \nabla_\theta\ell(\theta_t;D_i^{(s)})\{Z(D_i^{(s)})-\hat d_s\}^\intercal=\operatorname{Cov}_t\{\nabla_\theta\ell(\theta_t;D),Z(D)\}+O_p(m^{-1/2}),$ we get
\begin{align*}
\sum_{i=1}^{n_s}
\hat w_i\nabla_\theta\ell(\theta_t;D_i^{(s)})
&=
\frac{
\frac1{n_s}\sum_{i=1}^{n_s}
\nabla_\theta\ell(\theta_t;D_i^{(s)})
+
\left[
\frac1{n_s}\sum_{i=1}^{n_s}
\nabla_\theta\ell(\theta_t;D_i^{(s)})
\{Z(D_i^{(s)})-\hat d_s\}^\intercal
\right]\hat\lambda_2
+
o_p(m^{-1/2})
}{
1+O_p(m^{-1})
}\\
&=\frac1{n_s}\sum_{i=1}^{n_s}
\nabla_\theta\ell(\theta_t;D_i^{(s)})
+\left[\operatorname{Cov}_t\{\nabla_\theta\ell(\theta_t;D),Z(D)\}\right]\hat\lambda_2 +o_p(m^{-1/2})
\end{align*}

We next evaluate the second term in the Taylor expansion of Eqn \eqref{eqn: eb_taylor_main}:
\begin{align*}
\left[
\sum_{i=1}^{n_s}
\hat w_i\nabla_\theta^2\ell(\theta_t;D_i^{(s)})
\right]
(\hat\theta_{\mathrm{EB}}-\theta_t)
&=
\left[
\frac1{n_s}\sum_{i=1}^{n_s} n_s\hat w_i \nabla_\theta^2\ell(\theta_t;D_i^{(s)})
\right]
(\hat\theta_{\mathrm{EB}}-\theta_t)
\\
&=
\left[
\frac1{n_s}\sum_{i=1}^{n_s}
\{1+(n_s\hat w_i-1)\}
\nabla_\theta^2\ell(\theta_t;D_i^{(s)})
\right]
(\hat\theta_{\mathrm{EB}}-\theta_t)
\\
&=
\left[
\frac1{n_s}\sum_{i=1}^{n_s}
\nabla_\theta^2\ell(\theta_t;D_i^{(s)})
\right]
(\hat\theta_{\mathrm{EB}}-\theta_t) +
\left[
\frac1{n_s}\sum_{i=1}^{n_s}
(n_s\hat w_i-1)
\nabla_\theta^2\ell(\theta_t;D_i^{(s)})
\right]
(\hat\theta_{\mathrm{EB}}-\theta_t).
\end{align*}

By Assumption \eqref{assump: sup_op_bd}, Eqn \eqref{eqn: max_nw} and Lemma \ref{lemma: distr_clt},
\begin{align*}
\left\|
\left[
\frac1{n_s}\sum_{i=1}^{n_s}
(n_s\hat w_i-1)
\nabla_\theta^2\ell(\theta_t;D_i^{(s)})
\right]
(\hat\theta_{\mathrm{EB}}-\theta_t)
\right\|_2
&\le
\left[
\frac1{n_s}\sum_{i=1}^{n_s}
|n_s\hat w_i-1|
\|\nabla_\theta^2\ell(\theta_t;D_i^{(s)})\|_{\mathrm{op}}
\right]
\|\hat\theta_{\mathrm{EB}}-\theta_t\|_2
\\
&\le
\max_i|n_s\hat w_i-1|
\frac1{n_s}\sum_{i=1}^{n_s}F_2(D_i^{(s)})
\|\hat\theta_{\mathrm{EB}}-\theta_t\|_2
\\
&=
o_p(1)O_p(1)\|\hat\theta_{\mathrm{EB}}-\theta_t\|_2
\\
&=
o_p(\|\hat\theta_{\mathrm{EB}}-\theta_t\|_2).
\end{align*}
Also, by Lemma \ref{lemma: distr_clt}, $\frac1{n_s}\sum_{i=1}^{n_s}
\nabla_\theta^2\ell(\theta_t;D_i^{(s)})=E_t[\nabla_\theta^2\ell(\theta_t;D)]+O_p(m^{-1/2}).$
Therefore,
\begin{align*}
\left[
\sum_{i=1}^{n_s}
\hat w_i\nabla_\theta^2\ell(\theta_t;D_i^{(s)})
\right]
(\hat\theta_{\mathrm{EB}}-\theta_t)
&=
\bigg( \underbrace{E_t[\nabla_\theta^2\ell(\theta_t;D)]}_{H_t} + O_p(m^{-1/2}) \bigg)
\underbrace{(\hat\theta_{\mathrm{EB}}-\theta_t)}_{o_p(1)}
+
o_p(\|\hat\theta_{\mathrm{EB}}-\theta_t\|_2) \\
&=H_t(\hat\theta_{\mathrm{EB}}-\theta_t)+o_p(\|\hat\theta_{\mathrm{EB}}-\theta_t\|_2) +o_p(m^{-1/2})
\end{align*}
Combining the two pieces, we have
\begin{align*}
0=
\sum_{i=1}^{n_s}
\hat w_i
\nabla_\theta\ell(\hat\theta_{\mathrm{EB}};D_i^{(s)})&=\frac1{n_s}\sum_{i=1}^{n_s}
\nabla_\theta\ell(\theta_t;D_i^{(s)})\\
&\quad +\left[
\operatorname{Cov}_t\{
\nabla_\theta\ell(\theta_t;D),Z(D)\}
\right]\hat\lambda_2+
H_t(\hat\theta_{\mathrm{EB}}-\theta_t)
+
o_p(m^{-1/2})
+
o_p(\|\hat\theta_{\mathrm{EB}}-\theta_t\|_2).
\end{align*}
Rearranging and using that $\hat\lambda_2=O_p(m^{-1/2})$ by Prop \ref{prop: lamdba_op1},
\begin{align*}
\hat\theta_{\mathrm{EB}}-\theta_t
&=
-H_t^{-1}
\bigg\{
\frac1{n_s}\sum_{i=1}^{n_s}
\nabla_\theta\ell(\theta_t;D_i^{(s)})+
\operatorname{Cov}_t\{
\nabla_\theta\ell(\theta_t;D),Z(D)\}
\hat\lambda_2
+
o_p(m^{-1/2})
\bigg\}
+
o_p(\|\hat\theta_{\mathrm{EB}}-\theta_t\|_2)
\\
&=O_p(m^{-1/2})+o_p(\|\hat\theta_{\mathrm{EB}}-\theta_t\|_2).
\end{align*}
Taking norms, there is a nonnegative
$\varepsilon_m=o_p(1)$ such that
\begin{align*}
    &\|\hat\theta_{\mathrm{EB}}-\theta_t\|_2 \le O_p(m^{-1/2})+ \varepsilon_m\|\hat\theta_{\mathrm{EB}}-\theta_t\|_2\\
    &\Leftrightarrow (1-\varepsilon_m) \|\hat\theta_{\mathrm{EB}}-\theta_t\|_2\le O_p(m^{-1/2})
\end{align*}
Since $P(\varepsilon_m\le1/2)\to1$, with probability tending to one, as $m\rightarrow \infty$
$$
\frac12\|\hat\theta_{\mathrm{EB}}-\theta_t\|_2\le(1-\varepsilon_m)\|\hat\theta_{\mathrm{EB}}-\theta_t\|_2
\le O_p(m^{-1/2}).
$$
Therefore,
\begin{equation}\label{eqn: theta_eb_rate}
\hat\theta_{\mathrm{EB}}-\theta_t
=
O_p(m^{-1/2}),
\end{equation}
and
\begin{align*}
\hat\theta_{\mathrm{EB}}-\theta_t
&=
-H_t^{-1} \bigg(\frac1{n_s}\sum_{i=1}^{n_s}
\nabla_\theta\ell(\theta_t;D_i^{(s)})+\operatorname{Cov}_t\{\nabla_\theta\ell(\theta_t;D),Z(D)\}\hat\lambda_2\bigg) +
o_p(m^{-1/2}).
\end{align*}
Finally, substituting the linearization of $\hat\lambda_2=\left[\widehat{\operatorname{Var}}_s\{Z(D)\}
\right]^{-1}(\hat d_t-\hat d_s)+o_p(m^{-1/2})$,
\begin{align}\label{eqn: eb_linearization}
\hat\theta_{\mathrm{EB}}-\theta_t
&=
-H_t^{-1}\bigg(\frac1{n_s}\sum_{i=1}^{n_s}
\nabla_\theta\ell(\theta_t;D_i^{(s)})+
\operatorname{Cov}_t\{
\nabla_\theta\ell(\theta_t;D),Z(D)\}
\left[\operatorname{Var}_t\{Z(D)\}\right]^{-1}
(\hat d_t-\hat d_s)\bigg) + o_p(m^{-1/2}).
\end{align}

We now compare this expansion with SAGE at $\gamma=1$. First, we apply the definition of SAGE and the linearization of $\hat\theta_{\rm init}-\theta_t$ from Prop \ref{prop: theta_old_consist} to get
$$\hat\theta_{\mathrm{SAGE}}(1)-\theta_t = \hat\theta_{\mathrm{init}}-\widehat H_s^{-1}\widehat\Psi_t
-\theta_t =-H_t^{-1}\frac1{n_s}\sum_{i=1}^{n_s}\nabla_\theta\ell(\theta_t;D_i^{(s)})-
\widehat H_s^{-1}\widehat\Psi_t+o_p(m^{-1/2}).$$ Then, by Lemma \ref{lemma: grad_consist} and consistency of the sample Hessian from Prop \ref{prop: hess_consist} and OLS sample coefficients,
\begin{align*}
\hat\theta_{\mathrm{SAGE}}(1)-\theta_t
&=
-H_t^{-1}
\frac1{n_s}\sum_{i=1}^{n_s}
\nabla_\theta\ell(\theta_t;D_i^{(s)})
-
\widehat H_s^{-1}
\left\{
\hat\beta^{(s)}(\hat\theta_{\mathrm{init}})
\right\}^\intercal
(\hat d_t-\hat d_s)
+
o_p(m^{-1/2})
\\
&=
-H_t^{-1}
\frac1{n_s}\sum_{i=1}^{n_s}
\nabla_\theta\ell(\theta_t;D_i^{(s)})-
\{H_t^{-1}+o_p(1)\}
\left[
\{\beta^{(t)}(\theta_t)\}^\intercal+o_p(1)
\right]
(\hat d_t-\hat d_s)
+
o_p(m^{-1/2})
\\
&=
-H_t^{-1}
\frac1{n_s}\sum_{i=1}^{n_s}
\nabla_\theta\ell(\theta_t;D_i^{(s)})
-
H_t^{-1}
\{\beta^{(t)}(\theta_t)\}^\intercal
(\hat d_t-\hat d_s)
+
o_p(m^{-1/2})
\\
&=
-H_t^{-1}
\frac1{n_s}\sum_{i=1}^{n_s}
\nabla_\theta\ell(\theta_t;D_i^{(s)})-
H_t^{-1}
\operatorname{Cov}_t\{
\nabla_\theta\ell(\theta_t;D),Z(D)\}
\left[\operatorname{Var}_t\{Z(D)\}\right]^{-1}
(\hat d_t-\hat d_s)
+
o_p(m^{-1/2})
\\
&=
-H_t^{-1}\Bigg(
\frac1{n_s}\sum_{i=1}^{n_s}
\nabla_\theta\ell(\theta_t;D_i^{(s)})+
\operatorname{Cov}_t\{
\nabla_\theta\ell(\theta_t;D),Z(D)\}
\left[\operatorname{Var}_t\{Z(D)\}\right]^{-1}
(\hat d_t-\hat d_s)
\Bigg)
+
o_p(m^{-1/2}).
\end{align*}

Comparing with Eqn \ref{eqn: eb_linearization}, we get $\hat\theta_{\mathrm{EB}}-\hat\theta_{\mathrm{SAGE}}(1)=o_p(m^{-1/2})$, i.e.
$$\sqrt m\left\{\hat\theta_{\mathrm{EB}}-\hat\theta_{\mathrm{SAGE}}(1)\right\} \overset p\rightarrow0.$$
Therefore, the centered and scaled estimators have the same limiting distribution i.e.
$$\Sigma(\hat\theta_{\mathrm{EB}})=\Sigma(\hat\theta_{\mathrm{SAGE}}(1)).$$ Applying the asymptotic covariance optimality of SAGE at $\gamma^*$ from Theorem \ref{thm: var_opt} therefore gives
$$
\Sigma(\hat\theta_{\mathrm{EB}})
-
\Sigma(\hat\theta_{\mathrm{SAGE}}(\gamma^*))
=
\Sigma(\hat\theta_{\mathrm{SAGE}}(1))
-
\Sigma(\hat\theta_{\mathrm{SAGE}}(\gamma^*))
\succeq0.
$$
\end{proof}

\subsection{Proof of Lemma~\ref{lem: sage_gamma}}\label{app: gamma_consist}
\begin{proof}
Fix $L$, and let $r=\dim Z(D)$ denote the number of summary
coordinates retained under Assumption \eqref{assump: var(z) pd}.
It suffices to consider $L$ sufficiently large that
$r\geq1$. Let $Z(D)\in\mathbb R^r$ denote the summary
vector, and let $\hat d_s$ and $\hat d_t$ denote the source and target summary vectors. The fourth moment assumption gives $E_t\|\widetilde Z(D)\|_2^4<\infty$.
Since each summary coordinate is a fixed linear combination
of the regressors in $\widetilde Z(D)$, it follows that
$E_t\|Z(D)\|_2^4<\infty$ for each fixed $L$. Thus both the coordinates of $Z(D)$ and the entries
of $Z(D)Z(D)^\intercal$ are square-integrable and we can apply the distributional CLT. 

Therefore, by Lemma \ref{lemma: distr_clt},
$\hat E_s[Z(D)]
\overset{p}{\rightarrow}E_t[Z(D)]$ and $
\hat E_s[Z(D)Z(D)^\intercal]
\overset{p}{\rightarrow}
E_t[Z(D)Z(D)^\intercal]$.
Consequently,
$
\widehat S_s
\coloneqq
\hat E_s[Z(D)Z(D)^\intercal]
-
\hat d_s\hat d_s^\intercal
\overset{p}{\rightarrow}
\operatorname{Var}_t(Z(D)).
$
Since the population covariance is positive definite,
$\widehat S_s$ is invertible with probability
tending to one, and continuity of matrix inversion gives
$\widehat S_s^{-1}
\overset{p}{\rightarrow}
\left\{\operatorname{Var}_t(Z(D))\right\}^{-1}.$

By Lemma~\ref{lemma: distr_clt}, applied jointly to
the summary coordinates,
\[
m^{1/2}(\hat d_s-\hat d_t)
\overset{d}{\rightarrow}
N\left(
0,
(c_s+c_t+\operatorname{Var}(W^{(s)}))\operatorname{Var}_t\{Z(D)\}
\right).
\]
Since $m/n_s+m/n_t\to c_s+c_t$, Slutsky's theorem yields
\[
\frac{\hat d_s-\hat d_t}{\sqrt{1/n_s+1/n_t}}
\overset{d}{\rightarrow}
N\left(
0,
\left(1+\frac{\operatorname{Var}(W^{(s)})}{c_s+c_t}\right)
\operatorname{Var}_t\{Z(D)\}
\right).
\]

The norm of the whitened standardized summary difference computed in Step \ref{alg: step_z} of Algorithm \ref{alg: sage_gamma} can be written as 
$$
\|z\|_2^2
=
\frac{
(\hat d_s-\hat d_t)^\intercal
\widehat S_s^{-1}
(\hat d_s-\hat d_t)
}{
1/n_s+1/n_t
}.
$$
The continuous mapping theorem and Slutsky's theorem
therefore give, for fixed $L$ and $m\to\infty$,
\[
\frac{\|z\|_2^2}{r}
\overset{d}{\rightarrow}
\left(1+\frac{\operatorname{Var}(W^{(s)})}{c_s+c_t}\right)
\frac{\chi^2_r}{r},
\]
where $\chi^2_r$ denotes a chi-square random variable
with $r$ degrees of freedom.
It follows that for $m\rightarrow \infty$
$$
m\widehat{\delta^2_{\rm dist}}
=
\left(\frac{m}{n_s}+\frac{m}{n_t}\right)
\left(\frac{\|z\|_2^2}{r}-1\right)_+
\overset{d}{\rightarrow}
\left[
(c_s+c_t+\operatorname{Var}(W^{(s)}))\frac{\chi^2_r}{r}
-(c_s+c_t)
\right]_+.
$$

We now take the outer limit $L\to\infty$.
By assumption, $r\to\infty$, and
$$
E\left[\frac{\chi^2_r}{r}\right]=1,
\qquad
\operatorname{Var}\left(\frac{\chi^2_r}{r}\right)
=\frac{2}{r}\rightarrow0.
$$
Hence $\chi^2_r/r\overset{p}{\to}1$, and
continuity of the positive part function gives
\[
\left[
(c_s+c_t+\operatorname{Var}(W^{(s)}))\frac{\chi^2_r}{r}
-(c_s+c_t)
\right]_+
\overset{p}{\rightarrow}\operatorname{Var}(W^{(s)}).
\]
Combining the two limits and applying Portmanteau yields, for every $\varepsilon>0$,
$$
\lim_{L\to\infty}\limsup_{m\to\infty}
P\left(
\left|m\hat{\delta}^2_{\rm dist}-\operatorname{Var}(W^{(s)})\right|
>\varepsilon
\right)=0.
$$

Finally, for
$$
\hat\gamma
=
\frac{
m/n_s+m\hat{\delta}^2_{\rm dist}
}{
m/n_s+m\hat{\delta}^2_{\rm dist}+m/n_t
},
$$
applying the continuous mapping theorem gives
$$
\lim_{L\to\infty}\limsup_{m\to\infty}
P\left(
\left|
\hat\gamma-
\frac{c_s+\operatorname{Var}(W^{(s)})}{c_s+c_t+\operatorname{Var}(W^{(s)})}
\right|>\varepsilon
\right)=0.
$$
\end{proof}

\subsection{Proof of Corollary~\ref{cor: sage_normality}}\label{app: asymp_norm}

\begin{proof}
Throughout the proof, the influence functions and their plug-in estimators are evaluated at $\gamma$,
with this dependence suppressed (i.e $\Phi_k$ and $\widehat\Phi_k$ represent $\Phi_k(\gamma)$ and $\widehat\Phi_k(\gamma)$ respectively.). 

By Corollary~\ref{corr: grad_linear}, we have the linearization
\begin{align}
\label{eqn: sage_ci_linearization}
\sqrt m
(\hat\theta_\mathrm{SAGE}(\gamma)-\theta_t)
=
-H_t^{-1}
\left[
\sqrt m\,
\hat E_s\left[
\nabla_\theta\ell(\theta_t;D)
-
\gamma\Xi_s(D)
\right]
+
\gamma\sqrt m\,
\hat E_t[\Xi_t(D)]
\right]
+
o_p(1).
\end{align}
Define the $p$-dimensional source and target influence functions
\begin{align*}
\Phi_s(D) \coloneqq
-H_t^{-1}
\left\{
\nabla_\theta\ell(\theta_t;D)
-
\gamma\Xi_s(D)
\right\}, \qquad
\Phi_t(D) \coloneqq
\gamma H_t^{-1}\Xi_t(D).
\end{align*}
Since $\theta_t$ is the target risk minimizer,
$E_t[
\nabla_\theta\ell(\theta_t;D)
]
=
0$.
Moreover, by the definition of $\Xi_k(D)$,
$E_t[\Xi_s(D)]
=
E_t[\Xi_t(D)]
=
0$.
It follows that $E_t[\Phi_s(D)]=E_t[\Phi_t(D)]=0.$ We have that
\begin{align}
\label{eqn: sage_influence_representation}
\begin{split}
\sqrt m
(\hat\theta_\mathrm{SAGE}(\gamma)-\theta_t)
&=
\sqrt m
\left[
\left\{
\hat E_s[\Phi_s(D)]-E_t[\Phi_s(D)]
\right\}
-
\left\{
\hat E_t[\Phi_t(D)]-E_t[\Phi_t(D)]
\right\}
\right]
+
o_p(1)
\\
&=
\sqrt m
\left[
\big\{
\hat E_s[\Phi_s(D)]-\hat E_t[\Phi_t(D)]
\big\}
-
\big\{
E_t[\Phi_s(D)]-E_t[\Phi_t(D)]
\big\}
\right]
+
o_p(1).
\end{split}
\end{align}

By Lemma~\ref{lemma: distr_clt}, for each fixed $L$,
\begin{equation}
\label{eqn: joint_Phi_clt}
\sqrt m
\begin{pmatrix}
\hat E_t[\Phi_t(D)]-E_t[\Phi_t(D)]
\\
\hat E_s[\Phi_s(D)]-E_t[\Phi_s(D)]
\end{pmatrix}
\overset{d}{\rightarrow}
N
\left(
0,
\begin{pmatrix}
c_t\operatorname{Var}_t\{\Phi_t(D)\}
&
0_{p\times p}
\\
0_{p\times p}
&
(c_s+\operatorname{Var}(W^{(s)}))
\operatorname{Var}_t\{\Phi_s(D)\}
\end{pmatrix}
\right).
\end{equation}
Applying the continuous mapping theorem and Eqn \eqref{eqn: sage_influence_representation} and setting $\gamma=\gamma^*$ gives
\begin{equation}
\label{eqn: sage_ci_clt}
\sqrt m
(\hat\theta_\mathrm{SAGE}(\gamma^*)-\theta_t)
\overset{d}{\rightarrow}
N(0,\Sigma_{\mathrm{SAGE}}(\gamma^*)),
\end{equation}
where $\Sigma_{\mathrm{SAGE}}(\gamma^*)
\coloneqq
c_t\operatorname{Var}_t\{\Phi_t(D)\}
+
(c_s+\operatorname{Var}(W^{(s)}))
\operatorname{Var}_t\{\Phi_s(D)\}.$
\end{proof}

\subsection{Proof of Lemma~\ref{lem: CI}}\label{app: CI}

\begin{proof}
Since this proof takes limits in $L$, we write the dependence on $L$ explicitly where needed, i.e., $\Phi_{k,L}$ and $\Sigma_{\mathrm{SAGE},L}(\gamma^*)$ denote $\Phi_k$ and $\Sigma_{\mathrm{SAGE}}(\gamma^*)$ from Corollary~\ref{cor: sage_normality} for a given $L$.
Let $\widehat\Phi_s$ and $\widehat\Phi_t$ denote the plug-in
estimators of $\Phi_s$ and $\Phi_t$, formed by replacing the population
quantities with their empirical analogues. For each fixed $L$, consistency of the plug-in quantities,
together with the continuity and moment conditions in
Appendix~\ref{app: assumptions}, gives for $k\in\{s,t\}$,
$
\hat E_s\left[
\|\widehat\Phi_k(D)-\Phi_k(D)\|_2^2
\right]
\overset{p}{\rightarrow}0.
$
Under the fourth moment assumption, both $\Phi_k(D)$ and the entries of
$\Phi_k(D)\Phi_k(D)^\intercal$ are square-integrable.

Thus, applying Lemma~\ref{lemma: distr_clt} and using Cauchy-Schwarz gives
\begin{align*}
\hat E_s[\widehat\Phi_k(D)]
&=
\hat E_s[\Phi_k(D)]+o_p(1)
\overset{p}{\rightarrow}
E_t[\Phi_k(D)],
\\
\hat E_s[
\widehat\Phi_k(D)\widehat\Phi_k(D)^\intercal
]
&=
\hat E_s[
\Phi_k(D)\Phi_k(D)^\intercal
]+o_p(1)
\overset{p}{\rightarrow}
E_t[
\Phi_k(D)\Phi_k(D)^\intercal
].
\end{align*}
Following the empirical covariance estimation of
\cite{random_shift1,random_shift2}, let
\begin{align*}
\widehat{\operatorname{Var}}_s
\left\{
\widehat\Phi_k(D_i^{(s)})
\right\}
\coloneqq{}&
\hat E_s[
\widehat\Phi_k(D)\widehat\Phi_k(D)^\intercal
]
-
\hat E_s[\widehat\Phi_k(D)]
\hat E_s[\widehat\Phi_k(D)]^\intercal.
\end{align*}
Therefore,
\begin{equation}\label{eqn: Phi_var_consistency}
\widehat{\operatorname{Var}}_s
\left\{
\widehat\Phi_s(D_i^{(s)})
\right\}
\overset{p}{\rightarrow}
\operatorname{Var}_t\{\Phi_s(D)\}, \qquad
\widehat{\operatorname{Var}}_s
\left\{
\widehat\Phi_t(D_i^{(s)})
\right\}
\overset{p}{\rightarrow}
\operatorname{Var}_t\{\Phi_t(D)\}.
\end{equation}
Let $\hat{\delta}^2_{\rm dist}$ be as defined in
Eqn \eqref{eqn: delta_hat}. Define the covariance estimator
at $\gamma^*$ by
\begin{equation*}
\widehat{V}_{\rm SAGE}(\gamma^*)
\coloneqq
\frac{1}{n_t}
\widehat{\operatorname{Var}}_s
\left\{
\widehat\Phi_t(D_i^{(s)})
\right\}
+
\left(
\frac{1}{n_s}+\hat{\delta}^2_{\rm dist}
\right)
\widehat{\operatorname{Var}}_s
\left\{
\widehat\Phi_s(D_i^{(s)})
\right\}.
\end{equation*}
We now establish coordinatewise variance consistency at $\gamma^*$. For the remainder of the proof, we display the dependency on $L$ , which was previously suppressed for notational simplicity. For a fixed coordinate $j$, let $e_j$ denote the corresponding
standard basis vector. Assume that
$e_j^\intercal\Sigma_{\mathrm{SAGE},L}(\gamma^*)e_j>0$
for all sufficiently large $L$ to write
\begin{align}\label{eqn: coordinate_var_gamma_star}
m e_j^\intercal\widehat{V}_{\rm SAGE}(\gamma^*)e_j
&=
\frac{m}{n_t}
\operatorname{Var}_t
\{e_j^\intercal\Phi_{t,L}(D)\}+
\left(\frac{m}{n_s}+m\hat\delta^2_{\rm dist}\right)
\operatorname{Var}_t
\{e_j^\intercal\Phi_{s,L}(D)\}
+o_p(1)\\
&=
c_t\operatorname{Var}_t
\{e_j^\intercal\Phi_{t,L}(D)\}+
(c_s+m\hat\delta^2_{\rm dist})
\operatorname{Var}_t
\{e_j^\intercal\Phi_{s,L}(D)\}
+o_p(1).
\end{align}
By definition,
$e_j^\intercal\Sigma_{\mathrm{SAGE},L}(\gamma^*)e_j
=c_t\operatorname{Var}_t
\{e_j^\intercal\Phi_{t,L}(D)\}+
(c_s+\operatorname{Var}(W^{(s)}))
\operatorname{Var}_t
\{e_j^\intercal\Phi_{s,L}(D)\},$ so subtracting this from Eqn \eqref{eqn: coordinate_var_gamma_star} gives
\begin{align*}
m e_j^\intercal\widehat{V}_{\rm SAGE}(\gamma^*)e_j
-
e_j^\intercal\Sigma_{\mathrm{SAGE},L}(\gamma^*)e_j=
(m\hat\delta^2_{\rm dist}-\operatorname{Var}(W^{(s)}))
\operatorname{Var}_t
\{e_j^\intercal\Phi_{s,L}(D)\}
+o_p(1).
\end{align*}
Dividing by $e_j^\intercal\Sigma_{\mathrm{SAGE},L}(\gamma^*)e_j>0$ gives
$$
\frac{
m e_j^\intercal \widehat{V}_{\rm SAGE}(\gamma^*)e_j
}{
e_j^\intercal\Sigma_{\mathrm{SAGE},L}(\gamma^*)e_j
}-1
=
\frac{
\operatorname{Var}_t\{e_j^\intercal\Phi_{s,L}(D)\}
}{
e_j^\intercal\Sigma_{\mathrm{SAGE},L}(\gamma^*)e_j
}
(m\hat\delta^2_{\rm dist}-\operatorname{Var}(W^{(s)}))
+o_p(1).
$$
Since
\[
0\leq
\frac{
\operatorname{Var}_t\{e_j^\intercal\Phi_{s,L}(D)\}
}{
e_j^\intercal\Sigma_{\mathrm{SAGE},L}(\gamma^*)e_j
}
\leq
\frac{1}{c_s+\operatorname{Var}(W^{(s)})},
\]
by the triangle inequality,
\[
\left|
\frac{
m e_j^\intercal\widehat{V}_{\rm SAGE}(\gamma^*)e_j
}{
e_j^\intercal\Sigma_{\mathrm{SAGE},L}(\gamma^*)e_j
}-1
\right|
\leq
\frac{
|m\hat{\delta}^2_{\rm dist}-\operatorname{Var}(W^{(s)})|
}{
c_s+\operatorname{Var}(W^{(s)})
}
+o_p(1).
\]
Therefore, for every $\varepsilon>0$, the union bound gives
$$
P\left(
\left|
\frac{
m e_j^\intercal\widehat{V}_{\rm SAGE}(\gamma^*)e_j
}{
e_j^\intercal\Sigma_{\mathrm{SAGE},L}(\gamma^*)e_j
}-1
\right|>\varepsilon
\right)
\leq
P\left(
|m\hat{\delta}^2_{\rm dist}-\operatorname{Var}(W^{(s)})|
>
\frac{\varepsilon(c_s+\operatorname{Var}(W^{(s)}))}{2}
\right)
+o(1).
$$
Lemma~\ref{lem: sage_gamma} gives
\begin{equation}
\label{eqn: Sigma_SAGE_consistency}
\lim_{L\to\infty}\limsup_{m\to\infty}
P\left(
\left|
\frac{
m e_j^\intercal\widehat{V}_{\rm SAGE}(\gamma^*)e_j
}{
e_j^\intercal\Sigma_{\mathrm{SAGE},L}(\gamma^*)e_j
}-1
\right|>\varepsilon
\right)
=0.
\end{equation}
Combining the fixed-$L$ CLT in
Eqn \eqref{eqn: sage_ci_clt} with variance consistency and Slutsky's,
\begin{equation}
\label{eqn: sage_studentized}
\frac{
\hat\theta_{\mathrm{SAGE},j}(\gamma^*)-\theta_{t,j}
}{
\sqrt{e_j^\intercal\widehat{V}_{\rm SAGE}(\gamma^*)e_j}
}
\overset{d}{\rightarrow}
N(0,1),
\end{equation}
with limits taken first in $m$, followed by $L$. Consequently,
\begin{equation}
\label{eqn: sage_CI_final}
\hat\theta_{\mathrm{SAGE},j}(\gamma^*)
\pm
z_{1-\alpha/2}
\sqrt{
e_j^\intercal
\widehat{V}_{\rm SAGE}(\gamma^*)
e_j
}
\end{equation}
is an asymptotically valid $1-\alpha$ confidence interval for
$\theta_{t,j}$.

We now replace $\gamma^*$ by $\hat\gamma$ and write
$\Phi_k(\gamma)$ and $\widehat\Phi_k(\gamma)$ to explicitly display
the step-size dependence previously suppressed. We first show that, for each fixed $L$, the expansion in Eqn \eqref{eqn: sage_ci_linearization} holds with $\gamma^*$
replaced by $\hat\gamma$. Define the remainder by
$$
R_m(\gamma)
\coloneqq
\sqrt m
(\hat\theta_\mathrm{SAGE}(\gamma)-\theta_t)
+
H_t^{-1}
\left[
\sqrt m\,
\hat E_s\left[
\nabla_\theta\ell(\theta_t;D)
-
\gamma\Xi_s(D)
\right]
+
\gamma\sqrt m\,
\hat E_t[\Xi_t(D)]
\right].
$$
By Corollary~\ref{corr: grad_linear},
$R_m(0)=o_p(1)$ and $R_m(1)=o_p(1)$.
Substituting
$\hat\theta_\mathrm{SAGE}(\gamma)
=\hat\theta_\mathrm{init}
-\gamma\widehat H_s^{-1}\widehat\Psi_t$
into the definition of $R_m(\gamma)$ gives
$$
\begin{aligned}
R_m(\gamma)
&=
R_m(0)
+
\gamma\sqrt m
\left[
-\widehat H_s^{-1}\widehat\Psi_t
+
H_t^{-1}
\left\{
\hat E_t[\Xi_t(D)]-\hat E_s[\Xi_s(D)]
\right\}
\right]
\\
&=
R_m(0)+\gamma\{R_m(1)-R_m(0)\}
\\
&=
(1-\gamma)R_m(0)+\gamma R_m(1).
\end{aligned}
$$
Evaluating this identity at $\hat\gamma$ and using
$\hat\gamma\in[0,1]$, we get
$$
\|R_m(\hat\gamma)\|_2
\leq
(1-\hat\gamma)\|R_m(0)\|_2
+
\hat\gamma\|R_m(1)\|_2
\leq
\|R_m(0)\|_2+\|R_m(1)\|_2
=o_p(1).
$$
Therefore, for each fixed $L$, the expansion in
Eqn \eqref{eqn: sage_ci_linearization} also holds with
$\gamma^*$ replaced by $\hat\gamma$.

Next, to establish the limiting distribution at $\hat\gamma$, note
that the SAGE update can be written as
$$
\frac{
\sqrt m\{
\hat\theta_{\mathrm{SAGE},j}(\hat\gamma)
-
\hat\theta_{\mathrm{SAGE},j}(\gamma^*)
\}
}{
\sqrt{e_j^\intercal\Sigma_{\mathrm{SAGE},L}(\gamma^*)e_j}
}
=
-(\hat\gamma-\gamma^*)
\frac{
\sqrt m\,e_j^\intercal\widehat H_s^{-1}\widehat\Psi_t
}{
\sqrt{e_j^\intercal\Sigma_{\mathrm{SAGE},L}(\gamma^*)e_j}
}.
$$
For every $\varepsilon>0$ and $M>0$, the union bound gives
$$
P\left(
\left|
\frac{
\sqrt m\{
\hat\theta_{\mathrm{SAGE},j}(\hat\gamma)
-
\hat\theta_{\mathrm{SAGE},j}(\gamma^*)
\}
}{
\sqrt{e_j^\intercal\Sigma_{\mathrm{SAGE},L}(\gamma^*)e_j}
}
\right|>\varepsilon
\right) \leq
\underbrace{P\left(
|\hat\gamma-\gamma^*|>\frac{\varepsilon}{M}
\right)}_{\mathbf{(A)}}
+
\underbrace{P\left(
\left|
\frac{
\sqrt m\,e_j^\intercal\widehat H_s^{-1}\widehat\Psi_t
}{
\sqrt{e_j^\intercal\Sigma_{\mathrm{SAGE},L}(\gamma^*)e_j}
}
\right|>M
\right)}_{\mathbf{(B)}}.
$$
For each fixed $M$, Lemma~\ref{lem: sage_gamma} implies that
$\mathbf{(A)}$ converges to zero with limits taken
first in $m$, followed by $L$.

We now evaluate $\mathbf{(B)}$. Corollary~\ref{corr: grad_linear}
and the distributional CLT in Lemma~\ref{lemma: distr_clt}
give that, for each fixed $L$, 
$\sqrt m\,e_j^\intercal\widehat H_s^{-1}\widehat\Psi_t$ has the limiting variance
$$
(c_s+\operatorname{Var}(W^{(s)}))
\operatorname{Var}_t\{e_j^\intercal H_t^{-1}\Xi_s\}
+
c_t\operatorname{Var}_t\{e_j^\intercal H_t^{-1}\Xi_t\}=
(c_s+c_t+\operatorname{Var}(W^{(s)}))
\operatorname{Var}_t\{e_j^\intercal H_t^{-1}\Xi_t\},
$$
where the equality follows from $\Xi_s(D)=\Xi_t(D)$
by definition. Dividing by
$\sqrt{e_j^\intercal\Sigma_{\mathrm{SAGE},L}(\gamma^*)e_j}$
therefore gives
\begin{equation}
\label{eqn: clt_remainder}
\frac{
\sqrt m\,e_j^\intercal\widehat H_s^{-1}\widehat\Psi_t
}{
\sqrt{e_j^\intercal\Sigma_{\mathrm{SAGE},L}(\gamma^*)e_j}
}
\overset{d}{\rightarrow}
N\left(
0,
\frac{
(c_s+c_t+\operatorname{Var}(W^{(s)}))
\operatorname{Var}_t\{e_j^\intercal H_t^{-1}\Xi_t\}
}{
e_j^\intercal\Sigma_{\mathrm{SAGE},L}(\gamma^*)e_j
}
\right).
\end{equation}
Next, since $\Phi_t(\gamma^*)=\gamma^*H_t^{-1}\Xi_t$,
$$
0\leq
\frac{
(c_s+c_t+\operatorname{Var}(W^{(s)}))
\operatorname{Var}_t\{e_j^\intercal H_t^{-1}\Xi_t\}
}{
e_j^\intercal\Sigma_{\mathrm{SAGE},L}(\gamma^*)e_j
}=
\frac{c_s+c_t+\operatorname{Var}(W^{(s)})}{(\gamma^*)^2}
\frac{
\operatorname{Var}_t\{e_j^\intercal\Phi_t(\gamma^*)\}
}{
e_j^\intercal\Sigma_{\mathrm{SAGE},L}(\gamma^*)e_j
}
\leq
\frac{
c_s+c_t+\operatorname{Var}(W^{(s)})
}{
c_t(\gamma^*)^2
},
$$
where the last inequality follows since
$e_j^\intercal\Sigma_{\mathrm{SAGE},L}(\gamma^*)e_j \geq
c_t\operatorname{Var}_t\{e_j^\intercal\Phi_t(\gamma^*)\}.$
By Portmanteau and Chebyshev's applied to Eqn \eqref{eqn: clt_remainder},
$$
\limsup_{m\to\infty}
P\left(
\left|
\frac{
\sqrt m\,e_j^\intercal\widehat H_s^{-1}\widehat\Psi_t
}{
\sqrt{e_j^\intercal\Sigma_{\mathrm{SAGE},L}(\gamma^*)e_j}
}
\right|>M
\right)
\leq
\frac{
c_s+c_t+\operatorname{Var}(W^{(s)})
}{
M^2c_t(\gamma^*)^2
}.
$$
Taking limits first in $m$, then in $L$,
and finally letting $M\to\infty$ in the union bound above
yields, for every $\varepsilon>0$,
$$
\lim_{L\to\infty}\limsup_{m\to\infty}
P\left(
\left|
\frac{
\sqrt m\{
\hat\theta_{\mathrm{SAGE},j}(\hat\gamma)
-
\hat\theta_{\mathrm{SAGE},j}(\gamma^*)
\}
}{
\sqrt{e_j^\intercal\Sigma_{\mathrm{SAGE},L}(\gamma^*)e_j}
}
\right|>\varepsilon
\right)
=
0.
$$
Combining this with the standardized version of
Eqn \eqref{eqn: sage_ci_clt} therefore gives
\begin{equation}\label{eqn: sage_hatgamma_clt}
\frac{
\sqrt m\{
\hat\theta_{\mathrm{SAGE},j}(\hat\gamma)-\theta_{t,j}
\}
}{
\sqrt{e_j^\intercal\Sigma_{\mathrm{SAGE},L}(\gamma^*)e_j}
}
\overset{d}{\rightarrow}N(0,1),
\end{equation}
with limits taken first in $m$, followed by $L$.

It remains to establish variance consistency at $\hat\gamma$.
By definition,
\begin{align*}
\widehat\Phi_s(\hat\gamma)-\widehat\Phi_s(\gamma^*)
&=
-\widehat H_s^{-1}
\left\{
\nabla_\theta\ell(\hat\theta_\mathrm{init};D)
-
\hat\gamma\widehat\Xi(D)
\right\}
+
\widehat H_s^{-1}
\left\{
\nabla_\theta\ell(\hat\theta_\mathrm{init};D)
-
\gamma^*\widehat\Xi(D)
\right\}\\
&=
(\hat\gamma-\gamma^*)\widehat H_s^{-1}\widehat\Xi(D)\\
&=
\frac{\hat\gamma-\gamma^*}{\gamma^*}
\widehat\Phi_t(\gamma^*).
\end{align*}
By the same calculation,
$\widehat\Phi_t(\hat\gamma)
-
\widehat\Phi_t(\gamma^*)
=\frac{\hat\gamma-\gamma^*}{\gamma^*}
\widehat\Phi_t(\gamma^*).$
Applying this and the reverse triangle inequality to the following difference in standard deviations, 
\begin{align}\label{eqn: diff_std}
\begin{split}
    \left|
\sqrt{e_j^\intercal\widehat{V}_{\rm SAGE}(\hat\gamma)e_j}
-
\sqrt{e_j^\intercal\widehat{V}_{\rm SAGE}(\gamma^*)e_j}
\right|&\leq
\left[
\left(\frac{1}{n_s}+\hat\delta^2_{\rm dist}\right)
\widehat{\operatorname{Var}}_s
\left\{
e_j^\intercal
\big(\widehat\Phi_s(\hat\gamma)-\widehat\Phi_s(\gamma^*)\big)
\right\}
+
\frac{1}{n_t}
\widehat{\operatorname{Var}}_s
\left\{
e_j^\intercal
\big(\widehat\Phi_t(\hat\gamma)-\widehat\Phi_t(\gamma^*)\big)
\right\}
\right]^{1/2}\\
&=
\frac{|\hat\gamma-\gamma^*|}{\gamma^*}
\left[
\left(\frac{1}{n_s}+\frac{1}{n_t}+\hat\delta^2_{\rm dist}\right)
\widehat{\operatorname{Var}}_s
\{e_j^\intercal\widehat\Phi_t(\gamma^*)\}
\right]^{1/2}.
\end{split}
\end{align}

Since $
e_j^\intercal\widehat{V}_{\rm SAGE}(\gamma^*)e_j
\geq
\frac{1}{n_t}
\widehat{\operatorname{Var}}_s
\{e_j^\intercal\widehat\Phi_t(\gamma^*)\},$ we write
$\sqrt{
\widehat{\operatorname{Var}}_s
\{e_j^\intercal\widehat\Phi_t(\gamma^*)\}
}
\leq
\sqrt{n_t}\,
\sqrt{e_j^\intercal\widehat{V}_{\rm SAGE}(\gamma^*)e_j}.$
Substituting this bound into Eqn \eqref{eqn: diff_std} and dividing by
$\sqrt{e_j^\intercal\widehat{V}_{\rm SAGE}(\gamma^*)e_j}$ gives
\begin{equation}\label{eqn: std_ratio_gamma_hat}
    \left|
\sqrt{
\frac{
e_j^\intercal\widehat{V}_{\rm SAGE}(\hat\gamma)e_j
}{
e_j^\intercal\widehat{V}_{\rm SAGE}(\gamma^*)e_j
}
}
-1
\right|
\leq
\frac{|\hat\gamma-\gamma^*|}{\gamma^*}
\sqrt{
1+\frac{n_t}{n_s}+n_t\hat\delta^2_{\rm dist}
}.
\end{equation}

This holds on the event the denominator is positive, which occurs with probability tending to one by Eqn \eqref{eqn: Sigma_SAGE_consistency}. By Lemma~\ref{lem: sage_gamma},
$|\hat\gamma-\gamma^*|\overset{p}{\rightarrow}0$ and $n_t\hat\delta^2_{\rm dist}\overset{p}{\rightarrow}\frac{\operatorname{Var}(W^{(s)})}{c_t}.$ Together with $n_t/n_s\to c_s/c_t$ and $\gamma^*>0$, we get
$$
\frac{1}{\gamma^*}
\sqrt{1+\frac{n_t}{n_s}+n_t\hat\delta^2_{\rm dist}}
\overset{p}{\rightarrow}
\frac{1}{\gamma^*}
\sqrt{1+\frac{c_s}{c_t}+\frac{\operatorname{Var}(W^{(s)})}{c_t}}
<\infty,
$$
with limits first in $m$ then in $L$. Therefore, together with continuous mapping,
$$
\frac{
e_j^\intercal\widehat{V}_{\rm SAGE}(\hat\gamma)e_j
}{
e_j^\intercal\widehat{V}_{\rm SAGE}(\gamma^*)e_j
}
\overset{p}{\rightarrow}1.
$$

In addition, Eqn \eqref{eqn: Sigma_SAGE_consistency} gives
$$
\frac{
m e_j^\intercal\widehat{V}_{\rm SAGE}(\gamma^*)e_j
}{
e_j^\intercal\Sigma_{\mathrm{SAGE},L}(\gamma^*)e_j
}
\overset{p}{\rightarrow}1.
$$
Multiplying these two ratios therefore yields
$$
\frac{
m e_j^\intercal\widehat{V}_{\rm SAGE}(\hat\gamma)e_j
}{
e_j^\intercal\Sigma_{\mathrm{SAGE},L}(\gamma^*)e_j
}
=
\frac{
e_j^\intercal\widehat{V}_{\rm SAGE}(\hat\gamma)e_j
}{
e_j^\intercal\widehat{V}_{\rm SAGE}(\gamma^*)e_j
}
\frac{
m e_j^\intercal\widehat{V}_{\rm SAGE}(\gamma^*)e_j
}{
e_j^\intercal\Sigma_{\mathrm{SAGE},L}(\gamma^*)e_j
}
\overset{p}{\rightarrow}1.
$$
Combining this variance consistency with Eqn \ref{eqn: sage_hatgamma_clt} and Slutsky's, we get
$$
\frac{
\hat\theta_{\mathrm{SAGE},j}(\hat\gamma)-\theta_{t,j}
}{
\sqrt{e_j^\intercal\widehat{V}_{\rm SAGE}(\hat\gamma)e_j}
}
\overset{d}{\rightarrow}N(0,1),
$$
with limits taken first in $m$, followed by $L$. Therefore,
$$
\hat\theta_{\mathrm{SAGE},j}(\hat\gamma)
\pm
z_{1-\alpha/2}
\sqrt{e_j^\intercal\widehat{V}_{\rm SAGE}(\hat\gamma)e_j}
$$
is an asymptotically valid $1-\alpha$ confidence interval
for $\theta_{t,j}$, with limits taken first in $m$,
followed by $L$.
\end{proof}

\section{Proofs of supporting results}

\subsection{Distributional CLT}\label{app: clt}

\begin{proof}
    This lemma follows as a special case of Lemma 1 of \cite{duc} and applying the Cramér-Wold device since Lemma 1 gives a distributional CLT for scalar-valued rather than vector-valued functions.  Specifically, Lemma 1 considers scalar Borel measurable functions that are square-integrable under $P_f$. Taking $P_f=P_1=P_t$ and $P_2=P_s$,
$$
\sqrt{m}
\begin{pmatrix}
\frac{1}{n_t} \sum_{i=1}^{n_t} \phi_t(D_i^{(t)}) - E_t[\phi_t(D)]\\
\frac{1}{n_s} \sum_{i=1}^{n_s} \phi_s(D_i^{(s)})-E_t[\phi_s(D)]
\end{pmatrix}
\overset{d}{\rightarrow} N\left(0, \begin{pmatrix}
    c_t\operatorname{Var}_t(\phi_t(D))& 0\\
    0& (\operatorname{Var}(W^{(s)})+c_s)\operatorname{Var}_t(\phi_s(D))
\end{pmatrix}\right),
$$
since $P_t$ is fixed, the target component has no random-weight contribution and the limiting cross-covariance is zero. 

Now define $a_s\in\mathbb R^{p_s}$ and $a_t\in\mathbb R^{p_t}$. Applying the above for $\phi_s(D)=a_s^\intercal f_s(D)$ and 
$\phi_t(D)=a_t^\intercal f_t(D)$ gives
\begin{align*}
&\sqrt{m}\left(
a_t^\intercal\left(
\frac{1}{n_t}\sum_{i=1}^{n_t}f_t(D_i^{(t)})
-
E_t[f_t(D)]\right)+
a_s^\intercal\left(
\frac{1}{n_s}\sum_{i=1}^{n_s}f_s(D_i^{(s)})
-
E_t[f_s(D)]
\right)
\right) \\
&\qquad\overset{d}{\rightarrow}
N\left(
0,
c_ta_t^\intercal\operatorname{Var}_t(f_t(D))a_t + (\operatorname{Var}(W^{(s)})+c_s)a_s^\intercal\operatorname{Var}_t(f_s(D))a_s
\right).
\end{align*}
Then by the Cramèr-Wold device, 
\begin{equation*}
\sqrt{m}
\begin{pmatrix}

\frac{1}{n_t} \sum_{i=1}^{n_t} f_t(D_i^{(t)}) 
-
E_t[f_t(D)]\\
\frac{1}{n_s} \sum_{i=1}^{n_s} f_s(D_i^{(s)}) 
-
E_t[f_s(D)]
\end{pmatrix}
\overset{d}{\rightarrow}
N\left(0,
\begin{pmatrix}
c_t\operatorname{Var}_t(f_t(D))
&
0_{p_t\times p_s}
\\
0_{p_s\times p_t}
&
\left(c_s+\operatorname{Var}(W^{(s)})\right)\operatorname{Var}_t(f_s(D))
\end{pmatrix}
\right)
\end{equation*}

\end{proof}

\subsection{Consistency and linearization of $\hat\theta_\mathrm{init}$}\label{app: theta_old_consist} 
\begin{proposition}
\label{prop: theta_old_consist}
Under the assumptions in Appendix~\ref{app: assumptions}, we have
\begin{equation*}
    \hat\theta_\mathrm{init}
    =
    \theta_t+O_p(m^{-1/2}),
\end{equation*}
and the linearization
\begin{equation}
    \hat\theta_\mathrm{init}-\theta_t=
    -H_t^{-1}
    \hat E_s[
        \nabla_\theta\ell(\theta_t;D)
    ]
    +
    o_p(m^{-1/2}).
\end{equation}
\end{proposition}

\begin{proof}
To establish the first statement, we adapt the argument
of Lemma~2 of \cite{duc}, with $K=1$, $P_f=P_t$,
$P_1=P_s$, and $\hat\theta=\hat\theta_\mathrm{init}$.
Since that lemma assumes bounded loss, we verify the
argument under our assumptions.
By Assumption \eqref{assump: lip_loss},
$\ell(\theta;D)$ is square-integrable for each fixed $\theta$.
Lemma~\ref{lemma: distr_clt} therefore gives
$\hat E_s[\ell(\theta;D)]\overset p\rightarrow
E_t[\ell(\theta;D)]$.
Moreover, $L_1$ is square-integrable, so the same lemma gives
$\hat E_s[L_1(D)]\overset p\rightarrow E_t[L_1(D)]<\infty$.
Compactness of $\Theta$ and a finite-covering argument
using this Lipschitz bound then yield
\[
\sup_{\theta\in\Theta}
\left|
\hat E_s[\ell(\theta;D)]-E_t[\ell(\theta;D)]
\right|
\overset p\rightarrow0.
\]
Since $\theta_t$ uniquely minimizes the continuous target
risk on $\Theta$, it follows that
$\hat\theta_{\mathrm{init}}\overset p\rightarrow\theta_t$.

We now linearize
$\hat\theta_\mathrm{init}-\theta_t$.
Let $\Delta_m \coloneqq \hat\theta_\mathrm{init}-\theta_t$
and recall that
$\widehat H_s(\theta)
\coloneqq \hat E_s[\nabla_\theta^2\ell(\theta;D)].$
Since $\hat\theta_\mathrm{init}\overset p\rightarrow\theta_t$
and $\theta_t\in\operatorname{int}(\Theta)$, the source
first-order condition holds with probability tending to one.
A Taylor expansion around $\theta_t$ gives
\begin{align}\label{eqn: thetahat_taylor1}
0
=
\hat E_s[
    \nabla_\theta\ell(\hat\theta_\mathrm{init};D)
]
=
\hat E_s[
    \nabla_\theta\ell(\theta_t;D)
]
+
\widehat H_s(\theta_t)\Delta_m
+
\frac{1}{2}
\hat E_s\left[
    \nabla_\theta^3\ell(\theta_t;D)
    [\Delta_m,\Delta_m]
\right]
+
r_m.
\end{align}
Here,
$\nabla_\theta^3\ell(\theta;D)[u,v]$
denotes the $p$-dimensional vector obtained by applying the third
derivative tensor to $u$ and $v$.
By the Lipschitz condition on
$\nabla_\theta^3\ell(\cdot)$,
\begin{equation}
\label{eqn: score_taylor_remainder}
\|r_m\|_2
\leq
C\hat E_s[M_3(D)]
\|\Delta_m\|_2^3
\end{equation}
for a fixed constant $C<\infty$.
Moreover,
\begin{align*}
\left\|
\frac12
\hat E_s\left[
\nabla_\theta^3\ell(\theta_t;D)
[\Delta_m,\Delta_m]
\right]
\right\|_2
&\leq
\frac12
\hat E_s[F_3(D)]
\|\Delta_m\|_2^2.
\end{align*}
By Lemma~\ref{lemma: distr_clt},
$\hat E_s[M_3(D)]=O_p(1)$ and
$\hat E_s[F_3(D)]=O_p(1)$. Also by Lemma~\ref{lemma: distr_clt}, we have
$\hat E_s[\nabla_\theta\ell(\theta_t;D)]=O_p(m^{-1/2})$,
since $E_t[\nabla_\theta\ell(\theta_t;D)]=0$,
and $\widehat H_s(\theta_t)-H_t=O_p(m^{-1/2})$.
Since
$\widehat H_s(\theta_t)\overset{p}{\rightarrow}H_t$
and $H_t$ is positive definite,
$\widehat H_s(\theta_t)$ is invertible with probability tending to one,
$
\widehat H_s(\theta_t)^{-1}
=
H_t^{-1}+O_p(m^{-1/2})
=
O_p(1).
$

We first establish the rate of $\Delta_m$.
Rearranging Eqn.~\eqref{eqn: thetahat_taylor1},
\[
\|\Delta_m\|_2
\leq
\left\|
\widehat H_s(\theta_t)^{-1}
\hat E_s[\nabla_\theta\ell(\theta_t;D)]
\right\|_2
+
C_m\|\Delta_m\|_2^2,
\]
where
$
C_m
\coloneqq
\|\widehat H_s(\theta_t)^{-1}\|_{\mathrm{op}}
\left\{
\frac12\hat E_s[F_3(D)]
+
C\hat E_s[M_3(D)]\|\Delta_m\|_2
\right\}
=
O_p(1).
$
Since $\Delta_m=o_p(1)$ by consistency,
$C_m\|\Delta_m\|_2=o_p(1)$.
Thus, with probability tending to one,
$C_m\|\Delta_m\|_2\leq1/2$, and we get
\[
\|\Delta_m\|_2
\leq
2\left\|
\widehat H_s(\theta_t)^{-1}
\hat E_s[\nabla_\theta\ell(\theta_t;D)]
\right\|_2.
\]
Consequently, $\Delta_m=O_p(m^{-1/2})$.
Substituting this into Eqn \eqref{eqn: score_taylor_remainder} gives
$\|r_m\|_2=O_p(m^{-3/2})$ and
\[
\left\|
\frac12
\hat E_s\left[
\nabla_\theta^3\ell(\theta_t;D)
[\Delta_m,\Delta_m]
\right]
\right\|_2
=
O_p(m^{-1}).
\]

Therefore, rearranging
Eqn.~\eqref{eqn: thetahat_taylor1} gives
\begin{align}
\label{eqn: theta_linearize}
\begin{split}
\Delta_m
&=
-
\widehat H_s(\theta_t)^{-1}
\hat E_s[
\nabla_\theta\ell(\theta_t;D)
]
-
\frac12
\widehat H_s(\theta_t)^{-1}
\hat E_s\left[
\nabla_\theta^3\ell(\theta_t;D)
[\Delta_m,\Delta_m]
\right]
-
\widehat H_s(\theta_t)^{-1}r_m
\\
&=
-
\widehat H_s(\theta_t)^{-1}
\hat E_s[
\nabla_\theta\ell(\theta_t;D)
]
+
O_p(m^{-1})
\\
&=
-
H_t^{-1}
\hat E_s[
\nabla_\theta\ell(\theta_t;D)
]
+
\left(
H_t^{-1}
-
\widehat H_s(\theta_t)^{-1}
\right)
\hat E_s[
\nabla_\theta\ell(\theta_t;D)
]
+
O_p(m^{-1})
\\
&=
-
H_t^{-1}
\hat E_s[
\nabla_\theta\ell(\theta_t;D)
]
+
O_p(m^{-1/2})O_p(m^{-1/2})
+
O_p(m^{-1})
\\
&=
-
H_t^{-1}
\hat E_s[
\nabla_\theta\ell(\theta_t;D)
]
+
O_p(m^{-1})
\\
&=
-
H_t^{-1}
\hat E_s[
\nabla_\theta\ell(\theta_t;D)
]
+
o_p(m^{-1/2}).
\end{split}
\end{align}
\end{proof}

\subsection{Consistency of $\widehat\Psi_t$ and $\widehat H_s$}
\begin{lemma}\label{lemma: grad_consist}

Under the assumptions in Appendix \ref{app: assumptions},
\begin{equation}\label{eqn: psi_projection_approx}
\widehat\Psi_t
=
\beta^{(t)}(\theta_t)^\intercal(\hat d_t-\hat d_s)+o_p(m^{-1/2})
=
O_p(m^{-1/2}).
\end{equation}

\end{lemma} 

\begin{proof}
We first show that $\hat\beta^{(s)}(\hat\theta_\mathrm{init})-\beta^{(t)}(\theta_t)=O_p(m^{-1/2})$. Recall that
$$
\hat C^{(s)}(\theta)
=
\hat E_s[
\widetilde Z(D)\widetilde Z(D)^\intercal
]^{-1}
\hat E_s[
\widetilde Z(D)
\nabla_\theta\ell(\theta;D)^\intercal]
$$
is the empirical least-squares coefficient matrix obtained by
projecting
$\nabla_\theta\ell(\theta;D)$
onto an intercept and $Z(D)$ using the source observations. Its
intercept and slope coefficients are respectively
$\hat\alpha^{(s)}(\theta)$ and $\hat\beta^{(s)}(\theta)$.
Similarly,
\[
C^{(t)}(\theta)
=
E_t[
\widetilde Z(D)\widetilde Z(D)^\intercal
]^{-1}
E_t[
\widetilde Z(D)
\nabla_\theta\ell(\theta;D)^\intercal
]
\]
is the corresponding population target projection, with coefficients
$\alpha^{(t)}(\theta)$ and $\beta^{(t)}(\theta)$.

We first compare $\hat C^{(s)}(\theta_t)$ and $C^{(t)}(\theta_t)$ By
Lemma~\ref{lemma: distr_clt}, $\hat E_s[
\widetilde Z(D)\widetilde Z(D)^\intercal
]
=E_t[
\widetilde Z(D)\widetilde Z(D)^\intercal
]+O_p(m^{-1/2})$ and $\hat E_s[
\widetilde Z(D)
\nabla_\theta\ell(\theta_t;D)^\intercal]
= E_t[
\widetilde Z(D)
\nabla_\theta\ell(\theta_t;D)^\intercal
]
+O_p(m^{-1/2})$.

Since $E_t[
\widetilde Z(D)\widetilde Z(D)^\intercal
]$ is full rank by Assumption~\eqref{assump: var(z) pd}, $\hat E_s[
\widetilde Z(D)\widetilde Z(D)^\intercal]$
is invertible with probability tending to one, and its inverse is
$O_p(1)$. Therefore, $\hat E_s[
\widetilde Z(D)\widetilde Z(D)^\intercal
]^{-1}=E_t[
\widetilde Z(D)\widetilde Z(D)^\intercal
]^{-1}+O_p(m^{-1/2})$.

Using the definitions of
$\hat C^{(s)}(\theta_t)$ and $C^{(t)}(\theta_t)$, we therefore have
\begin{align}\label{eqn: projection_C_fixed_rate}
\begin{split}
\hat C^{(s)}(\theta_t)
-
C^{(t)}(\theta_t)
&=
\hat E_s[
\widetilde Z\widetilde Z^\intercal
]^{-1}
\hat E_s[
\widetilde Z
\nabla_\theta\ell(\theta_t;D)^\intercal
]
-
E_t[
\widetilde Z\widetilde Z^\intercal
]^{-1}E_t[
\widetilde Z
\nabla_\theta\ell(\theta_t;D)^\intercal]
\\
&=
\hat E_s[
\widetilde Z\widetilde Z^\intercal
]^{-1}
\left(
\hat E_s[
\widetilde Z
\nabla_\theta\ell(\theta_t;D)^\intercal]
-
E_t\left[
\widetilde Z
\nabla_\theta\ell(\theta_t;D)^\intercal
\right]
\right)\\
&\qquad+
\left[
\hat E_s[
\widetilde Z\widetilde Z^\intercal
]^{-1}
-
E_t[
\widetilde Z\widetilde Z^\intercal
]^{-1}
\right]
E_t\left[
\widetilde Z
\nabla_\theta\ell(\theta_t;D)^\intercal
\right]\\
&\quad=
O_p(1)O_p(m^{-1/2})
+
O_p(m^{-1/2})O(1)\\
&\quad=
O_p(m^{-1/2}),
\end{split}
\end{align}
where the arguments of $\widetilde Z(D)$ have been suppressed for readability.

We next look at $\hat C^{(s)}(\hat\theta_\mathrm{init})$. We first write
\begin{align}
\hat C^{(s)}(\hat\theta_\mathrm{init})
-
\hat C^{(s)}(\theta_t)
&=
\hat E_s[
\widetilde Z(D)\widetilde Z(D)^\intercal
]^{-1}
\hat E_s\left[
\widetilde Z(D)
\left(
\nabla_\theta\ell(\hat\theta_\mathrm{init};D)
-
\nabla_\theta\ell(\theta_t;D)
\right)^{\intercal}
\right].
\label{eqn: projection_C_theta_difference}
\end{align}
A Taylor expansion of
$\nabla_\theta\ell(\hat\theta_\mathrm{init};D)$ around $\theta_t$
gives
\begin{equation}
\label{eqn: projection_gradient_taylor}
\nabla_\theta\ell(\hat\theta_\mathrm{init};D)
=
\nabla_\theta\ell(\theta_t;D)
+
\nabla_\theta^2\ell(\theta_t;D)
(\hat\theta_\mathrm{init}-\theta_t)
+
r_m(D),
\end{equation}
where, by the Lipschitz condition on
$\nabla_\theta^2\ell(\theta;D)$ in
Assumption~\eqref{assump: lip}, $\|r_m(D)\|_2
\leq
\frac{1}{2}
M_2(D)
\|\hat\theta_\mathrm{init}-\theta_t\|_2^2$.

Substituting \eqref{eqn: projection_gradient_taylor} into
\eqref{eqn: projection_C_theta_difference} gives
\begin{align*}
&\left\|
\hat C^{(s)}(\hat\theta_\mathrm{init})
-
\hat C^{(s)}(\theta_t)
\right\|_F\\
&\leq
\left\|
\hat E_s[
\widetilde Z(D)\widetilde Z(D)^\intercal
]^{-1}
\right\|_{\mathrm{op}}
\bigg[
\hat E_s[
\|\widetilde Z(D)\|_2
\|\nabla_\theta^2\ell(\theta_t;D)\|_{\mathrm{op}}]
\|\hat\theta_\mathrm{init}-\theta_t\|_2
+
\frac{1}{2}
\hat E_s\left[
\|\widetilde Z(D)\|_2M_2(D)
\right]
\|\hat\theta_\mathrm{init}-\theta_t\|_2^2
\bigg]
\\
&=
O_p(1)
\left(
O_p(1)O_p(m^{-1/2})
+
O_p(1)O_p(m^{-1})
\right)
\\
&=
O_p(m^{-1/2}),
\end{align*}
where we use $\hat\theta_\mathrm{init}-\theta_t
=O_p(m^{-1/2})$ from Proposition~\ref{prop: theta_old_consist}. Combining this result with
\eqref{eqn: projection_C_fixed_rate}, we get
\begin{align*}
\hat C^{(s)}(\hat\theta_\mathrm{init})
-
C^{(t)}(\theta_t)=
\left(
\hat C^{(s)}(\hat\theta_\mathrm{init})
-
\hat C^{(s)}(\theta_t)
\right)
+
\left(
\hat C^{(s)}(\theta_t)
-
C^{(t)}(\theta_t)
\right)
=
O_p(m^{-1/2}).
\end{align*}
Therefore, $\hat\beta^{(s)}(\hat\theta_\mathrm{init})-\beta^{(t)}(\theta_t)=O_p(m^{-1/2})$.

We next establish the representation of
$\widehat\Psi_t$. Since the least-squares projection contains an
intercept, we have
\begin{equation}
\label{eqn: projection_intercept_normal_equation}
\hat E_s\left[
\nabla_\theta\ell(\hat\theta_\mathrm{init};D)
-
\hat\alpha^{(s)}(\hat\theta_\mathrm{init})
-
\sum_{g=1}^K
\left\{
\hat\beta_g^{(s)}(\hat\theta_\mathrm{init})
\right\}^{\intercal}
Z_g(D)
\right]
=
0.
\end{equation}
Moreover, since
$\hat\theta_\mathrm{init}$ is an interior source empirical risk
minimizer with probability tending to one, $\hat E_s[
\nabla_\theta\ell(\hat\theta_\mathrm{init};D)
]=0$. It follows from
\eqref{eqn: projection_intercept_normal_equation} that
\begin{align*}
0
&=
\hat\alpha^{(s)}(\hat\theta_\mathrm{init})
+
\sum_{g=1}^K
\left\{
\hat\beta_g^{(s)}(\hat\theta_\mathrm{init})
\right\}^{\intercal}
\hat E_s[Z_g(D)]\\
&=
\hat\alpha^{(s)}(\hat\theta_\mathrm{init})
+
\sum_{g=1}^K
\left\{
\hat\beta_g^{(s)}(\hat\theta_\mathrm{init})
\right\}^{\intercal}
\hat d_{s,g}\\
&=
\hat\alpha^{(s)}(\hat\theta_\mathrm{init})
+
\left\{
\hat\beta^{(s)}(\hat\theta_\mathrm{init})
\right\}^{\intercal}
\hat d_s.
\end{align*}
Therefore, we can write
\begin{align}
\widehat\Psi_t
=
\left\{
\hat\beta^{(s)}(\hat\theta_\mathrm{init})
\right\}^{\intercal}
(\hat d_t-\hat d_s)
=
\hat\alpha^{(s)}(\hat\theta_\mathrm{init})
+
\left\{
\hat\beta^{(s)}(\hat\theta_\mathrm{init})
\right\}^{\intercal}
\hat d_t.
\end{align}

Finally, Lemma~\ref{lemma: distr_clt} gives
\begin{align}
\hat d_t-\hat d_s
=
(\hat d_t-d_t)
-
(\hat d_s-d_t)
=
O_p(m^{-1/2}).
\label{eqn: projection_d_difference_rate}
\end{align}
Therefore,
\begin{align*}
\widehat\Psi_t
-
\left\{
\beta^{(t)}(\theta_t)
\right\}^{\intercal}
(\hat d_t-\hat d_s)
&=
\left[
\left\{
\hat\beta^{(s)}(\hat\theta_\mathrm{init})
\right\}^{\intercal}
-
\left\{
\beta^{(t)}(\theta_t)
\right\}^{\intercal}
\right]
(\hat d_t-\hat d_s)
\\
&=
O_p(m^{-1/2})O_p(m^{-1/2})
\\
&=
o_p(m^{-1/2})
\end{align*}
Since
$\beta^{(t)}(\theta_t)$ is fixed and
$\hat d_t-\hat d_s=O_p(m^{-1/2})$, we also get
$\widehat\Psi_t=O_p(m^{-1/2})$
\end{proof}

\begin{proposition}
\label{prop: hess_consist}
Under the assumptions in Appendix~\ref{app: assumptions},
\begin{equation*}
    \widehat H_s
    =
    H_t
    +
    O_p(m^{-1/2}) \quad \text{and} \quad  \widehat H_s^{-1}
    =
    H_t^{-1}
    +
    O_p(m^{-1/2}).
\end{equation*}
\end{proposition}

\begin{proof}
We can write
\begin{align*}
\widehat H_s-H_t
=
\underbrace{\hat E_s[
    \nabla_\theta^2\ell(\theta_t;D)
]
-
E_t[
    \nabla_\theta^2\ell(\theta_t;D)
]}_{\mathbf{(1)}}
+
\underbrace{\hat E_s\left[
    \nabla_\theta^2
    \ell(\hat\theta_\mathrm{init};D)
    -
    \nabla_\theta^2
    \ell(\theta_t;D)
\right]}_{\mathbf{(2)}}.
\end{align*}
$\mathbf{(1)}$ is $O_p(m^{-1/2})$ by
Lemma~\ref{lemma: distr_clt}. For the second term, triangle inequality and Assumption~\eqref{assump: lip} give
\begin{align*}
\left\|
\hat E_s\left[
    \nabla_\theta^2
    \ell(\hat\theta_\mathrm{init};D)
    -
    \nabla_\theta^2
    \ell(\theta_t;D)
\right]
\right\|_{\mathrm{op}}
&\leq \hat E_s\left[
\left\|
\nabla_\theta^2\ell(\hat\theta_\mathrm{init};D)
-
\nabla_\theta^2\ell(\theta_t;D)
\right\|_{\mathrm{op}}
\right]
\\
&\leq
\hat E_s[M_2(D)]
\|
\hat\theta_\mathrm{init}-\theta_t
\|_2
\\
&=
O_p(1)O_p(m^{-1/2})
\\
&=
O_p(m^{-1/2}),
\end{align*}
where
$\hat E_s[M_2(D)]=O_p(1)$ by
Lemma~\ref{lemma: distr_clt}, and $\hat\theta_\mathrm{init}-\theta_t=O_p(m^{-1/2})$ by Proposition~\ref{prop: theta_old_consist}. Therefore,
$$
    \widehat H_s
    =H_t+
    O_p(m^{-1/2}).
$$
Since $H_t$ is positive definite,
$\widehat H_s$ is invertible with probability tending to one. Therefore, we also have that $ \widehat H_s^{-1}
    =H_t^{-1}+
    O_p(m^{-1/2})$.
\end{proof}

\subsection{Asymptotic linearity of SAGE}

\begin{corollary}
\label{corr: grad_linear}
For any fixed $\gamma\geq0$, under the assumptions in Appendix~\ref{app: assumptions}, 
\begin{align}
\label{eqn: sage_linear}
\sqrt m
(\hat\theta_\mathrm{SAGE}(\gamma)-\theta_t)=
-H_t^{-1}
\left[
\sqrt m\,
\hat E_s\left[
\nabla_\theta\ell(\theta_t;D)
-
\gamma\Xi_s(D)
\right]
+
\gamma\sqrt m\,
\hat E_t[\Xi_t(D)]
\right]
+
o_p(1).
\end{align}
\end{corollary}

\begin{proof}
We first show that 
\begin{align}
\label{eqn: psi_linear}
\sqrt m\,\widehat\Psi_t
=
\sqrt m\,\hat E_t[\Xi_t(D)]
-
\sqrt m\,\hat E_s[\Xi_s(D)]
+
o_p(1).
\end{align}
By Lemma~\ref{lemma: grad_consist},
\begin{equation}
\label{eqn: psi_linear_first}
\sqrt m\,\widehat\Psi_t
=
\beta^{(t)}(\theta_t)^\intercal
\sqrt m(\hat d_t-\hat d_s)
+
o_p(1).
\end{equation}
We also have by definition of $\hat d_k$,
\begin{align*}
\beta^{(t)}(\theta_t)^\intercal
(\hat d_t-\hat d_s)
&=
\beta^{(t)}(\theta_t)^\intercal
\left\{
(\hat d_t-d_t)
-
(\hat d_s-d_t)
\right\}
\\
&\quad=
\hat E_t\left[
\beta^{(t)}(\theta_t)^\intercal
\{Z^{(t)}(D)-d_t\}
\right]
-
\hat E_s\left[
\beta^{(t)}(\theta_t)^\intercal
\{Z^{(s)}(D)-d_t\}
\right]
\\
&\quad=
\hat E_t[\Xi_t(D)]
-
\hat E_s[\Xi_s(D)].
\end{align*}
Substituting this into
\eqref{eqn: psi_linear_first} gives
\eqref{eqn: psi_linear}.

We next derive the asymptotic linear representation of the SAGE
estimator. By its definition,
\begin{align}
\sqrt m
(\hat\theta_\mathrm{SAGE}(\gamma)-\theta_t)
&=
\sqrt m
(\hat\theta_\mathrm{init}-\theta_t)
-
\gamma\sqrt m\,
\widehat H_s^{-1}\widehat\Psi_t.
\label{eqn: sage_linear_start}
\end{align}
By Proposition~\ref{prop: hess_consist} and
Lemma~\ref{lemma: grad_consist}, $\widehat H_s^{-1}=H_t^{-1}+
O_p(m^{-1/2})$ and $\widehat\Psi_t
= O_p(m^{-1/2})$ respectively.
Consequently,
\begin{align*}
\sqrt m
\left(
\widehat H_s^{-1}-H_t^{-1}
\right)
\widehat\Psi_t=
\sqrt m\,
O_p(m^{-1/2})
O_p(m^{-1/2})=
O_p(m^{-1/2})
\end{align*}
It follows that
\begin{equation}
\label{eqn: hessian_psi_replacement}
\sqrt m\,
\widehat H_s^{-1}\widehat\Psi_t
=
H_t^{-1}
\sqrt m\,\widehat\Psi_t
+
o_p(1).
\end{equation}
Proposition~\ref{prop: theta_old_consist} gives us the linearization,
\begin{equation*}
\sqrt m
(\hat\theta_\mathrm{init}-\theta_t)
=
-H_t^{-1}
\sqrt m\,
\hat E_s[
\nabla_\theta\ell(\theta_t;D)
]
+
o_p(1).
\end{equation*}
Combining this result with
\eqref{eqn: sage_linear_start},
\eqref{eqn: hessian_psi_replacement}, and
\eqref{eqn: psi_linear}, we obtain
\begin{align*}
\sqrt m(\hat\theta_\mathrm{SAGE}(\gamma)-\theta_t)
&=
-H_t^{-1}
\sqrt m\,
\hat E_s[
\nabla_\theta\ell(\theta_t;D)
]-
\gamma H_t^{-1}
(\sqrt m\,\hat E_t[\Xi_t(D)]
-
\sqrt m\,\hat E_s[\Xi_s(D)])
+
o_p(1)
\\
&\quad=
-H_t^{-1}
\left[
\sqrt m\,
\hat E_s\left[
\nabla_\theta\ell(\theta_t;D)
-
\gamma\Xi_s(D)
\right]
+
\gamma\sqrt m\,
\hat E_t[\Xi_t(D)]
\right]
+
o_p(1),
\end{align*}
which proves \eqref{eqn: sage_linear}.
\end{proof}

\subsection{Properties of $\hat\theta_{EB}$}

\begin{proposition}\label{prop: lamdba_op1} Define the entropy-balancing problem in Lagrangian form,
$$\mathcal L(w,\lambda_1,\lambda_2)
=
\sum_{i=1}^{n_s}w_i\log(n_sw_i)
+
\lambda_1\left(\sum_{i=1}^{n_s}w_i-1\right)
-
\lambda_2^\intercal
\left(
\sum_{i=1}^{n_s}w_iZ(D_i^{(s)})-\hat d_t
\right).$$
Let $\hat\lambda_2\in\mathbb R^r$ be a finite Lagrange multiplier
associated with the reduced balancing constraints, where
$r=\dim Z(D)$ after the coordinate removal in Assumption \eqref{assump: var(z) pd}.
Under the assumptions of Lemma~\ref{lemma: entropy_bal},
$\hat\lambda_2=O_p(m^{-1/2})$ and $\hat \lambda_2$ has the linearization, 
$$\hat\lambda_2
=
\left[
\widehat{\operatorname{Var}}_s\{Z(D)\}
\right]^{-1}
(\hat d_t-\hat d_s)
+
o_p(m^{-1/2}).$$
\end{proposition}
\begin{proof}
Eqn. \eqref{eqn: weights_eb} from Appendix \ref{app: entropy_bal} gives 
\begin{equation}\label{eqn: weights_eb_ns}
    n_s\hat w_i
=
\frac{
\exp\left(
\hat\lambda_2^\intercal\{Z(D_i^{(s)})-\hat d_s\}
\right)
}{\frac{1}{n_s}
\sum_{j=1}^{n_s}
\exp\left(
\hat\lambda_2^\intercal\{Z(D_j^{(s)})-\hat d_s\}
\right)
},
\end{equation}
   Taking logs on both sides gives
\begin{align}\label{eqn: log_weights0}
    \log(n_s\hat w_i)
&=
\hat\lambda_2^\intercal Z(D_i^{(s)})
-
\log\left\{
\frac1{n_s}\sum_{j=1}^{n_s}
\exp\left(
\hat\lambda_2^\intercal Z(D_j^{(s)})
\right)
\right\}
\end{align}
Multiplying each $i$th element by $\hat w_i-n_s^{-1}$
and summing over $i$ and using that $\sum_i(\hat w_i-n_s^{-1})=0$,
\begin{align}\label{eqn: log_weights}
\begin{split}
\frac1{n_s}\sum_{i=1}^{n_s}
(n_s\hat w_i-1)\log(n_s\hat w_i)
&=
\sum_{i=1}^{n_s}
(\hat w_i-n_s^{-1})
\hat\lambda_2^\intercal Z(D_i^{(s)})
\\
&=
\hat\lambda_2^\intercal
\left\{
\sum_{i=1}^{n_s}\hat w_i Z(D_i^{(s)})
-
\frac1{n_s}\sum_{i=1}^{n_s}Z(D_i^{(s)})
\right\},
\end{split}
\end{align}
where the second term in Eqn \eqref{eqn: log_weights0} after multiplying by $\sum_i(\hat w_i-n_s^{-1})$ is 0 since it is independent of $i$.
The entropy-balancing weights are chosen to satisfy $\sum_{i=1}^{n_s}\hat w_i Z(D_i^{(s)})=\hat d_t$, while $n_s^{-1}\sum_{i=1}^{n_s}Z(D_i^{(s)})=\hat d_s$
by definition. Substituting into \eqref{eqn: log_weights} gives
$$
\frac1{n_s}\sum_{i=1}^{n_s}
(n_s\hat w_i-1)\log(n_s\hat w_i)
=
\hat\lambda_2^\intercal(\hat d_t-\hat d_s).
$$
For $a\ge\kappa$, the mean value theorem for the exponential
function gives $(a-1)\log a\ge\kappa(\log a)^2$. By assumption, $\min_i n_s\hat w_i\ge\kappa$ with probability tending to one as $m\to\infty$. On this event,
\begin{align*}
\hat\lambda_2^\intercal(\hat d_t-\hat d_s)
&\geq
\frac{\kappa}{n_s}\sum_{i=1}^{n_s}
\{\log(n_s\hat w_i)\}^2
\\
&\geq
\frac{\kappa}{n_s}\sum_{i=1}^{n_s}
\left\{
\log(n_s\hat w_i)
-\frac1{n_s}\sum_{j=1}^{n_s}\log(n_s\hat w_j)
\right\}^2
\\
&=
\frac{\kappa}{n_s}\sum_{i=1}^{n_s}
\left[
\hat\lambda_2^\intercal
\{Z(D_i^{(s)})-\hat d_s\}
\right]^2
\qquad \text{(Eqn.~\eqref{eqn: log_weights0})}
\\
&=
\kappa\,\hat\lambda_2^\intercal
\left[
\frac1{n_s}\sum_{i=1}^{n_s}
\{Z(D_i^{(s)})-\hat d_s\}
\{Z(D_i^{(s)})-\hat d_s\}^\intercal
\right]\hat\lambda_2
\\
&=
\kappa\,\hat\lambda_2^\intercal
\widehat{\operatorname{Var}}_s\{Z(D)\}
\hat\lambda_2
\\
&=
\kappa\sum_j\hat v_j
(\hat u_j^\intercal\hat\lambda_2)^2
\\
&\geq
\kappa\,\lambda_{\min}\left(
\widehat{\operatorname{Var}}_s\{Z(D)\}
\right)
\sum_j(\hat u_j^\intercal\hat\lambda_2)^2
\\
&=
\kappa\,\lambda_{\min}\left(
\widehat{\operatorname{Var}}_s\{Z(D)\}
\right)\|\hat\lambda_2\|_2^2,
\end{align*}
where $\hat v_j$ and $\hat u_j$ denote the eigenvalues and
corresponding orthonormal eigenvectors of
$\widehat{\operatorname{Var}}_s\{Z(D)\}$, so that
$\widehat{\operatorname{Var}}_s\{Z(D)\}=\sum_j\hat v_j\hat u_j\hat u_j^\intercal.$ 

Applying Cauchy-Schwarz to the LHS gives $\hat\lambda_2^\intercal(\hat d_t-\hat d_s) \leq \|\hat\lambda_2\|_2\,\|\hat d_t-\hat d_s\|_2.$
Combining these inequalities and dividing by
$\|\hat\lambda_2\|_2$ when $\hat\lambda_2\neq0$ yields
$$
\kappa\lambda_{\min}\left(
\widehat{\operatorname{Var}}_s\{Z(D)\}
\right)\|\hat\lambda_2\|_2
\leq
\|\hat d_t-\hat d_s\|_2.
$$
Note that this inequality also holds when $\hat\lambda_2=0$. 
By Lemma~\ref{lemma: distr_clt}, for every $\varepsilon>0$,
$$
P\left(
\left|
\lambda_{\min}\left(
\widehat{\operatorname{Var}}_s\{Z(D)\}
\right)
-
\lambda_{\min}\left(
\operatorname{Var}_t\{Z(D)\}
\right)
\right|<\varepsilon
\right)\rightarrow1.
$$
Taking
$\varepsilon=\frac12\lambda_{\min}
(\operatorname{Var}_t\{Z(D)\})>0$,
we obtain, with probability tending to one,
$$
\lambda_{\min}\left(
\widehat{\operatorname{Var}}_s\{Z(D)\}
\right)
>
\lambda_{\min}\left(
\operatorname{Var}_t\{Z(D)\}
\right)-\varepsilon
=
\frac12\lambda_{\min}\left(
\operatorname{Var}_t\{Z(D)\}
\right).
$$
Therefore, we have 
\begin{align}\label{eqn: lambda_rate}
    \|\hat\lambda_2\|_2 \leq
\frac{\|\hat d_t-\hat d_s\|_2}{\kappa\lambda_{\min}\left(
\widehat{\operatorname{Var}}_s\{Z(D)\}
\right)} \leq \frac{2\|\hat d_t-\hat d_s\|_2}{
\kappa\,\lambda_{\min}\left(
\operatorname{Var}_t\{Z(D)\}
\right)}=O_p(1)O_p(m^{-1/2}),
\end{align}
where we use that $\hat{d}_t-\hat d_s=O_p(m^{-1/2})$. by Lemma \ref{lemma: distr_clt}, so  
\begin{equation}\label{eqn: lamdba_op1}
    \hat\lambda_2 = O_p(m^{-1/2})
\end{equation}

We now linearize $\hat \lambda_2$. Fix $0<\rho<\eta/2$. The exponential moment assumption gives $E_t\left[
\left(\exp\{\rho\|Z(D)\|_2\}\right)^2
\right] = E_t\left[\exp\{2\rho\|Z(D)\|_2\}\right] \leq E_t\left[\exp\{\eta\|Z(D)\|_2\}\right] <\infty.$ Thus, $\exp\{\rho\|Z(D)\|_2\}$ is square-integrable, and by Lemma~\ref{lemma: distr_clt},
$
\frac1{n_s}\sum_{i=1}^{n_s}
\exp\{\rho\|Z(D_i^{(s)})\|_2\}
\overset p\rightarrow
E_t\left[\exp\{\rho\|Z(D)\|_2\}\right].$
Since $n_s=O(m)$,
$$
\exp\left\{\rho\max_i\|Z(D_i^{(s)})\|_2\right\}
\le
\sum_{i=1}^{n_s}
\exp\{\rho\|Z(D_i^{(s)})\|_2\}
=
O_p(m).
$$
Taking logs gives
$\max_i\|Z(D_i^{(s)})\|_2=O_p(\log m)$. Therefore,
\begin{align}\label{eqn: max_lambda_Z}
\begin{split}
\max_i
\left|
\hat\lambda_2^\intercal
\{Z(D_i^{(s)})-\hat d_s\}
\right|
&\leq
\|\hat\lambda_2\|_2
\left\{
\max_i\|Z(D_i^{(s)})\|_2+\|\hat d_s\|_2
\right\}
\\
&=O_p(m^{-1/2})\left(O_p(\log m) + O_p(1)\right) \qquad \text{(Eqn \eqref{eqn: lamdba_op1})}\\
&=O_p\left(\frac{\log m}{\sqrt m}\right)\\
&=o_p(1).
\end{split}
\end{align}

The weights satisfy $\sum_i\hat w_iZ(D_i^{(s)})=\hat d_t$ and $\sum_i\hat w_i=1$. Therefore,
\[
\sum_{i=1}^{n_s}
\hat w_i\{Z(D_i^{(s)})-\hat d_s\}
=
\sum_{i=1}^{n_s}\hat w_iZ(D_i^{(s)})
-
\hat d_s\sum_{i=1}^{n_s}\hat w_i
=
\hat d_t-\hat d_s.
\]
Substituting the weight identity of Eqn.~\eqref{eqn: weights_eb} into the above gives
\begin{equation}\label{eqn: eb_summary_equation}
\frac{
\frac1{n_s}\sum_{i=1}^{n_s}
\exp\left(
\hat\lambda_2^\intercal\{Z(D_i^{(s)})-\hat d_s\}
\right)
\{Z(D_i^{(s)})-\hat d_s\}
}{
\frac1{n_s}\sum_{i=1}^{n_s}
\exp\left(
\hat\lambda_2^\intercal\{Z(D_i^{(s)})-\hat d_s\}
\right)
}
=
\hat d_t-\hat d_s.
\end{equation}

We expand the numerator and denominator separately. Since $\max_i | \hat\lambda_2^\intercal \{Z(D_i^{(s)})-\hat d_s\}|=o_p(1)$ with probability tending to one, the exponents have absolute value at most one. On this event, applying $|e^x-1-x|\le\frac12e^{|x|}x^2\le\frac e2x^2$ with $x=\hat\lambda_2^\intercal\{Z(D_i^{(s)})-\hat d_s\}$ gives, for every $i$,
\begin{align}\label{eqn: taylor_exp}
\begin{split}
\left|
\exp\left(
\hat\lambda_2^\intercal\{Z(D_i^{(s)})-\hat d_s\}
\right)
-1
-\hat\lambda_2^\intercal\{Z(D_i^{(s)})-\hat d_s\}
\right|
&\leq
\frac e2
\left[
\hat\lambda_2^\intercal\{Z(D_i^{(s)})-\hat d_s\}
\right]^2
\\
&\leq
\frac e2
\|\hat\lambda_2\|_2^2
\|Z(D_i^{(s)})-\hat d_s\|_2^2,
\end{split}
\end{align}
where we get the last inequality by Cauchy-Schwarz.

Evaluating the denominator,
\begin{align*}
\frac1{n_s}\sum_{i=1}^{n_s}
\exp\left(
\hat\lambda_2^\intercal\{Z(D_i^{(s)})-\hat d_s\}
\right)
&=
\frac1{n_s}\sum_{i=1}^{n_s}
\left[
1+\hat\lambda_2^\intercal\{Z(D_i^{(s)})-\hat d_s\}
+r_{i,m}
\right]
\\
&=
1+
\hat\lambda_2^\intercal
\underbrace{
\left[
\frac1{n_s}\sum_{i=1}^{n_s}
\{Z(D_i^{(s)})-\hat d_s\}
\right]
}_{=\,0}
+
\frac1{n_s}\sum_{i=1}^{n_s}r_{i,m}
\\
&=1+\frac1{n_s}\sum_{i=1}^{n_s}r_{i,m},
\end{align*}
where $r_{i,m}=\exp\left(\hat\lambda_2^\intercal\{Z(D_i^{(s)})-\hat d_s\}
\right)-1-\hat\lambda_2^\intercal\{Z(D_i^{(s)})-\hat d_s\}.$  Furthermore, by assumption, $\|Z(D)\|_2^2$ is square-integrable so applying Lemma \ref{lemma: distr_clt} with the inequality $\|a-b\|_2^2 \le 2\|a\|_2^2 + 2\|b\|_2^2$ gives
\begin{equation}\label{eqn: z-ds}
   \frac1{n_s}\sum_{i=1}^{n_s}
\|Z(D_i^{(s)})-\hat d_s\|_2^2
\le
2\left\{
\frac1{n_s}\sum_{i=1}^{n_s}\|Z(D_i^{(s)})\|_2^2
+
\|\hat d_s\|_2^2
\right\}
=
O_p(1). 
\end{equation}

Applying the result of \eqref{eqn: taylor_exp}, we have that 
\begin{align*}
\left|\frac1{n_s}\sum_{i=1}^{n_s}r_{i,m}
\right| &\leq
\frac1{n_s}\sum_{i=1}^{n_s}|r_{i,m}|\\
&\leq
\frac{e}{2n_s}\sum_{i=1}^{n_s}
\left[
\hat\lambda_2^\intercal
\{Z(D_i^{(s)})-\hat d_s\}
\right]^2
\\
&\leq
\frac e2\|\hat\lambda_2\|_2^2
\frac1{n_s}\sum_{i=1}^{n_s}
\|Z(D_i^{(s)})-\hat d_s\|_2^2 \qquad \text{(Cauchy-Schwarz)}
\\
&=
O_p(m^{-1})O_p(1) \qquad \text{(Eqns \eqref{eqn: lamdba_op1} and \eqref{eqn: z-ds})}\\
&=
O_p(m^{-1}).
\end{align*}
Therefore, the denominator is
\begin{equation}\label{eqn: denom_lambda}
    \frac1{n_s}\sum_{i=1}^{n_s}
\exp\left(
\hat\lambda_2^\intercal\{Z(D_i^{(s)})-\hat d_s\}
\right)
=
1+O_p(m^{-1}).
\end{equation}

For the numerator, we have
\begin{align*}
&\frac1{n_s}\sum_{i=1}^{n_s}
\exp\left(
\hat\lambda_2^\intercal\{Z(D_i^{(s)})-\hat d_s\}
\right)
\{Z(D_i^{(s)})-\hat d_s\}\\
&=
\underbrace{\frac1{n_s}\sum_{i=1}^{n_s}
\{Z(D_i^{(s)})-\hat d_s\}}_{=0}
+
\left[
\frac1{n_s}\sum_{i=1}^{n_s}
\{Z(D_i^{(s)})-\hat d_s\}
\{Z(D_i^{(s)})-\hat d_s\}^\intercal
\right]\hat\lambda_2
+\frac1{n_s}\sum_{i=1}^{n_s}\tilde{r}_{i,m},
\end{align*}
where $\tilde{r}_{i,m}=r_{i,m}\{Z(D_i^{(s)})-\hat d_s\}$. Applying the above argument of the remainder term, $r_{i,m}$, we get
\begin{align*}
\left\|\frac1{n_s}\sum_{i=1}^{n_s}\tilde r_{i,m}
\right\|_2 \le \frac e2\|\hat\lambda_2\|_2^2 \frac1{n_s}\sum_{i=1}^{n_s}
\|Z(D_i^{(s)})-\hat d_s\|_2^3=
O_p(m^{-1}).
\end{align*}

Therefore, we get that the numerator is
\begin{equation}\label{eqn: num_lambda}
    \frac1{n_s}\sum_{i=1}^{n_s}
\exp\left(
\hat\lambda_2^\intercal\{Z(D_i^{(s)})-\hat d_s\}
\right)
\{Z(D_i^{(s)})-\hat d_s\}
=
\widehat{\operatorname{Var}}_s\{Z(D)\}\hat\lambda_2
+
O_p(m^{-1}).
\end{equation}

Substituting Eqns \eqref{eqn: denom_lambda} and \eqref{eqn: num_lambda} into Eqn.~\eqref{eqn: eb_summary_equation} gives
\begin{align*}
\hat d_t-\hat d_s =\frac{\widehat{\operatorname{Var}}_s\{Z(D)\}\hat\lambda_2
+O_p(m^{-1})}{1+O_p(m^{-1})} =
\widehat{\operatorname{Var}}_s\{Z(D)\}\hat\lambda_2
+O_p(m^{-1}),
\end{align*}
where the last equality uses that $\hat\lambda_2=O_p(m^{-1/2})$. Rearranging gives us 
\begin{equation}\label{eqn: lambda_eb}
\hat\lambda_2
=
\left[
\widehat{\operatorname{Var}}_s\{Z(D)\}
\right]^{-1}
(\hat d_t-\hat d_s)
+
o_p(m^{-1/2}).
\end{equation}
\end{proof}

\begin{proposition}[Consistency of $\hat\theta_{EB}$]\label{prop: eb_consist}
Suppose that $\max_i|n_s\hat w_i-1|=o_p(1)$ and the assumptions of Lemma \ref{lemma: entropy_bal},
$$\hat{\theta}_{EB} \overset{p}{\rightarrow} \theta_t$$
\end{proposition}
\begin{proof}
We first write the difference in the weighted loss under entropy-balancing weights and the source average loss as
\[
\sum_{i=1}^{n_s}\hat w_i\ell(\theta;D_i^{(s)})
-
\frac1{n_s}\sum_{i=1}^{n_s}\ell(\theta;D_i^{(s)})
=
\frac1{n_s}\sum_{i=1}^{n_s}
(n_s\hat w_i-1)\ell(\theta;D_i^{(s)}).
\]
Taking absolute values and then the supremum over
$\Theta$ gives
\begin{align*}
\sup_{\theta\in\Theta}
\left|
\sum_{i=1}^{n_s}\hat w_i\ell(\theta;D_i^{(s)})
-
\frac1{n_s}\sum_{i=1}^{n_s}\ell(\theta;D_i^{(s)})
\right|
&\leq
\max_i|n_s\hat w_i-1|
\frac1{n_s}\sum_{i=1}^{n_s}
\sup_{\theta\in\Theta}|\ell(\theta;D_i^{(s)})|
\\
&=
o_p(1)O_p(1)\\
&=o_p(1),
\end{align*}
where we apply the assumption that $\max_i|n_s\hat w_i-1|=o_p(1)$ and Assumption \eqref{assump: lip_loss} for the second equality.

By compactness of $\Theta$, the envelope and Lipschitz
conditions in Assumption~\eqref{assump: lip_loss}, and
Lemma~\ref{lemma: distr_clt}, a finite-net argument gives
\[
\sup_{\theta\in\Theta}
\left|
\frac1{n_s}\sum_{i=1}^{n_s}\ell(\theta;D_i^{(s)})
-
E_t[\ell(\theta;D)]
\right|
=o_p(1).
\]
Combining this with the preceding bound yields
\[
\sup_{\theta\in\Theta}
\left|
\sum_{i=1}^{n_s}\hat w_i\ell(\theta;D_i^{(s)})
-
E_t[\ell(\theta;D)]
\right|
=o_p(1).
\]

Since $\theta\mapsto E_t[\ell(\theta;D)]$ is continuous
on the compact parameter space $\Theta$ and has the
unique minimizer $\theta_t$, for every $\varepsilon>0$, $
\inf_{\substack{\theta\in\Theta\\ \|\theta-\theta_t\|_2\ge\varepsilon}} E_t[\ell(\theta;D)] > E_t[\ell(\theta_t;D)].
$
Together with the uniform convergence,
this verifies the conditions of
\citet[Theorem~5.7]{vanderVaart1998asymptotic}. Therefore, any global minimizer
$\hat\theta_{\mathrm{EB}}$ of
$\sum_{i=1}^{n_s}\hat w_i\ell(\theta;D_i^{(s)})$
satisfies
\begin{equation}
\hat\theta_{\mathrm{EB}}
\overset p\rightarrow\theta_t.
\end{equation}
\end{proof}

\section{Additional empirical studies} \label{app: empirical}
\subsection{Full NHANES experiment} \label{app: nhanes}

\begin{figure}[H]
    \centering
    \includegraphics[width=\linewidth]{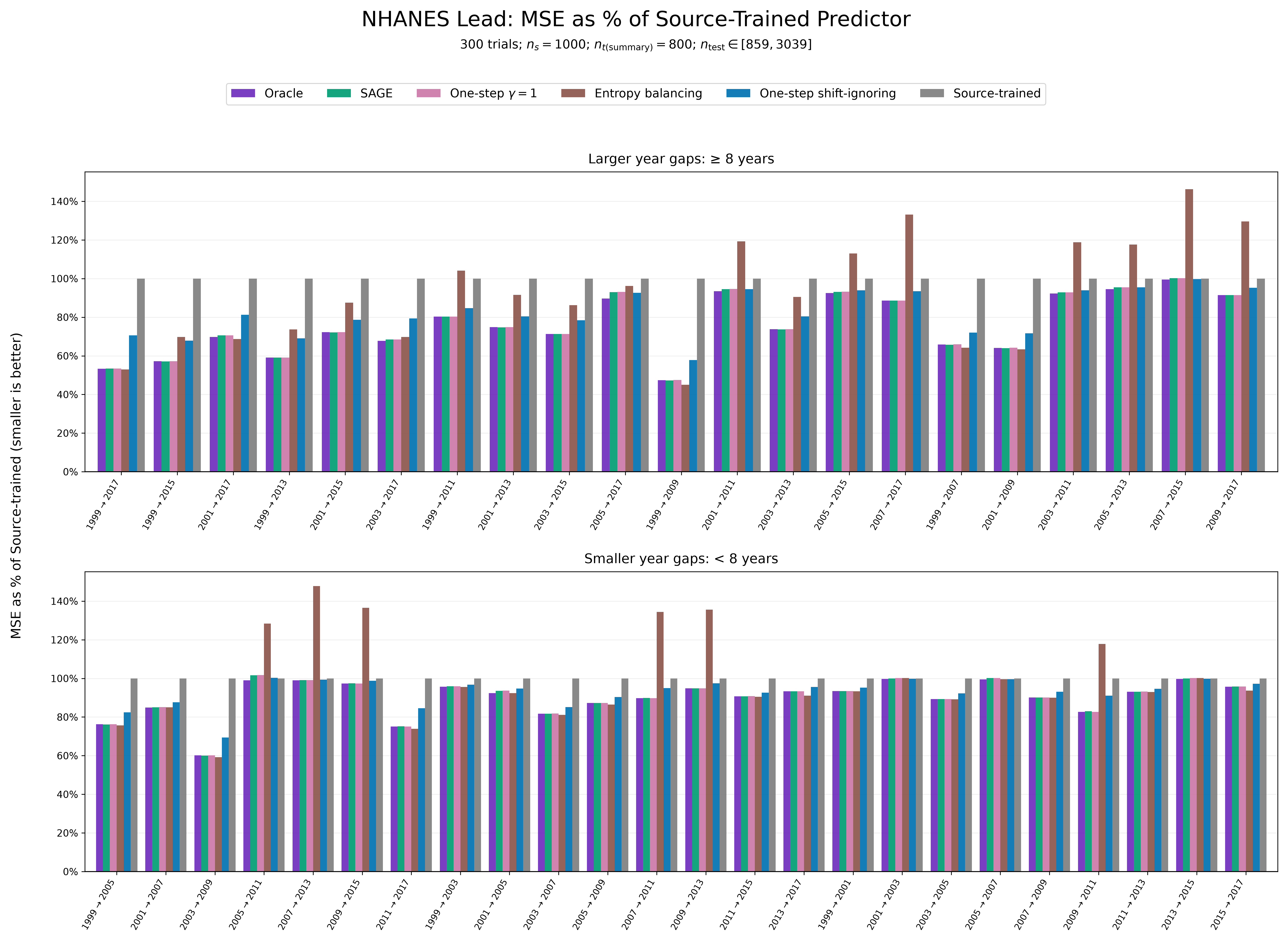}
        \captionsetup{width=\linewidth}
    \caption{
NHANES lead prediction results. Bars show the percentage target test MSE relative to the source-trained predictor, averaged over 300 trials. SAGE improves prediction for most year pairs, with substantial gains for larger gaps (top panel), and generally outperforms the one-step shift-ignoring estimator. Its estimated step-sizes are close to one, yielding performance similar to the full-step  update (SAGE $\gamma=1$) and approaches the oracle benchmark.
}
    \label{fig: nhanes_lead_old}
\end{figure}

\subsection{Diagnostic Plots} \label{app: qq}

\begin{figure}[H]
        \centering
        \includegraphics[width=\linewidth]{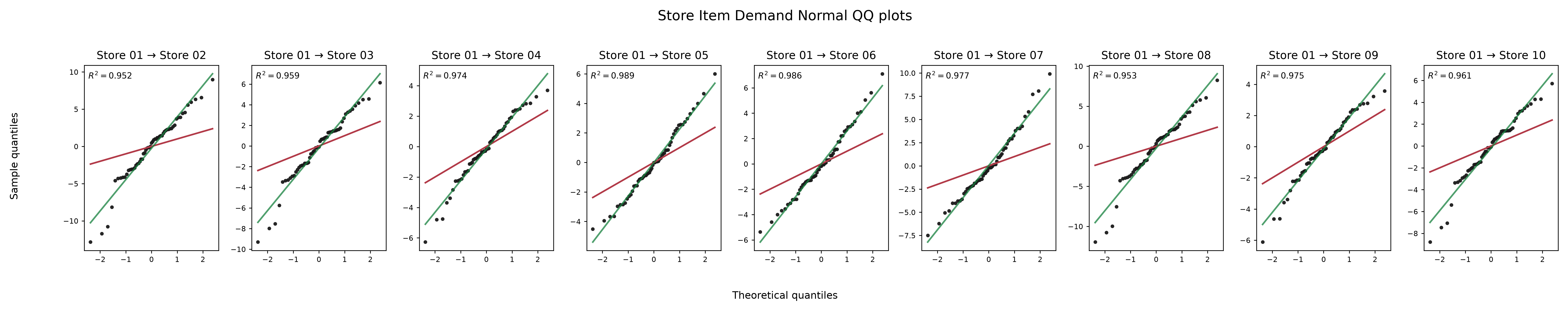}
        \caption{Normal QQ plots for the Store Item Demand dataset with source$\rightarrow$target pairs.  Each point corresponds to the standardized mean difference in Eqn~\eqref{eqn: qq_stat} for one whitened coordinate. The red line is the standard Gaussian reference line, corresponding to i.i.d.\ sampling from the same distribution. The green line represents a Gaussian pattern with inflated variance. }
        \label{fig: qqplot_store}
    \end{figure}
    
\begin{figure}[H]
        \centering
    \includegraphics[width=\linewidth]{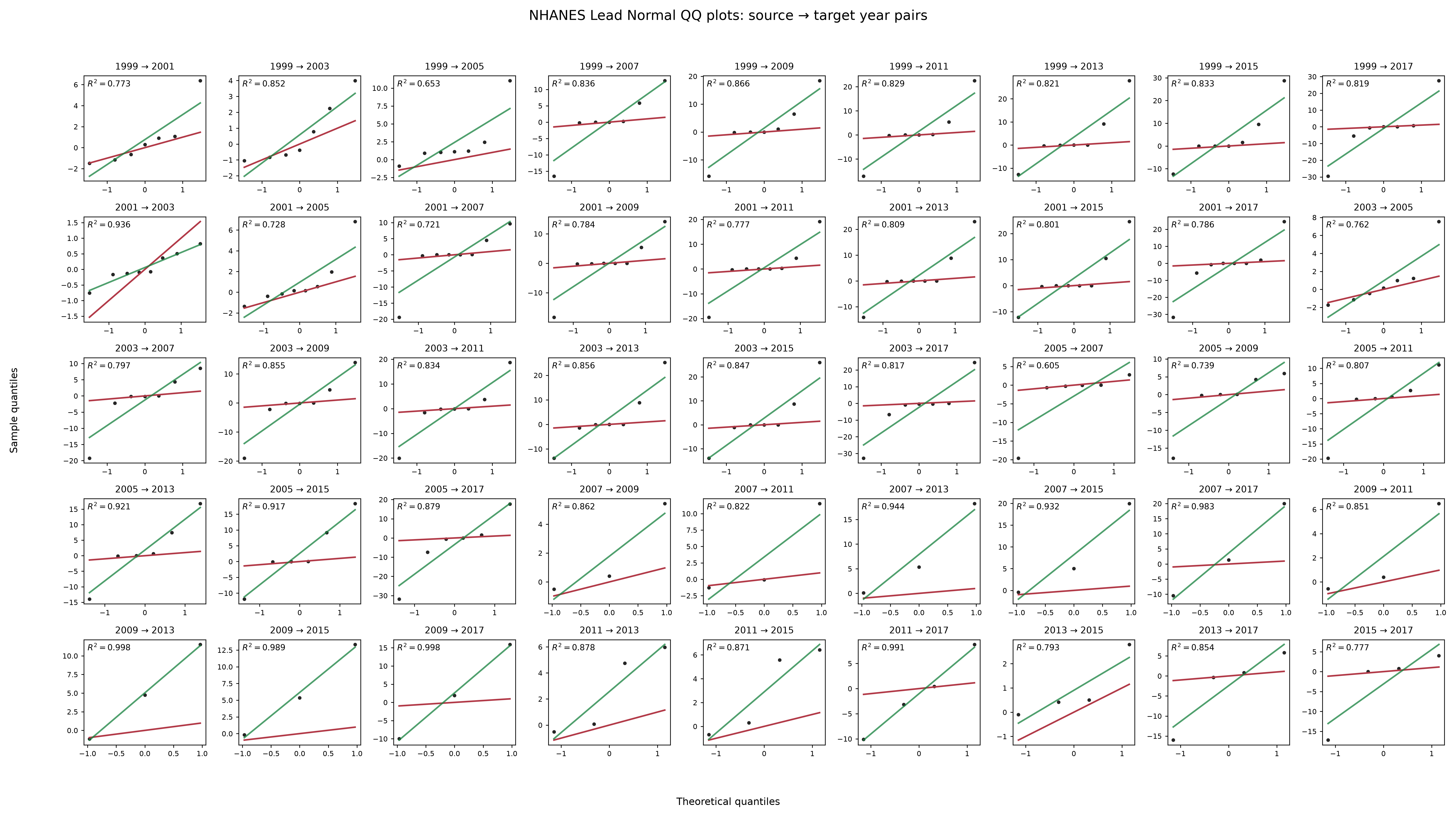}
  \captionsetup{width=\linewidth}
    \caption{Normal-QQ plots for NHANES lead data with year pairs source$\rightarrow$target.  Each point corresponds to the standardized mean difference in Eqn~\eqref{eqn: qq_stat} for one whitened coordinate. The red line is the standard Gaussian reference line, corresponding to i.i.d.\ sampling from the same distribution. The green line represents a Gaussian pattern with inflated variance. }
    \label{fig: qqplot_nhanes_lead}
\end{figure}

\begin{figure}[H]
        \centering
        \includegraphics[width=\linewidth]{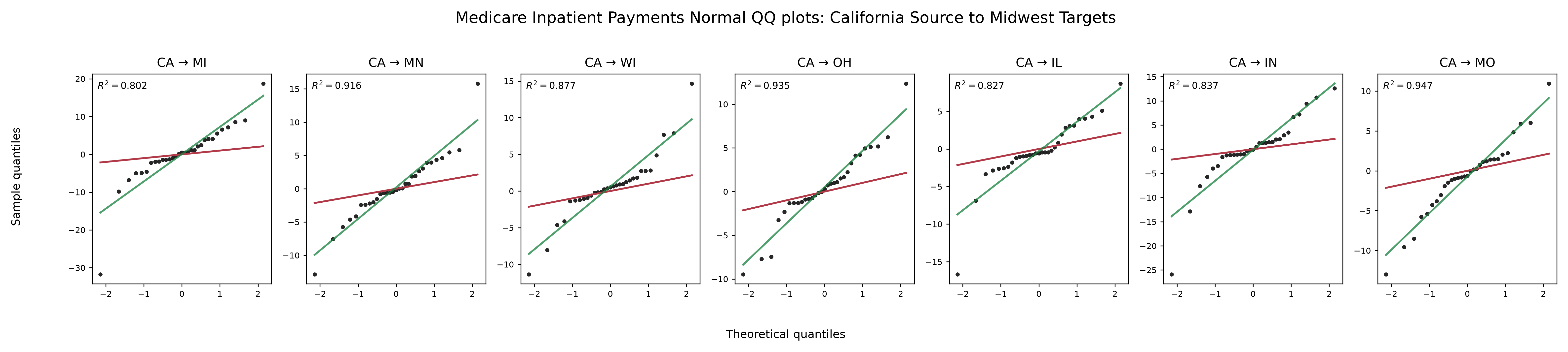}
        \caption{Normal-QQ plots for Medicare inpatient payment data with source$\rightarrow$target pairs.  Each point corresponds to the standardized mean difference in Eqn~\eqref{eqn: qq_stat} for one whitened coordinate. The red line is the standard Gaussian reference line, corresponding to i.i.d.\ sampling from the same distribution. The green line represents a Gaussian pattern with inflated variance. }
        \label{fig: qqplot_medicare}
    \end{figure}

\section{SAGE extension algorithms}

\begin{algorithm}[H]
\caption{SAGE with subgroup sample variances}
\label{alg: sage var}
\begin{algorithmic}[1]

\Statex \emph{Replace lines 6 and 8 of
Algorithm~\ref{alg: grad_linear} with the steps below. After augmentation, apply the coordinate reduction of
Assumption~\ref{assump: var(z) pd} using the same retained coordinates for source and target, and fix coefficients
of omitted features at zero in the minimization below.}
\Require Target variances
$\{\widehat{\mathrm{Var}}_t[X\mid A_g=1],
\widehat{\mathrm{Var}}_t[y\mid A_g=1]\}_{g:n_t(g)>1}$
\Statex Use variances with denominator $n_t(g)$:
for $n_t(g)>1$, multiply variances reported with denominator
$n_t(g)-1$ by $\{n_t(g)-1\}/n_t(g)$;
for $n_t(g)=1$, set the variances to zero.

\Require Elementwise target second moments for nonempty subgroups:
\Statex \hspace{0.6cm}
$\hat{E}_t[X^{\circ 2}\mid A_g=1]
= \widehat{\mathrm{Var}}_t[X\mid A_g=1]
+ \hat{E}_t[X\mid A_g=1]^{\circ 2}$
\Statex \hspace{0.6cm}
$\hat{E}_t[y^2\mid A_g=1]
= \widehat{\mathrm{Var}}_t[y\mid A_g=1]
+ \hat{E}_t[y\mid A_g=1]^2$
\Require Elementwise source second moments for nonempty subgroups:
\Statex \hspace{0.6cm}
$\hat{E}_s[X^{\circ 2}\mid A_g=1]
= \frac{\sum_{i=1}^{n_s}A_{ig}^{(s)}(X_i^{(s)})^{\circ 2}}{\sum_{i=1}^{n_s}A_{ig}^{(s)}}$,
\quad
$\hat{E}_s[y^2\mid A_g=1]
= \frac{\sum_{i=1}^{n_s}A_{ig}^{(s)}(y_i^{(s)})^2}{\sum_{i=1}^{n_s}A_{ig}^{(s)}}$
\Require Additional second moment summary differences:
\Statex \hspace{0.6cm}
$\Delta_{X^2}(g)
\gets
\hat\pi_{t,g}\hat E_t[X^{\circ 2}\mid A_g=1]
-
\hat\pi_{s,g}\hat E_s[X^{\circ 2}\mid A_g=1]$
\Statex \hspace{0.6cm}
$\Delta_{y^2}(g)
\gets
\hat\pi_{t,g}\hat E_t[y^2\mid A_g=1]
-
\hat\pi_{s,g}\hat E_s[y^2\mid A_g=1]$

\setcounter{ALG@line}{5}

\State Fit on source using all reported subgroups:
\Statex 
$(\hat\alpha^{(s)},
\{\hat\beta_{g,0}^{(s)},\hat\beta_{g,X}^{(s)},\hat\beta_{g,y}^{(s)},\hat\beta_{g,X^2}^{(s)},\hat\beta_{g,y^2}^{(s)}\}_{g=1}^K)
\in $
\Statex \hspace{0.6cm}
$
\arg\min_{\alpha,\{\beta_g\}_{g=1}^K}
\sum_{i=1}^{n_s}
\left\|
u_i
-
\alpha
-
\sum_{g=1}^K
A_{ig}^{(s)}
\left(
\beta_{g,0}
+
\beta_{g,X}^\intercal X_i^{(s)}
+
\beta_{g,y}y_i^{(s)}
+
\beta_{g,X^2}^\intercal (X^{\circ 2})_i^{(s)}
+
\beta_{g,y^2}(y^2)_i^{(s)}
\right)
\right\|_2^2$

\For{$g=1,\dots,K$}
    \State Estimate subgroup contribution using target summaries:
    \Statex \hspace{0.6cm}
    $\widehat{\Psi}_t(g)
    \gets
    \hat\beta_{g,0}^{(s)}\Delta_0(g)
    +
    \bigg(\hat\beta_{g,X}^{(s)}\bigg)^\intercal \Delta_X(g)
    +
    \hat\beta_{g,y}^{(s)}\Delta_y(g)
    +
    \bigg(\hat\beta_{g,{X^2}}^{(s)}\bigg)^\intercal \Delta_{X^2}(g)
    +
    \hat\beta_{g,{y^2}}^{(s)}\Delta_{y^2}(g)$
\EndFor

\end{algorithmic}
\end{algorithm}

\begin{algorithm}[H]
\caption{SAGE-NN: Summary adapted final layer neural network adaptation}
\label{alg: sage_nn}
\begin{algorithmic}[1]

\Require Source observations
$\{(X_i^{(s)},y_i^{(s)},A_{i1}^{(s)},\ldots,A_{iK}^{(s)})\}_{i=1}^{n_s}$
\Require Target subgroup counts and moments
$\{n_t(g),\widehat E_t[X\mid A_g=1],
\widehat E_t[y\mid A_g=1]\}_{g=1}^K$
\Require Target summary sample size $n_t$
\Require MLP encoder $\phi_\eta(X)$ with parameters $\eta$ and a linear output layer
\Require step-size $\hat\gamma$ estimated per Algorithm \ref{alg: sage_gamma}

\State Define the source MLP predictor before hidden feature standardization:
\Statex \hspace{0.6cm}
$f_{\eta,\omega}^{\mathrm{raw}}(X)
\gets
\omega^\intercal
\begin{pmatrix}
1^\intercal &
\phi_\eta(X)^\intercal
\end{pmatrix}^\intercal$

\State Fit the complete MLP using source observations:
\Statex \hspace{0.6cm}
$(\widehat\eta_s,\widehat\omega_s)
\gets
\arg\min_{\eta,\omega}
\frac{1}{n_s}\sum_{i=1}^{n_s}
\ell\left(
f_{\eta,\omega}^{\mathrm{raw}}(X_i^{(s)}),
y_i^{(s)}
\right)$

\State Freeze the source-trained encoder parameters $\widehat\eta_s$.

\State Compute the source mean and coordinatewise standard deviations of the
final hidden-layer representation:
\Statex \hspace{0.6cm}
$\bar\phi_s
\gets
\frac{1}{n_s}\sum_{i=1}^{n_s}
\phi_{\widehat\eta_s}(X_i^{(s)}),
\qquad
D_s
\gets
\operatorname{diag}
(\widehat\sigma_{s,1},\ldots,\widehat\sigma_{s,q})$

\State Define the standardized hidden representation and add an intercept:
\Statex \hspace{0.6cm}
$z_{\widehat\eta_s}(X)
\gets
D_s^{-1}
\left\{
\phi_{\widehat\eta_s}(X)-\bar\phi_s
\right\},
\qquad
\widetilde z_{\widehat\eta_s}(X)
\gets
\begin{pmatrix}
1^\intercal &
z_{\widehat\eta_s}(X)^\intercal
\end{pmatrix}^\intercal$

\State Refit the final linear layer by source least squares:
\Statex \hspace{0.6cm}
$\widehat\theta_{\mathrm{init}}
\gets
\arg\min_{\theta}
\frac{1}{2n_s}\sum_{i=1}^{n_s}
\left[
y_i^{(s)}
-
\theta^\intercal
\widetilde z_{\widehat\eta_s}(X_i^{(s)})
\right]^2$

\State Compute the source final layer gradients:
\Statex \hspace{0.6cm}
$u_i
\gets
\left[
(\widehat\theta_{\mathrm{init}})^\intercal
\widetilde z_{\widehat\eta_s}(X_i^{(s)})
-y_i^{(s)}
\right]
\widetilde z_{\widehat\eta_s}(X_i^{(s)}),
\quad i=1,\ldots,n_s$

\State Compute the source final layer Hessian:
\Statex \hspace{0.6cm}
$\widehat H_s
\gets
\frac{1}{n_s}\sum_{i=1}^{n_s}
\widetilde z_{\widehat\eta_s}(X_i^{(s)})
\widetilde z_{\widehat\eta_s}(X_i^{(s)})^\intercal$

\State Using $\{u_i\}_{i=1}^{n_s}$ as the multivariate responses, apply the
summary gradient reconstruction steps of Algorithm~\ref{alg: grad_linear}
with the source variables
$\{A_{ig}^{(s)},A_{ig}^{(s)}X_i^{(s)},A_{ig}^{(s)}y_i^{(s)}\}$
and the corresponding target subgroup summaries.

\State Let $\widehat\Psi_t$ denote the resulting summary adapted estimate of
the target final layer gradient.

\State Update only the final linear layer:
\Statex \hspace{0.6cm}
$\widehat\theta_{\mathrm{SAGE\text{-}NN}}
\gets
\widehat\theta_{\mathrm{init}}
-
\hat\gamma \widehat H_s^{-1}\widehat\Psi_t$

\State \Return
$f_{\widehat\eta_s,\widehat\theta_{\mathrm{SAGE\text{-}NN}}}(X)
=
\widehat\theta_{\mathrm{SAGE\text{-}NN}}^\intercal
\widetilde z_{\widehat\eta_s}(X)$

\end{algorithmic}
\end{algorithm}

\bibliographystyle{apacite}
\bibliography{Bibliography}

\end{document}